\pdfoutput=1

\documentclass[11pt,letterpaper]{article}

\usepackage{mystyle}

\usepackage{lineno}
\newcommand{\eda}{\xi} %empirical distribution over actions previously denoted by z
\newcommand{\reg}{\mcal{R}} %CPT regret
\newcommand{\hyp}{\mcal{H}} % hyperplane
\newcommand{\conv}{{co}} % closed convex hull
\newcommand{\strat}{\mcal{S}} % strategy for repeated game
\newcommand{\av}{\mcal{A}} % average CPT value
\newcommand{\cob}{\color{black}}

\newcommand{\black}{\color{black}}

\title{Learning in Games with Cumulative Prospect Theoretic Preferences}

\author{Soham R.\ Phade%
\thanks{Corresponding author}
\ and Venkat Anantharam% <-this % stops a space
\thanks{The authors are with the Department of Electrical Engineering and Computer Sciences, University of California, Berkeley, Berkeley, CA 94720.
        {\tt\small soham\_phade@berkeley.edu, ananth@eecs.berkeley.edu}}%
\thanks{Research supported by the NSF Science and Technology Center grant CCF-
0939370: "Science of Information", the NSF grants ECCS-1343398, CNS-1527846, CIF-1618145 and
CCF-1901004, and the William and Flora Hewlett Foundation supported Center for Long Term Cybersecurity at Berkeley.}% <-this % stops a space
}
\date{}
\begin{document}

\maketitle

%\tableofcontents
%\listoffigures
%\listoftables
%!TEX root = ../main.tex

\begin{abstract}
	We consider repeated games 
	where 
	%\red 
	the 
	%\black
	players 
	behave according to cumulative prospect theory (CPT). 
	We show that, 
	%\red 
	when 
the
players have calibrated strategies and behave according to CPT,
	%a natural analog 
	the natural analog 
	%\black
	of the notion of correlated equilibrium in the CPT case,
	as defined by Keskin, 
	is not enough to 
%guarantee the convergence of 
%\bb
capture all subsequential limits of
%\cob
the empirical distribution of action play.
We define the notion of 
a mediated CPT correlated equilibrium 
via an extension 
%of the game 
%\color{red} 
of the stage game 
%\color{black}
to a so-called mediated game. We then show, along the lines of 
%\red 
the result of Foster and Vohra about convergence to the set of correlated equilibria when the players behave according to expected utility theory that, in the CPT case, 
%\black 
under calibrated learning the empirical distribution 
%\bb
of action play 
%\cob
converges to the set of all mediated CPT correlated equilibria.
%\footnote{\red Note: introduced a new paragraph \black}

We also show that, in general, the set of CPT correlated equilibria is not approachable in the Blackwell approachability sense. 
%\com
We observe that 
%mediated games are 
%\color{red} 
a mediated game is 
%\color{black}
a specific type of a 
%game with communication, 
%\red 
{\em game with communication,} 
%\black
as introduced by Myerson, 
and as a consequence we get that the revelation principle does not hold under CPT.
%\cob
%\keywords{cumulative prospect theory; game theory; repeated games; calibrated learning; correlated equilibrium; no-regret learning}
\end{abstract}


%!TEX root = ../main.tex

\section{Introduction}
\label{sec: intro}

In non-cooperative game theory, a finite $n$-person game models a social system comprised of several \emph{decision makers} (or \emph{players}) with possibly different objectives, interacting 
%in a certain type of environment. 
%\bb
in some environment.
%\magenta Comment: In the preceding sentence what is 
%a ``finite $n$-person game" as opposed to just an
%``$n$-person game"? This terminology also shows up
%later, so if it needs to be fixed, fix all occurrences of it.\black
%\cob
%The notion of equilibrium is 
%\red 
Notions of equilibrium are 
%\black
central to game theory.
%\cor 
The neoclassical economics viewpoint of game theory 
%\cob 
attempts to explain 
%\red 
an 
%\black 
equilibrium as a self-evident outcome of the optimal behavior of the participating players, assuming them to be rational.
Two of the most well known notions of equilibrium for a finite $n$-person game are Nash equilibrium \citep{nash1951non} and correlated equilibrium \citep{aumann1974subjectivity}.
(See \citet{kreps1990game} for an excellent account of the strengths and weaknesses of these notions.)
An alternate approach, called 
%learning in games 
%\bb
{\em learning in games},
%\cob
is concerned with 
%dynamic considerations such as repeated games 
%\bb
justifying equilibrium behavior via a dynamic process
%\cob
where the players learn from the past play 
%and observations from the environment and adapt
and observations from the 
%\red 
environment, and 
%\black 
adapt
accordingly~\citep{,aumann1995repeated,fudenberg1998learning,young2004strategic}. 
In this paper, we will be concerned with this alternate approach.

Since decision makers are an integral part of any social system, their behavioral properties form an important aspect in modeling games. 
The study of game theory so far has been mainly based on the assumption that the behavior of the players towards their \emph{lottery} preferences 
(see Section~\ref{sec: prelim} for the definition of a lottery) 
can be modeled by \citet{von1945theory} \emph{expected utility theory} (EUT). 
EUT has a nice normative appeal to it, in particular when it comes to the \emph{independence axiom}, 
which basically says that if lottery $L_1$ is preferred over lottery $L_2$, and  $L$ is some other lottery, then, for $0 \leq \alpha \leq 1$, 
the combined lottery $\alpha L_1 + (1-\alpha)L$ is preferred over the combined lottery $\alpha L_2 + (1 - \alpha)L$. 
Even though this seems very intuitive, a systematic deviation from such behavior has been observed in multiple empirical studies 
%\bb
(for example, the \citet{allais1953extension} paradox). 
%\cob
This gave rise to the study of 
%non-expected utility theory that does away with 
%\red 
alternatives to EUT that do away with 
%\black
the independence axiom. 
\emph{Cumulative prospect theory} (CPT), as formulated by \citet{tversky1992advances}, is one such theory, 
which accommodates many of the empirically observed behavioral features without losing much tractability \citep{wakker2010prospect}. 
It is also a generalization of 
%the expected utility theory. 
%\red 
EUT. 
%\black
%in the sense that EUT is a special case of CPT.

It becomes even more important to consider non-EUT behavior in the theory of learning in games. For example, in a \emph{repeated game}, \citet{hart2005adaptive} argues that players tend to use simple procedures like \emph{regret} minimization. 
A player $i$ is said to have 
%a small regret
%\bb
no regret\footnote{also known as the internal regret or the conditional regret.}
if, for each pair of her actions $a_i,\tilde a_i$, 
she does not regret not having played action $\tilde a_i$ whenever she played action $a_i$. 
Such regrets can simply be computed as the difference in the average payoffs 
received by the player from playing action $\tilde a_i$ instead of action $a_i$, assuming the opponents stick to their actions.
While evaluating such regrets in the real world, however, 
%the players are prone to systematic deviations. 
%\bb
players who are modeled as evaluating lotteries according to CPT preferences
are likely to exhibit different kinds of learning behavior than that 
exhibited by EUT players.
%\cob
%The proposed model based on cumulative prospect theory is an attempt to handle these systematic deviations, arising due to several behavioral features exhibited by the participants, by allowing flexibility over the expected utility model through the \emph{probability weighting functions}. 
%\bb
The proposed model in this paper is an attempt to handle these systematic 
deviations in learning, anticipated from the
empirically observed behavioral features exhibited by human agents, as 
captured by 
%cumulative prospect theory.
%\red 
CPT. 
%\black
%\cob
We pose the following question: \emph{How do the predictions of the theory of learning in games change if the players behave according to CPT?}

The strategies of the players are said to be in a Nash equilibrium if no player is tempted to deviate from her strategy provided the strategies of the others remain unchanged. 
Suppose now that, before the game is played, there is a mediator who sends each player a private signal to play a certain action. Each player may then choose her action depending on this signal. 
A correlated equilibrium of the original game is obtained by taking the joint distribution over action profiles of all the players corresponding to a Nash equilibrium of the game with a mediator \citep{aumann1974subjectivity}. 
\citet{crawford1990equilibrium} studies games where players do not adhere to the independence axiom, and defines an analog for the Nash equilibrium. 
\citet{keskin2016equilibrium} defines analogs for both the notions of equilibrium, Nash and correlated, 
%based on CPT. 
%\red 
when the players have CPT preferences. 
%\black
We call them \emph{CPT Nash equilibrium} and \emph{CPT correlated equilibrium} respectively. 
In Section~\ref{sec: prelim}, we give a brief review of 
%cumulative prospect theory 
%\red 
CPT 
%\black
and Keskin's definitions for these equilibrium notions.
In the absence of the independence axiom, many of the linearities present in the model under EUT are lost. For example, the set of all correlated equilibria 
%\red 
for EUT players 
%\black
is a convex polytope \citep{aumann1987correlated}; however, the set of all CPT correlated equilibria need not be convex \citep{keskin2016equilibrium}. In fact, it can even be disconnected \citep{phade2019geometry}.

For a repeated game, \citet{foster1998asymptotic} describe a procedure based on \emph{calibrated learning} that guarantees the convergence of the \emph{empirical distribution} of action play to the set of correlated equilibria, when players behave according to EUT.
In Section~\ref{sec: calibrated_learning}, we formulate an analog for their procedure when players behave according to CPT. 
In Example~\ref{ex: nonconvergence_calib}, we describe a game for which the set of all CPT correlated equilibria is non-convex and we show that the empirical distribution 
%of play 
%\bb
of action play
%\cob
does not converge to this set.

%A classical relaxation in the game theory literature is to look at the convex hull of the set under consideration \citep{mannor2009online,mannor2014approachability}.
%{\color{red} The preceding sentence is too vague. \color{black}}
%We define 
%\bb
We then define
%\cob
an extension of the set of CPT correlated equilibria
%in section \ref{sec: calib_mediated}
and establish the convergence of the empirical distribution of action play to this extended set.
% in section \ref{sec: calib_mediated}.
It turns out that this extension has a nice 
game-theoretic
%\magenta Comment: Please look for other occurrences of ``game theoretic" and convert them to ``game-theoretic". \black 
interpretation, obtained by allowing the mediator 
to send any private 
%payoff-irrelevant 
signal (instead of restricting her to send a signal corresponding to some action). 
We formally define this setup in Section~\ref{sec: mediated_eq_CPT}, and call it a \emph{mediated game}.
\citet{myerson1986multistage} 
%considers 
%\bb
has considered
%\cob
a further generalization in which each player first reports her \emph{type} from a finite set $T_i$. 
The mediator collects the reports from all the players and then sends each one of them a private signal from a finite set $B_i$. 
The mediator is characterized by a rule $\psi : \prod_i T_i \to \Delta(\prod_i B_i)$ that maps each type profile to a probability distribution on the set of signal profiles from which it samples the private signals to be sent.
Based on her received signal, each player chooses her action. 
These are called \emph{games with communication}. 
The type sets $(T_i)_{i=1}^n$, the signal sets $(B_i)_{i=1}^n$, and the mediator rule $\psi$ together are said to comprise a \emph{communication system}. 
Under EUT, the set of all correlated equilibria of a game is characterized as the union, over all possible communication systems, 
of the sets of joint distributions on the action profiles of all players arising from all the Nash equilibria for the corresponding game with communication 
(for a detailed exposition see \citep{myerson2013game}).
%{\color{red} Please check the accuracy of the preceding sentence. need to rephrase to indicate that the set of joint distributions coming from nash eq of communication game.\color{black}} 
This is sometimes referred to as the \emph{Bayes-Nash revelation principle}, or simply the \emph{revelation principle}. Since a mediated game is a specific type of game with communication, characterized by players not reporting their type, 
%\red 
or equivalently by the mediator ignoring the types reported
by the players, 
%\black
our analysis shows that the revelation principle does not hold under CPT.
%\magenta Comment: The preceding discussion makes no sense. You seem to have defined a game with communication as one in which each player first reports his or her type. You then say that a mediated game is a special case of this because the players do not report their type. This makes no sense. \black

Calibrated learning is one way of studying learning in games. Some other approaches originate from Blackwell's \emph{approachability} theory and the regret-based framework of online learning (\cite{hart2000simple,fudenberg1995consistency}). 
In fact, \citet{foster1998asymptotic} establish the existence of calibrated learning schemes using such a regret-based framework and Blackwell's approachability theory. 
See \citet{perchet2009calibration} for a comparison between these approaches, and see also \citet{cesa2006prediction}. 
\citet{hannan1957approximation} introduced the concept of no-regret strategies in the context of repeated matrix games. 
No-regret learning in games is equivalent to the convergence of the empirical distribution 
%of play 
%\bb
of action play
%\cob
to the set of correlated equilibria \citep{hart2000simple,fudenberg1995consistency}.
We establish an analog of this result when players behave according to CPT.
We then ask if no-regret learning is possible under CPT. 

Blackwell's approachability theorem prescribes a strategy to steer the average payoff vector of a player 
%\color{red} 
in a game with vector payoffs 
%\color{black} 
towards a given target set, irrespective of the strategies of the other players. 
The theorem also gives a necessary and sufficient condition for the existence of such a strategy provided the target set is convex and the game environment remains fixed. 
Here, by game environment, we mean the rule by which the payoff vectors depend on the players' actions. 
Under EUT, \citet{hart2000simple} take these payoff vectors to be the regrets associated to a player and establish no-regret learning by showing that the 
%negative orthant 
%\red 
nonpositive orthant 
%\black
in the space of payoff vectors is approachable. 
Under CPT, although the target set is convex, the environment is not fixed. 
It depends on the empirical distribution of play at each step. A similar problem with dynamically evolving environment is considered in \citet{kalathil2017approachability}, 
where they get around this problem by considering a Stackelberg setting; one player (leader) plays an action first, then, 
after observing this action, the other player (follower) plays her action. In the absence of a Stackelberg setting, 
%as is in our case, 
%\red 
as in our case, 
%\black
we do not know of any result that characterizes approachability under dynamic environments. 
However, as far as games with CPT preferences are concerned, 
we answer this question by giving an example of a game for which a no-regret learning strategy does not exist (Example~\ref{ex: nonapproach_CPT}).

%!TEX root = ../main.tex

\section{Preliminaries}
\label{sec: prelim}

%\cor
%Please change $N = \{1,\dots,n\}$ to $[n] = \{1,\dots,n\}$,
%change $i \in N$ to $i \in [n]$ and change $N$ to $[n]$ in all
%places where $N$ refers to the set of players.
%\cob

We denote a finite $n$-person normal form game by 
$\Gamma := ([n],(A_i)_{i \in [n]},(x_i)_{i \in [n]})$, 
where $[n] := \{1,\dots,n\}$ is the set of \emph{players}, 
$A_i$ is the finite \emph{action} set of player $i$, 
and $x_i: A_1 \times \dots \times A_n \to \bbR$ is the \emph{payoff} function for player $i$. 
Let $A := \prod_{i \in [n]} A_i$ denote the set 
of all \emph{action profiles} $a := (a_i)_{i \in [n]}$, 
where $a_i \in A_i$. 
Let $A_{-i} := \prod_{j \in [n]\back i} A_j$ denote 
the set of all action profiles $a_{-i} \in A_{-i}$ of all players except player $i$. 
Let $x_i(a)$ denote the payoff to player $i$ when action profile $a$ is played, 
and let $x_i(\tilde a_i,a_{-i})$ denote 
the payoff to player $i$ when she chooses action $\tilde a_i \in A_i$ while the others stick to $a_{-i}$. 
For any finite set $S$, let $\Delta(S)$ denote the standard simplex of all probability distributions on the set $S$, i.e.,
\[
	\Delta(S) := \l\{\l(p(s),s \in S\r) \Big | p(s) \geq 0 \; \forall s \in S, \sum_{s \in S} p(s) = 1\r\},
\]
with the usual topology. 
Let $e_s$ denote the vector in $\Delta(S)$  with its $s$-th component equal to $1$ and the rest 
%\red 
equal to 
%\black 
$0$.
Let $\Delta^*(A)$ denote the set of all joint probability distributions that are of product form, i.e.,
%\[
%	 \Delta^*(A) := \{\mu \in \Delta(A): \mu(a) = \mu_1(a_1) \times \dots \times \mu_n(a_n), \forall \; a \in A\},
%\]
%\red
\[
	 \Delta^*(A) := \{\mu \in \Delta(A): \mu(a) = 
	 \mu_1(a_1) \mu_2(a_2) \dots \mu_n(a_n), \forall \; a \in A\},
\]
%\black
where $\mu_i(a_i)$ denotes the marginal probability distribution on $a_i$ induced by $\mu$, 
%Thus, 
%\red 
namely, 
%\black
for a joint distribution $\mu \in \Delta(A)$, we have 
\[
	\mu_i(a_i) = \sum_{a_{-i} \in A_{-i}} \mu(a_i,a_{-i}).
\]
For $a_i$ such that $\mu_i(a_i) > 0$, let
\[
	\mu_{-i}(a_{-i}|a_i) := \frac{\mu(a_i,a_{-i})}{\mu_i(a_i)}.
\]
%be the conditional distribution on $A_{-i}$.
%\color{red} VA: I deleted a phrase that was there after the preceding equation. \color{black}

%%%%%%%%%%%%%%%%%%%%%%%%%%%%%%%%%%%%%%%%%%%%%%%
%%%%%%%%%%%%%%%%%%%%%%%%%%%%%%%%%%%%%%%%%%%%%%%

We now describe the setup for cumulative prospect theory (CPT) (for more details see \citep{wakker2010prospect}). 
Each person is associated with a \emph{reference point} $r \in \bbR$, a corresponding \emph{value function} $v^r : \bbR \to \bbR$, and two \emph{probability weighting functions} $w^\pm:[0,1] \to [0,1]$, $w^+$ for gains and $w^-$ for losses. The function $v^r(x)$ satisfies:
\begin{inparaenum}[(i)]
	\item it is continuous in $x$;
	\item $v^r(r) = 0$;
    \item it is strictly increasing in $x$.
\end{inparaenum}
The value function is generally assumed to be convex in the losses frame ($x < r$) and concave in the gains frame ($x \geq r$), and to be steeper in the losses frame than in the gains frame in the sense that $v^r(r) - v^r(r -z) \ge v^r(r+z) - v^r(r)$ for all $z \ge 0$. 
However, these assumptions are not needed for the results in this paper to hold.
%{\color{red} Please check if the preceding sentence is accurate. \color{black}}
The probability weighting functions $w^\pm: [0,1] \to [0,1]$ satisfy:
\begin{inparaenum}[(i)]
	\item they are continuous;
	\item they are strictly increasing;
	\item $w^\pm(0) = 0$ and $w^\pm(1) = 1$.
\end{inparaenum}

Suppose a person faces a \emph{lottery} (or \emph{prospect}) $L := \{(p_j,z_j)\}_{1 \leq j \leq t}$,
where $z_j \in \bbR, 1 \leq j \leq t$, denotes an \emph{outcome} and $p_j, 1 \leq j \leq t$, is the probability with which outcome $z_j$ occurs. We assume 
%the lottery to be exhaustive, 
%\red 
that the lottery is {\em exhaustive}, 
%\black
i.e. $\sum_{j=1}^t p_j = 1$. (Note that we are allowed to have $p_j =0$ for some values of $j$ and 
%\color{red} 
we can have 
%\color{black} 
$z_k = z_l$ even when $k \neq l$.) 
Let $z := (z_j)_{1 \leq j \leq t}$ and $p := (p_j)_{1 \leq j \leq t}$. We denote $L$ as $(p,z)$ and refer to the vector $z$ as an \emph{outcome profile}.

Let $\alpha := (\alpha_1,\dots,\alpha_t)$ be a permutation of $(1,\dots,t)$ such that
\begin{equation}\label{eq: order}
	z_{\alpha_1} \geq z_{\alpha_2} \geq \dots \geq z_{\alpha_t}.
\end{equation}
Let $0 \leq j_r \leq t$ be such that $z_{\alpha_j} \geq r$ for $1 \leq j \leq j_r$ and $z_{\alpha_j} < r$ for $j_r < j \leq t$. (Here $j_r = 0$ when $z_{\alpha_j} < r$ for all $1 \leq j \leq t$.) The \emph{CPT value} $V(L)$ of the prospect $L$ is evaluated using the value function $v^r(\cdot)$ and the probability weighting functions $w^{\pm}(\cdot)$ as follows:
\begin{equation}\label{eq: CPT_value_discrete}
	V(L) := \sum_{j=1}^{j_{r}} \pi_j^+(p,\alpha) v^r(z_{\alpha_j}) + \sum_{j=j_r+1}^t \pi_j^-(p,\alpha) v^r(z_{\alpha_j}),
\end{equation}
where $\pi^+_j(p,\alpha),1 \leq j \leq j_{r}, \pi^-_j(p,\alpha), j_r < j \leq t$, are \emph{decision weights} defined via:
\begin{align*}
	\pi^+_{1}(p,\alpha) &:= w^+(p_{\alpha_1}),\\ %&\text{ if } &j_r > 0,\\
	\pi_j^+(p,\alpha) &:= w^+(p_{\alpha_1} + \dots + p_{\alpha_{j}}) - w^+(p_{\alpha_1} + \dots + p_{\alpha_{j-1}}) &\text{ for } &1 < j \leq t, \\
	 \pi_j^-(p,\alpha) &:= w^-(p_{\alpha_t} + \dots + p_{\alpha_j}) - w^-(p_{\alpha_t} + \dots + p_{\alpha_{j+1}}) &\text{ for } &1 \leq j < t,\\
	 \pi^-_{t}(p,\alpha) &:= w^-(p_{\alpha_t}). %&\text{ if } &j_r < t.
\end{align*}
Although the expression on the right in equation~(\ref{eq: CPT_value_discrete}) depends on the permutation $\alpha$, one can check that the formula evaluates to the same value $V(L)$ as long as the permutation $\alpha$ satisfies (\ref{eq: order}). 
%\blue
The CPT value in equation~(\ref{eq: CPT_value_discrete}) can equivalently be written as:
\begin{align}\label{eq: CPT_value_cumulative}
	V(L) &= \sum_{j = 1}^{j_r - 1} w^+\l(\sum_{i = 1}^j p_{\alpha_i}\r)\l[v^r(z_{\alpha_j}) - v^r(z_{\alpha_{j+1}})\r] \nonumber\\
 	&+ w^+\l(\sum_{i = 1}^{j_r} p_{\alpha_i}\r)v^r\l(z_{\alpha_{j_r}}\r) + w^-\l(\sum_{i = j_r + 1}^{t} p_{\alpha_i}\r)v^r(z_{\alpha_{j_r+1}}) \nonumber \\
	&+ \sum_{j = j_r + 1}^{t-1} w^-\l(\sum_{i = j+1}^t p_{\alpha_i}\r)\l[v^r(z_{\alpha_{j+1}}) - v^r(z_{\alpha_{j}})\r].
 \end{align}
% \black

A person is said to have CPT preferences if, given a choice between prospect $L_1$ and prospect $L_2$, she chooses the one with higher CPT value. 

%CPT satisfies \emph{strict stochastic dominance} \citep{chateauneuf1999axiomatization}: shifting positive probability mass from an outcome to a strictly preferred outcome leads to a strictly preferred prospect. For example, the prospect $L_1 = \{(0.6,40);(0.4,20)\}$ can be obtained from the prospect $L_2 = \{(0.5,40);(0.5,20)\}$ by shifting a probability mass of $0.1$ from outcome $20$ to a strictly better outcome $40$. The strict stochastic dominance condition says that $V^r(L_1) > V^r(L_2)$ (see equation~(\ref{eq: CPT_value_cumulative})).

%Also, CPT satisfies \emph{strict monotonicity} \citep{chateauneuf1999axiomatization}: any prospect becomes strictly better as soon as one of its outcomes is strictly improved. For example, if $L_1 = \{(0.6,40);(0.4,-10)\}$ and $L_2 = \{(0.6,40);(0.4,-20)\}$, then $V^r(L_1) > V^r(L_2)$ (see equation~(\ref{eq: CPT_value_discrete})).

%%%%%%%%%%%%%%%%%%%%%%%%%%%%%%%%%%%%%%%%%%%%%%%
%%%%%%%%%%%%%%%%%%%%%%%%%%%%%%%%%%%%%%%%%%%%%%%

We now describe the notion of correlated equilibrium incorporating CPT preferences, as defined by \citet{keskin2016equilibrium}\footnote{Keskin defines CPT equilibrium assuming
%$w^+(\cdot) = w^-(\cdot)$. 
%\color{red} 
that $w_i^+(\cdot) = w_i^-(\cdot)$ for each player $i$.
%\color{black} 
However, the definition can be easily extended to our general setting and the proof of existence goes through without difficulty.}. 
%\color{red} vacomment: The footnote has been moved from where it was.
%\color{black} 
For each player $i$, let $r_i,v_i^{r_i}(\cdot)$ and $w_i^\pm(\cdot)$ be the reference point, the value function, and the probability weighting functions, respectively, that player $i$
 uses to evaluate the CPT value $V_i(L)$ of a lottery $L$.
 %\red 
 We call these the {\em CPT features} of the player $i$. 
 %\black
 %Let $\{v_i^r(\cdot),r \in \bbR\}$ be a family of value functions, one for each reference point, 
%{\color{red} Why do you need value functions for each reference point? Each paper needs only a value function for her fixed reference point. \color{black}} and $w_i^\pm(\cdot)$ be the probability weighting functions for each player $i \in N$. 
%We assume that $v_i^r(x)$ is continuous in $x$ and $r$ for each $i$.
%For every player $i \in N$, let the reference point be determined by a continuous function $r_i:\Delta^{|S|-1} \to \bbR$. 
%Let denote  evaluated by player $i$, using the reference point $r_i$, the value function $v_i^r(\cdot)$ and the probability weighting functions $w_i^\pm(\cdot)$.

Suppose there is a \emph{mediator} characterized by a joint distribution $\mu \in \Delta(A)$ who draws 
an action profile $a = (a_i)_{i \in [n]}$ according to the distribution $\mu$ and sends signal $a_i$ to each player $i$. 
Player $i$ is signaled only her action $a_i$ and not the entire action profile $a = (a_i)_{i \in [n]}$. 
We assume that all the players know the distribution $\mu$. 
If player $i$ observes a signal to play $a_i$, and if she decides to deviate to a strategy $\tilde a_i \in A_i$, then she will face the lottery
%\[
%	L_i(\mu,a_i,\tilde a_i) := \l\{ \l(\mu_{-i}(a_{-i}|a_i), x_i(\tilde a_i,a_{-i}) \r)\r\}_{a_{-i} \in A_{-i}}.
%\]
%\bb
\[
	L_i(\mu_{-i}(a_{-i}|a_i),\tilde a_i) := \l\{ \l(\mu_{-i}(a_{-i}|a_i), x_i(\tilde a_i,a_{-i}) \r)\r\}_{a_{-i} \in A_{-i}}.
\]
%\cob
%\cor (I have changed the notational convention on the LHS of the equation, 
%for consistency with the later use of this quantity.)
%\cob

%%%%%%%%%%%%%%%%%%%%%%%%%%%%%%%%%%%%%%%%%%%%%%%

\begin{definition}[\cite{keskin2016equilibrium}]
\label{def: CPT_Nash_eq}

	A joint probability distribution 
	%$\mu \in \Delta(S)$ 
	%\color{red} 
	$\mu \in \Delta(A)$ 
	%\color{black}
	is said to be a \emph{CPT correlated equilibrium} of $\Gamma$ if it satisfies the following inequalities for all $i$ and for all $a_i,\tilde a_i \in A_i$ such that $\mu_i(a_i) > 0$:
	%\begin{equation}\label{eq: CPT_corr_ineq}
	%	V_i(L_i(\mu,a_i,a_i)) \geq V_i(L_i(\mu,a_i,\tilde a_i)). 
	%\end{equation}
    %\bb
    \begin{equation}\label{eq: CPT_corr_ineq}
		V_i(L_i(\mu_{-i}(a_{-i}|a_i),a_i)) \geq V_i(L_i(\mu_{-i}(a_{-i}|a_i),\tilde a_i)). 
	\end{equation}
    %\cob
    %\cor (Changed the notation for consistency, see above.)
    %\cob
\end{definition}

%%%%%%%%%%%%%%%%%%%%%%%%%%%%%%%%%%%%%%%%%%%%%%%

We denote the set of all the CPT correlated equilibria of a game $\Gamma$ by $C(\Gamma)$. Note that $C(\Gamma)$ also depends on the %reference points, the value functions and the probability weighting functions of all the players. 
%\red 
CPT features of each of the players. 
%\black
However, we suppress this dependence from the notation.

We now describe the notion of CPT Nash equilibrium as defined by \citet{keskin2016equilibrium}. For a mixed strategy $\mu \in \Delta^*(A)$, 
%\cor 
if each player $j$ decides to play $a_j$, drawn from the distribution $\mu_j$, then player $i$ will face the lottery 
%\cob
\[
	L_i(\mu_{-i},a_i) := \l\{\l(\mu_{-i}(a_{-i}), x_i(a_i,a_{-i}) \r)\r\}_{a_{-i} \in A_{-i}},
\]
where $\mu_{-i}(a_{-i}) := \prod_{j \neq i} \mu_j(a_j)$ plays the role of $\mu_{-i}(a_{-i}|a_i)$, which does not depend on $a_i$.
Suppose player $i$ decides to deviate and play a mixed strategy $\tilde \mu_i$ while the rest of the players continue to play $\mu_{-i}$. Then define
the average CPT value for player $i$ by
\[
	\av_i(\tilde \mu_i,\mu_{-i}) := \sum_{a_i \in A_i} \tilde \mu_i(a_i)V_i(L_i(\mu_{-i},a_i)).
\]
The best response set of player $i$ to a mixed strategy $\mu \in \Delta^*(A)$ is defined as
%\red
\begin{eqnarray}    \label{eq:bestresponseset}
	BR_i(\mu) &:=& \l\{\mu^*_i \in \Delta(A_i) | \forall \tilde \mu_i \in \Delta(A_i), \av_i(\mu^*_i,\mu_{-i}) \geq \av_i(\tilde \mu_i,\mu_{-i}) \r\} \nonumber \\
	&=& \l\{\mu^*_i \in \Delta(A_i) | \mbox{supp}(\mu_i^*) \subset \mbox{arg} \max_{a_i \in A_i} V_i(L_i(\mu_{-i},a_i)) \r\}.
\end{eqnarray}
Here $\supp(\cdot)$ denotes the support of the distribution within the 
%brackets.
%\red 
parentheses. 
%\black
%\black

%%%%%%%%%%%%%%%%%%%%%%%%%%%%%%%%%%%%%%%%%%%%%%%

\begin{definition}[\cite{keskin2016equilibrium}]
\label{def: CPT_Nash_eq_real}

	A mixed strategy $\mu^* \in \Delta^*(A)$ is a 
	%\red 
	{\em CPT Nash equilibrium} 
	%\red 
	of $\Gamma$ 
	%\black
	%\black 
	iff
	\[
		\mu^*_i \in BR_i(\mu^*) \text{ for all } i.
	\]
    %\cor
    \end{definition}
    \citet{keskin2016equilibrium} shows that for every game $\Gamma$ there exists a CPT Nash equilibrium. Further, he also shows that the set of all CPT Nash equilibria of a game $\Gamma$ is equal to $C(\Gamma) \cap \Delta^*(A)$. 
    %Thus, as a consequence, 
    %\red 
    As a consequence 
    %\black
    we have that the set $C(\Gamma)$ is nonempty. 
	A strategy $\mu \in \Delta^*(A)$ is called a \emph{pure} strategy if the support of $\mu_i$ is singleton for each $i$. We call $\mu^*$ a pure CPT Nash equilibrium if $\mu^*$ is a pure strategy. Note that every pure CPT Nash equilibrium is a pure Nash equilibrium for the EUT game where each player $i$ computes its value in the action profile $(a_i)_{i \in [n]}$ as $v_i^{r_i} (x_i(a_i,a_{-i}))$.
    %\cob
    %{\color{red} Is the set of correlated CPT equilibria intersected with the set of joint distributions with independent marginals precisely the set of CPT Nash equilibria? \color{black}}

%%%%%%%%%%%%%%%%%%%%%%%%%%%%%%%%%%%%%%%%%%%%%%%

%!TEX root = ../main.tex

\section{Calibrated learning in games}
\label{sec: calibrated_learning}

%We now look at the notion of equilibrium under CPT preferences from evolutionary viewpoint through repeated games model. 
Let 
%$\Gamma$ 
%\red 
$\Gamma = ([n],(A_i)_{i \in [n]},(x_i)_{i \in [n]})$ 
%\black
be a finite $n$-person game which is played repeatedly at each \emph{step} 
%(or \emph{round}), 
$t \geq 1$. 
The game $\Gamma$ is called the 
%\emph{one shot game} (or the \emph{stage game}) 
%\red 
\emph{stage game} 
%\black
of the repeated game. 
At every step $t$, each player $i$ draws an action $a_i^t \in A_i$ %from a distribution
%\red 
with the probability distribution 
%\black
$\sigma_i^t \in \Delta(A_i)$. 
We assume that the randomizations of the players are independent of each other and of the past randomizations. 
For example, if each player $i$ uses a uniform random variable $U_i^t$ to draw a sample from $\sigma_i^t$, then the random variables $\{U_i^t\}_{i \in [n], t \geq 1}$ are independent.
%\red 
Each player is assumed to know the action space of 
all the players in the 
stage game 
%\black
$\Gamma$, but does not know the payoff functions 
and the CPT parameters of the other players. 
%\black
We assume that, after playing her action $a_i^t$, each player observes the actions taken by all the other players 
and thus at any step $t$ all the players have access to the %\emph{history} 
%\red 
\emph{past history} 
%\black
of the play 
%\red at time $t$, \black
%\blue 
at step $t$, 
%\black
$H^{t-1} := (a^1,\dots,a^{t-1})$, 
where $a^t := (a_i^t)_{i \in [n]}$
 is the action profile played at step $t$. 
%We also assume that each player knows the 
%one step game 
%\red 
%stage game 
%\black
%$\Gamma$.
Let the strategy for player $i$ for the repeated game above be given by $\strat_i  := (\sigma_i^t, t \geq 1)$, 
where $\sigma_i^t : H^{t-1} \to \Delta(A_i)$, for each $t$.
% Let $\strat = (\strat_i)_{i \in N}$ denote of a profile of strategies for each player and let $P_\strat$ denote the probability distribution on the space of all action profile sequences $(a^1,a^2,\dots)$ generated from the strategies $\strat_i, i \in N$.

 We first describe the result of \citet{foster1997calibrated}. Suppose the players  follow the following natural strategy: At every step $t$, on the basis of 
 the past history of play, $H^{t-1}$, each player $i$ predicts a joint distribution $\mu_{-i}^t \in \Delta(A_{-i})$ on the action profile of all the other players. 
 This is player $i$'s \emph{assessment} of how her opponents might play 
 %in the next step. 
 %\red 
 at step $t$. 
 %\black
 %\red 
 The sequence of functions of past history giving rise
 to the assessment is called the {\em assessment scheme} 
 of the player. 
 %\black
 %\cor
 Depending on her assessment 
 %\red 
 at step $t$, 
 %\black
 player $i$ chooses a specific action among those that are most preferred by her
 %\red 
 in response to her assessment, 
 %\black
 called her \emph{best reaction}.\footnote{\citet{foster1997calibrated} refer to it as the best response. In order to avoid confusion with the best response set defined in section \ref{sec: prelim}, we prefer to use the term best reaction.}
 %\red 
 This is done using a fixed (time-invariant) function from 
 $\Delta(A_{-i})$ to $A_i$, which maps $\mu_{-i} \in \Delta(A_{-i})$
 to an action in $A_i$ that is in the best response set for
 $\mu_{-i}$; this function is called the 
 {\em best reaction map} of player $i$. 
 %\black
 \citet{foster1997calibrated} prove that
 \begin{inparaenum}[(i)]
 	\item if each player's assessments are \emph{calibrated} with respect to the sequence of action profiles of the other players and
 	\item if each player plays 
    %a best reaction
    %\bb
    the best reaction
    %\cob
 	 to her assessments,
 \end{inparaenum}  
 then the limit points of the empirical distribution of action play are correlated equilibria. 
 %\cob
 By \emph{action play} we mean the sequence of action profiles played by the players. We will give a formal definition of what is meant by calibration shortly. For the moment, 
 %roughly, 
% \color{red}
 roughly speaking,
% \color{black}
 calibration says that the empirical distributions conditioned on assessments converge to the assessments.
 %{\color{red} The preceding sentence doesn't sound right when compared with the
 %earlier way that ``calibration" was described. Need to rephrase this in some way. \color{black}}
 The best reaction of player $i$ to her assessment $\mu_{-i}$ of the actions of the other players, as considered by \citet{foster1997calibrated}, 
 is a specific action $a^*_i \in A_i$ that maximizes the expected payoff to player $i$ with respect to her 
 %assessment. i.e.,
%\red 
assessment, i.e., 
%\black
 \[
 	a^*_i \in \arg\max_{a_i \in A_i} \sum_{a_{-i} \in A_{-i}} \mu_{-i}(a_{-i}) x_i(a_i,a_{-i}).
 \]
 Thus the best reaction is an action in the best response set.
 %\cor
 Note that it is assumed that each player uses a fixed tie breaking rule if there is more than one action in the best response set.
 
 Suppose now that the players behave 
 %according to CPT.
 %\red 
 with CPT preferences. 
 %\black
 %\cob
 Given player $i$'s assessment $\mu_{-i}$ of the play of her opponents, she is faced with the following set of lotteries, one for each of her actions $a_i \in A_i$:
 \[
 	L_i(\mu_{-i}, a_i) := \l\{\mu_{-i}(a_{-i}),x_i\l(a_i,a_{-i}\r)\r\}_{ a_{-i} \in A_{-i} }.
 \]
 Out of these lotteries, the ones she prefers most are those with the maximum CPT value $V_i\l( L_i \l( \mu_{-i}, a_i \r)\r)$, evaluated using her 
 %reference point $r_i$, value function $v_i^{r_i}$, 
 %corresponding to her reference point
 %and her probability weighting functions $w_i^{\pm}$. 
 %\red 
 CPT features. 
 %\black
 The choice of the action she takes corresponding to her most preferred lottery (with any arbitrary but fixed tie breaking rule) will be called her best reaction,
 %\red 
 and the map from $\Delta(A_{-i})$ to $A_i$ giving the 
 best reaction as a function of the assessment will be called 
 the best reaction map of player $i$. 
 %\black
 Thus,
 %\red 
 once again, 
 %\black
 the best reaction is a specific action in the best response set. 
 %The main result of this section is that under these and a few more technical assumptions, the empirical distribution of action play converges to the set $B(\Gamma)$.

We now ask the following question: \emph{Suppose each player's assessments are calibrated with respect to the sequence of action profiles of the other players and she evaluates her best reaction in accordance with CPT preferences as explained above, then are the limit points of the empirical distribution of play contained in the set of CPT correlated equilibria?} Unfortunately, the answer is no %(Example~\ref{ex: nonconvergence_calib}). 
%\red 
(see Example~\ref{ex: nonconvergence_calib}). 
%\black
Before seeing why, let us give the promised formal definition of the notion of calibration.

%%%%%%%%%%%%%%%%%%%%%%%%%%%%%%%%%%%%%%%%%%%%%%%
%%%%%%%%%%%%%%%%%%%%%%%%%%%%%%%%%%%%%%%%%%%%%%%

Consider a sequence of outcomes $y^1,y^2,\dots$ generated by Nature, belonging to some finite set $S$. At each step $t$, the forecaster predicts a distribution $q^t \in \Delta(S)$.
%based on the history $(y^1,y^2,\dots,y^{t-1})$. 
 Let $N(q,t)$ denote the number of times the distribution $q$ is forecast up to step $t$, i.e. $N(q,t) := \sum_{\tau = 1}^t \1 \{q^\tau = q\}$,
 %\com
 where $\1\{\cdot\}$ is the indicator function that takes value $1$ if the expression inside $\{\cdot\}$ holds and $0$ otherwise. 
 %\cob
 %Let $\chi(y,t) = 1$ if Nature plays $y \in S$ on step $t$ and equal to zero otherwise.
 Let $\rho(q,y,t)$ be the fraction of the steps on which the forecaster predicts $q$ for which Nature plays $y \in S$, i.e.,
 %\red
\[
	\rho(q,y,t) := \begin{cases}
		0, &\text{ if } N(q,t) = 0,\\
		\frac{\sum\limits_{\tau = 1}^t \1 \{q^\tau = q\} \1 \{ y^\tau = y \} }{N(q,t)}, &\text{ otherwise}.
	\end{cases}
\]
%\black
The forecast is said to be calibrated with respect to the sequence of plays made by Nature if
%\red
\begin{equation}        \label{eq:calibrated}
	\lim_{t \to \infty} \sum_{q \in Q^t} | \rho(q,y,t) - q(y) | \frac{N(q,t)}{t} = 0, \text{ for all } y \in S,
\end{equation}
%\black
where the sum is over the set $Q^t$ of all distributions predicted by the forecaster up to step $t$.

%{\color{blue} example that shows calibrated learning need not converge to $C_{CPT}$ by extending Keskin's non convex example. Set up Bayes equilibrium and then show convergence result.}

%%%%%%%%%%%%%%%%%%%%%%%%%%%%%%%%%%%%%%%%%%%%%%%
%%%%%%%%%%%%%%%%%%%%%%%%%%%%%%%%%%%%%%%%%%%%%%%

%!TEX root = ../main.tex

\begin{example}
 \label{ex: nonconvergence_calib}
 We consider a modification of the $2$-player game proposed by \cite{keskin2016equilibrium}, 
 %where it is used 
 %\red 
 who uses it 
 %\black
 to demonstrate that the set of CPT correlated equilibria can be nonconvex.
Let the $2$-player game 
%\com
$\Gamma^*$ 
%\cob
be represented by the matrix in table \ref{tab: 2x4 game}, where $\beta = 1/w_1^+(0.5)$. 
For the probability weighting functions $w_i^{\pm}(\cdot)$, we employ the functions of the form suggested by \citet{prelec1998probability}, which, for $i = 1,2$, are given by
\[
	w_i^\pm(p) = \exp \{-(-\ln p)^{\gamma_i}\},
\]
where $\gamma_1 = 0.5$ and $\gamma_2 = 1$.
We thus have $w_1^+(0.5) = 0.435$ and $\beta = 2.299$.
Let the reference points be $r_1 = r_2 = 0$. 
Let $v_i^{r_i}(\cdot)$ be 
%identity 
%\bb
the identity function
%\cob
for $i = 1,2$.
%\cor
Notice that player $2$ is indifferent amongst her actions. 
%\cob

%%%%%%%%%%%%%%%%%%%%%%%%%%%%%%%%%%%%%%%%%%%%%%%

\begin{table}
\centering
\begin{tabular}{c | c | c | c | c |}
	 \multicolumn{1}{c}{}	& \multicolumn{1}{c}{I}   &  \multicolumn{1}{c}{II}  &  \multicolumn{1}{c}{III}  &  \multicolumn{1}{c}{IV} \\
	 \cline{2-5}
	 0 & $2 \beta,1$  &  $\beta + 1,1$  &  $0,1$  & $1,1$\\
	 \cline{2-5}
	 1	& $1.99,0$ & $1.99,0$ & $1.99,0$ & $1.99,0$\\
	 \cline{2-5}
	 \end{tabular} 
	 \caption{Payoff matrix for the game 
	 %\com
	 $\Gamma^*$
	 %\cob
	 in example~\ref{ex: nonconvergence_calib}. The rows and columns correspond to player $1$ and $2$'s actions respectively. 
     %\cor 
     The first entry in each cell corresponds to player $1$'s payoff and second to player $2$'s payoff.
     %\cob
     }\label{tab: 2x4 game}
\end{table}

%%%%%%%%%%%%%%%%%%%%%%%%%%%%%%%%%%%%%%%%%%%%%%%

%\com
Let $\mu_{odd} := (0.5,0,0.5,0)$ and $\mu_{even} := (0,0.5,0,0.5)$ 
be 
%\red 
probability 
%\black
distributions on player $2$'s actions. 
We can evaluate the CPT values of player $1$ for the following lotteries:
\begin{align*}
V_1(L_1(\mu_{odd},\text{0})) = 2\beta w_1^+(0.5) = 2,&& V_1(L_1(\mu_{odd},\text{1})) = 1.99,\\
V_1(L_1(\mu_{even},\text{0})) = 1 + \beta w_1^+(0.5) = 2,&& V_1(L_1(\mu_{even},\text{1})) = 1.99.
\end{align*}
Thus, player $1$'s best reaction to both these distributions $\mu_{odd}$ and $\mu_{even}$ is action 0. 
Since, player $2$ is indifferent amongst her actions, we get that the distributions $\mu^o$ and $\mu^e$, represented in tables \ref{tab: 2x4 game dist odd} and \ref{tab: 2x4 game dist even} respectively,
belong to the set $C(\Gamma^*)$. 
The mean of these two distributions is given by $\mu^*$ as represented in Table~\ref{tab: 2x4 game dist}. 
Let $\mu_{unif} := (0.25,0.25,0.25,0.25)$ be the uniform distribution on player $2$'s actions. 
%We have, the CPT values of player $1$ for the following lotteries:
%\red 
The CPT values of player $1$ for the lotteries corresponding to player $2$ playing $\mu_{unif}$ are: 
%\black
\begin{align*}
&V_1(L_1(\mu_{unif},\text{0})) = w_1^+(0.75) + \beta w_1^+(0.5) + (\beta -1) w_1^+(0.25) = 1.985,\\
&V_1(L_1(\mu_{unif},1)) = 1.99,
\end{align*}
%\color{red}
%\red
since $w_1^+(0.25) = 0.308$ and $w_1^+(0.75) = 0.585$.
%\color{black}
%Thus, 
We see that
player $1$'s best reaction to the distribution
$\mu_{unif}$ of player $2$ 
%\black
is action 1. 
%\color{red}
%vacomment: Please provide the values of $w_1^+(0.25)$ and $w_1^+(0.75)$ 
%above.
%\color{black}
This shows that $\mu^* \notin C(\Gamma^*)$, and hence $C(\Gamma^*)$ is not convex. 

Using this fact, we will 
%\red 
attempt to 
%\black
construct 
%a sequence of 
%calibrated assessments 
%\red 
an assessment scheme 
%\black
and 
%their corresponding best reactions 
%\red 
a best reaction function 
%\black
for each player such that 
%\red 
if each player makes assessments at each step according to her 
assessment scheme and acts according to the best reaction to her
assessment at each step, then the assessments of each
player are calibrated with respect to the sequence of action 
profiles of the other 
%players 
%\red
player
%\black
and
%\black
the limit of the generated empirical distribution of action play does not belong to $C(\Gamma^*)$.

Suppose player $2$ plays her actions in a cyclic manner starting with action I at step $1$, followed by 
%\color{red}
actions
%\color{black}
II, III, IV and then again I and so on. Suppose player $1$'s assessment of player $2$'s action is $\mu_{odd} = (0.5,0,0.5,0)$ and $\mu_{even} = (0,0.5,0,0.5)$ at each odd and even step respectively. Then it is easy to see that player $1$'s assessments are calibrated
%\red 
with respect to the sequence of actions of player $2$. 
%\black
(Here player $2$ plays the role of Nature from the point of view of player $1$.)
%We can evaluate the CPT values of player $1$ for the following lotteries as
%\begin{align*}
%V_1(L_1(\mu_{odd},\text{0})) = 2\beta w_1^+(0.5) = 2,&& V_1(L_1(\mu_{odd},1)) = 1.99\\
%V_1(L_1(\mu_{even},\text{0})) = 1 + \beta w_1^+(0.5) = 2,&& V_1(L_1(\mu_{even},1)) = 1.99
%\end{align*}
%Thus, player $1$'s best reaction to both these assessments $\mu_{odd}$ and $\mu_{even}$ is action 0. 
Since player $1$'s best reaction is action 0 to all her assessments, she would play action 0 throughout. The distribution $\mu^*$ 
%represented in table \ref{tab: 2x4 game dist} 
is a limit point of the empirical distribution of action play and does not belong to $C(\Gamma^*)$.

%\cob

%%%%%%%%%%%%%%%%%%%%%%%%%%%%%%%%%%%%%%%%%%%%%%%

\begin{table}
\centering
\begin{tabular}{c | c | c | c | c |}
	 \multicolumn{1}{c}{}	& \multicolumn{1}{c}{I}   &  \multicolumn{1}{c}{II}  &  \multicolumn{1}{c}{III}  &  \multicolumn{1}{c}{IV} \\
	 \cline{2-5}
	0 & $0.5$  &  $0$  &  $0.5$  & $0$\\
	 \cline{2-5}
	 1	& $0$ & $0$ & $0$ & $0$\\
	 \cline{2-5}
	 \end{tabular} 
	 \caption{
	 %\com
	 Empirical distribution $\mu^o$ for the action play in example~\ref{ex: nonconvergence_calib}.
	 %\cob
	 }\label{tab: 2x4 game dist odd}
\end{table}

%%%%%%%%%%%%%%%%%%%%%%%%%%%%%%%%%%%%%%%%%%%%%%%

\begin{table}
\centering
\begin{tabular}{c | c | c | c | c |}
	 \multicolumn{1}{c}{}	& \multicolumn{1}{c}{I}   &  \multicolumn{1}{c}{II}  &  \multicolumn{1}{c}{III}  &  \multicolumn{1}{c}{IV} \\
	 \cline{2-5}
	0 & $0$  &  $0.5$  &  $0$  & $0.5$\\
	 \cline{2-5}
	 1	& $0$ & $0$ & $0$ & $0$\\
	 \cline{2-5}
	 \end{tabular} 
	 \caption{
	 %\com
	 Empirical distribution $\mu^e$ for the action play in example~\ref{ex: nonconvergence_calib}.
	 %\cob
	 }\label{tab: 2x4 game dist even}
\end{table}

%%%%%%%%%%%%%%%%%%%%%%%%%%%%%%%%%%%%%%%%%%%%%%%

\begin{table}
\centering
\begin{tabular}{c | c | c | c | c |}
	 \multicolumn{1}{c}{}	& \multicolumn{1}{c}{I}   &  \multicolumn{1}{c}{II}  &  \multicolumn{1}{c}{III}  &  \multicolumn{1}{c}{IV} \\
	 \cline{2-5}
	0 & $0.25$  &  $0.25$  &  $0.25$  & $0.25$\\
	 \cline{2-5}
	 1	& $0$ & $0$ & $0$ & $0$\\
	 \cline{2-5}
	 \end{tabular} 
	 \caption{Empirical distribution $\mu^*$ for the action play in example~\ref{ex: nonconvergence_calib}.}\label{tab: 2x4 game dist}
\end{table}

%%%%%%%%%%%%%%%%%%%%%%%%%%%%%%%%%%%%%%%%%%%%%%%

We have not described player $2$'s assessments.
We would like to 
%\red 
come up with an assessment scheme and a best reaction map for player $2$ such that 
%\black
%show that 
%\red 
if player $2$ forms assessments according to this 
assessment scheme and acts according to this best reaction 
map, then the sequence of her actions is the cyclic sequence
that we require her to play and, further, 
%\black
player $2$'s assessments are calibrated 
%and her best reactions lead to the cyclic play exhibited. 
%\red 
with respect to the sequence of actions of player $1$
(which is the all $0$ sequence). 
%\black
%However, to do so we need to modify the game into a $3$-person game.
%\color{red}
However, 
%\red 
it turns out that 
%\black
we cannot do so in this game.

Instead, we need to modify the game 
%\red 
$\Gamma^*$ 
%\black
%into a $3$-person game.
%\red
into a $3$-person game, denoted $\tilde{\Gamma}^*$. 
%\black
%\color{black}
Let player $1$ have two actions \{0,1\}, and players $2$ and $3$ 
%both have four actions \{I,II,III,IV\} each.
%\color{red}
each have four actions \{I,II,III,IV\}.
%\color{black}
Let the payoffs to all the three players be $-1$ if players $2$ and $3$ play different actions. 
%\cor
If players $2$ and $3$ play the same action, then let the resulting payoff matrix be as represented in table \ref{tab: 2x4 game}, where the rows correspond to player $1$'s actions and the columns correspond to
the common actions of players $2$ and $3$. 
%\cob
Player $1$ receives the payoff represented by the first entry in each cell and 
%player 
%\red
players
%\black
$2$ and $3$ each receive the payoff represented by the second entry.
Let player $1$'s 
%CPT behavior be characterized by the reference point, value function and probability weighting functions 
%as before in the 
%\bb
%\red 
CPT features be 
%\black
as in the
%\cob
$2$-person game
%\red 
$\Gamma^*$. 
%\black
%\cor 
For players $2$ and $3$, let them be as for player $2$ in that game.
Let 
%player 
%\red 
players 
%\black
$2$ and $3$ play in the cyclic manner as above, in sync with each other. 
%\cob
Let player $1$ play action 0 throughout.
%Let player $2$ and player $3$'s assessment at step $t$ be the point distribution supported by the action profile $a^t_{-2}$ and $a^t_{-3}$ respectively. This assessment is clearly calibrated. 
%(Here the action pair comprised of the actions of players $1$ and $3$ play the role of the actions of Nature from the point of view of player $2$, and similarly for player $3$.) 
%\bb
Let player $2$'s assessment at step $t$ be the point distribution supported by the action profile $a^t_{-2}$ which equals $0$ for player $1$ and the action played
by player $2$ for player $3$. Similarly, let 
player $3$'s assessment at step $t$ be the point distribution supported by the action profile $a^t_{-3}$ which equals $0$ for player $1$ and the action played
by player $3$ for player $2$.
%These assessments are clearly calibrated. 
%\red 
Then, for each of the players $2$ and $3$, her 
assessments are calibrated with respect to the sequence of 
%actions 
%\red
action
%\black
profiles of her opponents. 
%\black
Here the action pair comprised of the actions of players $1$ and $3$ plays the role of the actions of Nature from the point of view of player $2$, and similarly the action pair comprised of the actions of players $1$ and $2$ plays the role of the actions of Nature from the point of view of
player $3$.
%\cob
The actions of player $2$ and $3$ at every step are best reactions to their corresponding assessments.
Let the assessment of player $1$ be 
%$\mu_{odd}$ and $\mu_{even}$ 
%\red 
$\tilde{\mu}_{odd}$ and $\tilde{\mu}_{even}$ 
%\black
at odd and even steps
%\color{red}
respectively,
%\color{black}
where now the distribution 
%$\mu_{odd}$ 
%\red 
$\tilde{\mu}_{odd}$ 
%\black
puts $0.5$ probability on action profiles (I,I) and (III,III), and %$\mu_{even}$ 
%\red 
$\tilde{\mu}_{even}$ 
%\black
puts $0.5$ probability on action profiles (II,II) and (IV,IV). 
%\cor
Again player $1$'s assessments are calibrated 
%\red 
with respect to the sequence of action profiles of her
opponents 
%\black
(where now action pairs comprised of the actions of player $2$ and player $3$ play the role of the actions of Nature from the point of view of player $1$) 
and her actions are best reactions to her assessments. 
%\cob
The limit point of the empirical distribution of action play is the distribution that puts probability $0.25$ on action profiles (0,I,I), (0,II,II), (0,III,III) and (0,IV,IV). 
Since action 0 is not a best response of player $1$ to the distribution 
%$\sigma_{unif}$ 
%\red 
$\tilde{\mu}_{unif}$ 
%\black
that puts probability $0.25$ on action profiles (I,I), (II,II), (III,III) and (IV,IV), 
the limit point of the empirical distribution 
%is not a CPT correlated equilibrium. 
%\red 
is not a CPT correlated equilibrium of the $3$-player game
$\tilde{\Gamma}^*$. 
%\black
%Thus, we have a game where the assessments of each player are calibrated and each player plays her best reaction, 
%but the limit empirical distribution of action play exists and is not a CPT correlated equilibrium.
%\color{red}
Thus, we have a game where the assessments of each player are calibrated
%\red 
with respect to the sequence of action profiles of 
%the other players, 
%\red 
her opponents, 
%\black
%\black
each player plays 
%her best reaction, 
%\red 
her best reaction to her assessments at each step, 
%\black
and
the limit empirical distribution of action play exists but is not a CPT correlated equilibrium.
%\color{black}
%{\color{red} Confirm that the limit exists. \color{black}}
\end{example} 

%%%%%%%%%%%%%%%%%%%%%%%%%%%%%%%%%%%%%%%%%%%%%%%
%%%%%%%%%%%%%%%%%%%%%%%%%%%%%%%%%%%%%%%%%%%%%%%

%!TEX root = ../main.tex

\subsection*{Mediated Games}
\label{sec: mediated_eq_CPT}

%\st{The non-convexity of the set $C(\Gamma)$ is essential to the argument in example $\ref{ex: nonconvergence_calib}$.}
%\cor 
%The preceding sentence needs to be expanded on to explain why this is the case.
%\cob

%\com
In example~\ref{ex: nonconvergence_calib}, the fact that action 0 is player $1$'s best reaction to the distributions 
%$\mu_{odd}$ and $\mu_{even}$, but not to $\mu_{unif}$, 
%\red 
%$\tilde{\mu}_{odd}$ and $\tilde{\mu}_{even}$, but not to $\tilde{\mu}_{unif}$, 
%\blue
${\mu}_{odd}$ and ${\mu}_{even}$, but not to ${\mu}_{unif}$,
%\black
plays an essential role in showing 
%both 
the non-convexity of the set $C(\Gamma^*)$  
%\red 
in the $2$-player game $\Gamma^*$,
%\black
%\blue
and the fact that action 0 is player $1$'s best reaction to the distributions 
$\tilde{\mu}_{odd}$ and $\tilde{\mu}_{even}$, but not to $\tilde{\mu}_{unif}$,
helps us in showing
%\black 
the non-convergence of calibrated learning 
%to the set $C(\Gamma^*)$.
%\red 
to the set $C(\tilde{\Gamma}^*)$ in the $3$-player game 
$\tilde{\Gamma}^*$. 
%\black
%\cob
\black
We now describe 
%\com
a convex extension 
%\cob
of the set $C(\Gamma)$ 
%\red 
in a general finite $n$-person game $\Gamma$, 
%\black 
and 
%in section \ref{sec: calib_mediated} 
establish the convergence of the empirical distribution of action play to this extended set
%\red
when each player plays the best reaction to her assessment at
each step and her assessment scheme is calibrated with respect to the sequence of action profiles of her opponents.
%\black
%{\color{red} What does ``such an extension" mean in the preceding sentence? \color{black}} 
%\cor
It turns out that this extended set of equilibria also has a 
%game theoretic 
%\red 
game-theoretic 
%\black
interpretation, as follows. 
%\cob
Suppose we add a signal system $(B_i)_{i \in [n]}$ to a game $\Gamma$, 
where each $B_i$ is a finite set.
%\bb
(In Appendix~\ref{app: polish_signal_space}, we 
%\red 
study what happens when we 
%\black
relax the assumption that the sets $B_i$ are finite and show that in a certain sense it is enough to consider only finite signal sets.)
%\cob 
Suppose there is a mediator who sends a signal $b_i \in B_i$ to player $i$. 
Let $B := \Pi_{i \in [n]} B_i$ be the set of all signal profiles $b = (b_i)_{i \in [n]}$, and 
%\red 
let 
%\black
$B_{-i} := \Pi_{j \neq i} B_j$ denote the set of signal profiles $b_{-i}$ of all players except player $i$.
Let $\tilde \Gamma := (\Gamma, (B_i)_{i \in [n]})$ denote such a 
game with a signal system. 
We call it a \emph{mediated game}.
The mediator is characterized by a distribution $\psi \in \Delta(B)$ that we call the \emph{mediator distribution}. Thus, the mediator draws a signal profile $b = (b_i)_{i \in [n]}$ from the mediator distribution $\psi$ and sends signal $b_i$ to player $i$.
Let $\psi_i$ denote the marginal probability distribution on $B_i$ induced by $\psi$, and for $b_i$ such that $\psi_i(b_i) > 0$, let $\psi_{-i}(\cdot | b_i)$ denote the conditional probability distribution on $B_{-i}$.
In the definition of a correlated equilibrium, the set $B_i$ is 
%restricted to the set 
%\bb
restricted to be the set 
%\cob
of actions $A_i$
for each player $i$.

A \emph{randomized strategy} for any player $i$ is given by a function $\sigma_i : B_i \to \Delta(A_i)$ and a \emph{randomized strategy profile} $\sigma = (\sigma_1,\dots,\sigma_n)$ gives the randomized strategy for all players. 
We define the \emph{best response set} of player $i$ to a randomized strategy profile $\sigma$ and a mediator distribution $\psi$ as
%\begin{align}
%\label{eq: best_response}
%	BR_i(\psi, \sigma) := \bigg\{ \sigma^*_i &: B^i \to \Delta(A_i) \bigg| \text{ for all } b_i \in B_i \nonumber\\
%	& \supp(\sigma^*_i(b_i)) \subset \arg \max_{a_i \in A_i} V_i\l(\l\{ \tilde \mu_{-i}(a_{-i}|b_i), x_i(a_i,a_{-i}) \r\}_{a_{-i} \in A_{-i}} \r) \bigg\}
%\end{align}
%\bb
\begin{align}
\label{eq: best_response}
	BR_i(\psi, \sigma) := \bigg\{ \sigma^*_i &: B_i \to \Delta(A_i) \bigg| \text{ for all } b_i \in 
    %B_i
    \supp(\psi_i)
    , \nonumber\\
	& \supp(\sigma^*_i(b_i)) \subset \arg \max_{a_i \in A_i} V_i\l(\l\{ \tilde \mu_{-i}(a_{-i}|b_i), x_i(a_i,a_{-i}) \r\}_{a_{-i} \in A_{-i}} \r) \bigg\},
\end{align}
%\cob
where
\begin{equation}
\label{eq: tilde_mu}
	\tilde \mu_{-i}(a_{-i} | b_i ) := \sum_{b_{-i} \in B_{-i}} \psi_{-i}(b_{-i}|b_i) \prod_{j \in [n] \back i} \sigma_j(b_j)(a_j),
\end{equation}
and $\supp(\cdot)$ denotes the support of the distribution within the %brackets.
%\red 
parentheses. 
%\black

%%%%%%%%%%%%%%%%%%%%%%%%%%%%%%%%%%%%%%%%%%%%%%%

\begin{definition}
\label{def: mediated_cpt_nash_eq}
 A randomized strategy profile $\sigma$ is said to be a  \emph{mediated CPT Nash equilibrium} 
%of a game $\tilde \Gamma$ 
%\red 
of a mediated game 
$\tilde \Gamma = (\Gamma, (B_i)_{i \in [n]})$ 
%\black
 with respect to a mediator distribution $\psi \in \Delta(B)$ 
 %iff
 %\red 
 if 
 %\black
	\[
		\sigma_i \in BR_i(\psi,\sigma)  \text{ for all } i \in [n].
	\]
Let $\Sigma(\Gamma, (B_i)_{i \in [n]}, \psi)$ denote the set of all mediated CPT Nash equilibria of $\tilde \Gamma = (\Gamma, (B_i)_{i \in [n]})$ with respect to a mediator distribution $\psi \in \Delta(B)$.
\end{definition}

%%%%%%%%%%%%%%%%%%%%%%%%%%%%%%%%%%%%%%%%%%%%%%%

%{\color{red} Since the notion of a mediated equilibrium may not be familiar to many people,
%it would be useful to give an example here - in particular, one that demonstrates that this is a nontrivial notion in the sense that the set of mediated corrrelated equilibria is neither empty nor the set of all distributions, and also that the set of mediated correlated CPT equilibria is 
%different from the set of mediated correlated equilibria. \color{black}}

For any mediator distribution $\psi \in \Delta(B)$, and any randomized strategy profile $\sigma$, let $\eta(\psi,\sigma) \in \Delta(A)$ be given by
\begin{equation}
\label{eq: eta_def}
	\eta(\psi,\sigma)(a) := \sum_{b \in B} \psi(b) \prod_{i \in [n]} \sigma_i(b_i)(a_i).
\end{equation}
Thus, $\eta(\psi,\sigma)$ gives the joint distribution over 
%\red
the 
%\black
action profiles of all 
%\red
the
%\black
players corresponding to the randomized strategy $\sigma$ and the mediator distribution $\psi$.

%%%%%%%%%%%%%%%%%%%%%%%%%%%%%%%%%%%%%%%%%%%%%%%

\begin{definition}
\label{def: mediated_cpt_corr_eq}
A probability distribution $\mu \in \Delta(A)$ is said to be a \emph{mediated CPT correlated equilibrium} of a game $\Gamma$ 
%iff 
%\red 
if 
%\black
there exist a signal system $(B_i)_{i \in [n]}$, a mediator distribution $\psi \in \Delta(B)$, and a mediated CPT Nash equilibrium $\sigma \in \Sigma(\Gamma, (B_i)_{i \in [n]}, \psi)$ such that $\mu = \eta(\psi,\sigma)$.
\end{definition}

%%%%%%%%%%%%%%%%%%%%%%%%%%%%%%%%%%%%%%%%%%%%%%%

%\red 
Consider an arbitrary mediated game
$\tilde{\Gamma} = (\Gamma, (B_i)_{i \in [n]})$ with
an arbitrary mediator distribution $\psi \in \Delta(B)$,
where $B = \prod_{i=1}^n B_i$. 
%\black
If all the players choose to ignore the signals sent by the mediator, then the corresponding randomized strategy profile $\sigma$ consists of constant functions 
%\com
$\sigma_i(b_i) \equiv \mu_i^*$.
%\cob 
%\cor
Further, 
%\blue 
as shown in Remark \ref{app: rem: allNash} in
Appendix \ref{app: eq_notions}, 
%\black
it follows from 
%\com
%Definitions~\ref{def: mediated_cpt_nash_eq} 
%and \ref{def: CPT_Nash_eq} 
%\cor
%and \ref{def: CPT_Nash_eq_real} 
%\red
Definitions~\ref{def: CPT_Nash_eq_real} 
%and \ref{def: CPT_Nash_eq} 
%\cor
and \ref{def: mediated_cpt_nash_eq}  
%\black
%\cob
that 
%$\sigma$ is a mediated CPT Nash equilibrium of the game $\tilde \Gamma$ with mediator distribution $\psi$, for any mediator distribution $\psi$, iff $\sigma$ is a CPT Nash equilibrium of the game $\Gamma$. 
%\com
the 
%mixed strategy 
%\red 
product probability distribution 
%\black
$\mu^* = \prod_{i \in [n]} \mu_i^*$ 
%\cob
%\bb
is a CPT Nash equilibrium of the game $\Gamma$ iff
$\sigma$ is a mediated CPT Nash equilibrium of the 
%\red 
mediated 
%\black
game $\tilde \Gamma$ with 
%\red 
respect to the 
%\black
mediator distribution $\psi$. 
%for any mediator distribution $\psi$.
%\color{red} 
%for the mediator distribution $\psi$ being used. 
%\color{black}
%\cob
%\red 
In particular, 
%\black
since every game $\Gamma$ has at least one CPT Nash equilibrium, we %have 
%\red 
see 
%\black
that every mediated game $\tilde \Gamma$ 
%with mediator distribution $\psi$ 
has at least one mediated CPT Nash equilibrium with respect to the mediator distribution $\psi$, 
%\red 
for any mediator distribution 
$\psi$. 
%\black
%\cob

%\cor
%Please note that I have reordered some of the paragraphs above.
%\cob

Let $D(\Gamma)$ denote the set of all mediated CPT correlated equilibria of a game $\Gamma$. 
By definition, $D(\Gamma)$ is the union over all signal systems $(B_i)_{i \in [n]}$
and mediator distributions $\psi \in \Delta(B)$ of 
$\{ \eta(\psi, \sigma) : \sigma \in \Sigma(\Gamma, (B_i)_{i \in [n]}, \psi)\}$.
%This characterizes the union of the sets of equilibria of all mediated games that can be generated from $\Gamma$.
When $B_i = A_i$ for all $i \in [n]$ and 
%\red 
$\sigma = (\sigma_1, \ldots, \sigma_n)$ is the 
deterministic strategy profile given, with an abuse of notation,
by 
%\black
$\sigma_i(b_i)(a_i) = \1\{b_i = a_i\}$,
%$\sigma \in \Sigma(\Gamma,(B_i)_{i \in [n]},\psi)$
%\bb
one can check,
%\red 
see Remark \ref{app: rem: CinD} in
Appendix \ref{app: eq_notions}, 
%\black
that $\sigma \in \Sigma(\Gamma,(A_i)_{i \in [n]},\psi)$
%\cob 
%if and only if 
%\red 
iff 
%\black
$\psi \in C(\Gamma)$. 
%Thus, we have $C(\Gamma) \subset D(\Gamma)$.
%\bb
In this case $\eta(\psi,\sigma) = \psi$ and so we have $C(\Gamma) \subset D(\Gamma)$.
%\cob
Under EUT, \citet{aumann1987correlated} proves that $D(\Gamma) = C(\Gamma)$. 
%This is sometimes referred to as the \emph{Bayesian-Nash revelation-principle}.
However, under CPT, this property, in general, does not hold true. Lemma~\ref{lem: mediated_corr_eq} shows how $D(\Gamma)$ compares with $C(\Gamma)$.

For any $i, a_i, \tilde a_i \in A_i$, let $C(\Gamma, i,a_i,\tilde a_i)$ denote the set of all probability vectors $\pi_{-i} \in \Delta(A_{-i})$ such that
%\begin{equation}
%\label{eq: best_reaction_ineq}
%	V_i(\{\pi_{-i}(a_{-i}),x_i(a_i,a_{-i})\}) \geq V_i(\{\pi_{-i}(a_{-i}),x_i(\tilde a_i,a_{-i})\}).
%\end{equation}
%\com
\begin{equation}
\label{eq: best_reaction_ineq}
    V_i(L_i(\pi_{-i},a_i)) \geq V_i(L_i(\pi_{-i},\tilde a_i)).
\end{equation}
%\cob
%Since $V_i^{r_i}(\{\pi_{-i}(a_{-i}),u_i(a_i,a_{-i})\})$ is a continuous function of $\pi_{-i}$, the set $C(\Gamma, i,a_i,\tilde a_i)$ is closed.
It is clear from the definition of CPT correlated equilibrium that, for a joint probability distribution $\mu \in C(\Gamma)$, provided $\mu_i(a_i) > 0$, the probability vector 
%$\pi_{-i}(\cdot)  = \mu(\cdot|a_i) \in \Delta(A_{-i})$ 
%\bb
$\pi_{-i}(\cdot)  = \mu_{-i}(\cdot|a_i) \in \Delta(A_{-i})$ 
%\cob
should belong to $C(\Gamma, i,a_i,\tilde a_i)$ for all $\tilde a_i \in A_i$. Let
\[
	C(\Gamma, i,a_i) := \cap_{\tilde a_i \in A_i} C(\Gamma, i,a_i,\tilde a_i).
\]
Now, for all $i$, define a subset $C(\Gamma, i) \subset \Delta(A)$, as follows:
%\[
%	C(\Gamma, i) := \{\mu \in \Delta(A) | \mu(\cdot|a_i) \in C(\Gamma, i,a_i), \forall a_i \in \supp \l(\mu_i\r) \}.
%\]
%\bb
\[
	C(\Gamma, i) := \{\mu \in \Delta(A) | \mu_{-i}(\cdot|a_i) \in C(\Gamma, i,a_i), \forall a_i \in \supp \l(\mu_i\r) \}.
\]
%\cob
Note that, since 
%$V_i(\{\pi_{-i}(a_{-i}),x_i(a_i,a_{-i})\})$
%\com
$V_i(L_i(\pi_{-i},a_i))$
%\cob
 is a continuous function of $\pi_{-i}$, the sets $C(\Gamma, i,a_i,\tilde a_i)$, $C(\Gamma,i,a_i)$ and  $C(\Gamma,i)$ are all closed.
%%%%%%%%%%%%%%%%%%%%%%%%%%%%%%%%%%%%%%%%%%%%%%%

\begin{lemma}
\label{lem: mediated_corr_eq}
For any game $\Gamma$, we have
\begin{enumerate}[(i)]
	\item 
    %\bb
    For all $i \in [n]$,
    %\cob 
    $\conv \l(C(\Gamma, i)\r) = \{\mu \in \Delta(A) | \mu_{-i}(\cdot|a_i) \in \conv \l(C(\Gamma, i,a_i)\r), \forall a_i \in \supp \l(\mu_i\r) \}$,
	\item $C(\Gamma) = \cap_{i \in [n]} C(\Gamma, i)$, and
	\item $D(\Gamma) = \cap_{i \in [n]} \conv(C(\Gamma, i))$.
\end{enumerate}
where $\conv(S)$ denotes the convex hull of a set $S$.
\end{lemma}

%%%%%%%%%%%%%%%%%%%%%%%%%%%%%%%%%%%%%%%%%%%%%%%%%%

\begin{proof}
%[Proof of lemma \ref{lem: mediated_corr_eq}]
%\bb
Fix $i \in [n]$.
%\cob
	Note that, since the sets $C(\Gamma,i)$ and $C(\Gamma,i,a_i)$ 
    %\bb
    for each $a_i \in A_i$
    %\cob
    are closed, the convex hulls of these sets are closed.
	 Suppose $\mu = \lambda \mu^1 + (1 - \lambda) \mu^2$ where $\mu^1,\mu^2 \in C(\Gamma,i)$ and $0 < \lambda < 1$. If $a_i \in \supp (\mu_i)$, then one of the following three cases holds:

%\cor
	\textbf{Case 1} [$a_i \in \supp (\mu^1_i)$, $a_i \in \supp (\mu^2_i)$]. Then, $\mu^1_{-i}(\cdot|a_i), \mu^2_{-i}(\cdot|a_i) \in C(\Gamma,i,a_i)$ and we have,
	\[
		\mu_{-i}(\cdot|a_i) = \frac{\lambda \mu^1_i(a_i) \mu^1_{-i}(\cdot|a_i) + (1 - \lambda) \mu^2_i(a_i) \mu^2_{-i}(\cdot|a_i)}{\lambda \mu^1_i(a_i) + (1 - \lambda) \mu^2_i(a_i)}.
	\]
	Let $\theta = (\lambda \mu^1_i(a_i))/(\lambda \mu^1_i(a_i) + (1 - \lambda) \mu^2_i(a_i))$. Then $\mu_{-i}(\cdot|a_i) = \theta \mu^1_{-i}(\cdot|a_i) + (1 - \theta) \mu^2_{-i}(\cdot|a_i)$ and hence $\mu_{-i}(\cdot|a_i) \in \conv \l(C(\Gamma,i,a_i) \r)$. 
    %Thus $\mu$ belongs to the set on the right hand side of the equation in (i).

	\textbf{Case 2} [$a_i \in \supp (\mu^1_i)$, $a_i \notin \supp (\mu^2_i)$]
    Here $\mu_{-i}(\cdot|a_i) = \mu^1_{-i}(\cdot|a_1)$. Hence $\mu_{-i}(\cdot|a_i) \in C(\Gamma,i,a_i)$.
    
    \textbf{Case 3} [$a_i \notin \supp (\mu^1_i)$, $a_i \in \supp (\mu^2_i)$] This can be handled similarly to case 2. 
    
    Also, the above argument can be easily extended to when $\mu$ is a convex combination of any finite number of distributions. 
    Since all our sets are compact subsets of finite dimensional Euclidean spaces, Caratheodory's theorem applies, and hence we need to consider only finite convex combinations. 
    
    This shows that the set on the 
    %LHS 
    %\red 
    left hand side 
    %\black
    is contained in the set on the 
    %RHS 
    %\red 
    right hand side 
    %\black
    of the equation in (i)
    %\bb
    for the given fixed $i \in [n]$.
    %\cob
    %\cob

	To prove 
	%the other inclusion, 
	%\red 
    the inclusion in the other direction, 
    %\black
	%\red 
    fix $i \in [n]$ and 
    %\black
	let $\mu$ belong to the set on the 
	%RHS 
	%\red 
    right hand side 
    %\black
	of the equation in (i).
    %\bb
    %Fix $i \in [n]$.
    %\cob
    If $a_i \in \supp (\mu_i)$, then $\mu_{-i}(\cdot|a_i)$
    is a linear combination of $|A_{-i}|$ joint distributions (allowing repetitions),
%\red 
call them 
%\black    
	\[
		\zeta_{-i,a_i}^1, \dots, \zeta_{-i,a_i}^{m_i}, \dots, \zeta_{-i,a_i}^{|A_{-i}|} \in C(\Gamma,i,a_i),
	\] 
	with coefficients 
	$\theta_{i,a_i}^{m_i}, 1 \leq m_i \leq |A_{-i}|$ respectively
    %\bb
    (where 
    %$0 < \theta_{i,a_i}^{m_i} \leq 1$ 
    %\red $0 < \theta_{i,a_i}^{m_i} < 1$ \black
    %\blue 
    $0 < \theta_{i,a_i}^{m_i} \leq 1$ 
    %\black
    %\footnote{
    %\blue We will have to allow for the possibility that $\theta_{i,a_i}^{m_i} = 1$ in case we have the trivial situation $|A_{-i}| = 1$.
    %}
    for all $1 \leq m_i \leq |A_{-i}|$ 
    %\red 
    can be ensured because we allow repetitions). 
    %\black
    %\cob
    %\cor
	%Note that, since $C(\Gamma)$ is nonempty, the set $C(\Gamma,i)$ is nonempty for each $i$. The set $C(\Gamma, i)$ can be constructed from the sets $\{C(\Gamma,i,a_i),a_i \in A_i\}$ as follows: let $p_{a_i} \in C(\Gamma, i,a_i)$ for all $a_i \in A_i$ such that $C(\Gamma,i,a_i) \neq \phi$, let $q_i \in \Delta(A_i)$ be such that $q_i(a_i) =  0$ for all $a_i  \in A_i$ such that $C(\Gamma,i,a_i) = \emptyset$, and define a joint probability distribution $\bar \mu \in \Delta(A)$ by $\bar \mu(a_i,a_{-i}) = q_i(a_i) p_{a_i}(a_{-i})$ if $C(\Gamma,i,a_i) \neq \phi$ and $\bar \mu(a_i,a_{-i}) = 0$ otherwise.
	%Then $\bar \mu \in C(\Gamma,i)$, and for every $\bar \mu \in C(\Gamma,i)$, the corresponding $q_i(a_i) = \bar \mu_i(a_i)$ for all $a_i \in A_i$  and $p_{a_i} = \bar \mu_{-i}(\cdot|a_i)$ for all $a_i \in A_i$ with $C(\Gamma,i,a_i) \neq \phi$.
	For each $\zeta_{-i,a_i}^{m_i}$, construct a distribution 
    %\com
    $\mu_{i,a_i}^{m_i} \in \Delta(A)$ 
    %\cob
    by %$\mu_{i,a_i}^{m_i}(b_i, b_{-i}) = 1(b_i = a_i)
    %\zeta_{-i,a_i}^{m_i}(b_i)$.
    %\bb
    %$\mu_{i,a_i}^{m_i}(b_i, b_{-i}) = 1(b_i = a_i)
    %\zeta_{-i,a_i}^{m_i}(b_{-i})$.
    %\cob
    %\com
    $\mu_{i,a_i}^{m_i}(\tilde a_i, \tilde a_{-i}) = 1\{\tilde a_i = a_i\}
    \zeta_{-i,a_i}^{m_i}(\tilde a_{-i})$.
    %\cob
    %taking $p_{a_i} = \zeta_{-i,a_i}^{m_i}$ and $q_i = e_{a_i}$. 
    %From the observation above, we have 
    Then
    $\mu_{i,a_i}^{m_i} \in C(\Gamma,i)$. 
	Let $\lambda_{i,a_i}^{m_i} := \mu_i(a_i) \theta^{m_i}_{i,a_i}$, for all $i, m_i, a_i$ such that $\mu_i(a_i) > 0$. One can now check that 
%$\mu = \sum_{i,m_i,a_i} \lambda_{i,a_i}^{m_i} \mu_{i,a_i}^{m_i}$. 
    %\bb
    $\mu = \sum_{m_i,a_i} \lambda_{i,a_i}^{m_i} \mu_{i,a_i}^{m_i}$ for the given fixed $i \in [n]$.
    %\cob
    Thus $\mu \in \conv \l(C(\Gamma,i)\r)$.
    %\cob

	Statement (ii) follows directly from the definition of CPT correlated equilibrium.
	
    %\cor
	For statement (iii), let $\mu \in \Delta(A)$ be such that $\mu \in \conv(C(\Gamma, i))$ for all $i$. 
    %\cob
    For any $a_i$ such that $\mu_i(a_i) > 0$, by (i), we have, $\mu_{-i}(\cdot|a_i) \in \conv \l(C(\Gamma, i,a_i)\r)$.
%	We first observe that for each player $i$, for all actions $a_i$, $C(\Gamma,i,a_i)$, and hence $\conv \l(C(\Gamma, i,a_i)\r)$, is a compact subset of a $(|A_{-i}|-1)$-dimensional Euclidean space  and hence by Caratheodory's theorem $\mu_{-i}(\cdot|a_i)$
As above, let $\mu_{-i}(\cdot|a_i)$ be a convex combination of $|A_{-i}|$ joint distributions (allowing repetitions),
%\red 
call them 
%\black
	\[
		\zeta_{-i,a_i}^1, \dots, \zeta_{-i,a_i}^{m_i}, \dots, \zeta_{-i,a_i}^{|A_{-i}|} \in C(\Gamma,i,a_i),
	\] 
	with coefficients 
	$\theta_{i,a_i}^{m_i}, 1 \leq m_i \leq |A_{-i}|$ respectively
    %\bb
    (where 
    %$0 < \theta_{i,a_i}^{m_i} \leq 1$ 
    %\red $0 < \theta_{i,a_i}^{m_i} < 1$ \black
    %\blue 
    $0 < \theta_{i,a_i}^{m_i} \leq 1$ 
    %\black
    for all $1 \leq m_i \leq |A_{-i}|$
    %\red 
    can be ensured because we allow repetitions). 
    %\black
    %\cob
For all $i$, let $B_i := A_i \times M_i$, 
    where $M_i := \{1,\dots,|A_{-i}|\}$.
	% and let $\pi_i^{A_i}(\cdot)$ and $\pi_i^{M_i}(\cdot)$ be the projection maps from $B_i$ to $A_i$ and $M_i$ respectively. 
    Let the mediator distribution be given by
% \begin{align}
% \label{eq: rev_psi}
% 	\psi(b_1,\dots,b_n) &= \mu(\pi_i^{A_i}(b_i), i \in N) \nonumber\\
% 	&\times \frac{\prod\limits_{i=1}^n \l\{\theta_i^{\pi_i^{M_i}(b_i)}(\pi_i^{A_i}(b_i)) \zeta_{-i}^{\pi_i^{M_i}(b_i)}(\pi_i^{A_i}(b_i))[\pi_j^{A_j}(b_j), j \in N \back i]\r\}}{\sum\limits_{\substack{(\tilde b_i, i \in N) \\
% 	 \pi_i^{A_i}(\tilde b_i) = \pi_i^{A_i}(\tilde b_i)} } \prod\limits_{i=1}^n \l\{\theta_i^{\pi_i^{M_i}(\tilde b_i)}(\pi_i^{A_i}(\tilde b_i)) \zeta_{-i}^{\pi_i^{M_i}(\tilde b_i)}(\pi_i^{A_i}(\tilde b_i))[\pi_j^{A_j}(\tilde b_j), j \in N \back i]\r\} }.
% \end{align}
% Let $\pi_i^{A_i}(b_i) = a_i, \pi_i^{M_i}(b_i) = m_i$ for all $i$. Then equation \eqref{eq: rev_psi} can be rewritten as
%\red
\begin{align}
\label{eq: rev_psi}
	\psi\l((a_1,m_1),\dots,(a_n,m_n)\r) =
    \begin{cases} 
    \frac{\mu(a)\prod_{i=1}^n \l\{\theta_{i,a_i}^{m_i} \zeta_{-i,a_i}^{m_i}(a_{-i})\r\}} {\sum_{\tilde m_i, i \in [n] } \prod\limits_{i=1}^n \l\{\theta_{i,a_i}^{\tilde m_i} \zeta_{-i,a_i}^{\tilde m_i}(a_{-i})\r\} },
    & \mbox{ if $\mu(a) > 0$},\\
    0, & \mbox{ otherwise}.
    \end{cases}
\end{align}
%\black
%\bb
It is useful to note that 
\begin{equation}		\label{eq:denom}
\sum_{\tilde m_i, i \in [n] } \prod\limits_{i=1}^n \l\{\theta_{i,a_i}^{\tilde m_i} \zeta_{-i,a_i}^{\tilde m_i}(a_{-i})\r\} =
\prod_{i=1}^n \mu_{-i}(a_{-i}|a_i),
 \end{equation}
and that, for every $i \in [n]$,
\begin{equation}		\label{eq:marginal}
\psi_i((a_i,m_i)) := \sum_{(a_j,m_j), j \in [n] \back i} 
\psi\l((a_1,m_1),\dots,(a_n,m_n)\r) = \mu_i(a_i) \theta_{i,a_i}^{m_i}.
\end{equation}
%\cob
Let the strategy for each player $i$ be
%\cor
%\begin{align}
%\label{eq: rev_sigma}
%	\sigma_i(a_i,m_i) = e_{a_i}.
%\end{align}
%\bb
%\red
\begin{align}
\label{eq: rev_sigma}
	\sigma_i(a_i,m_i)(\tilde{a}_i) = \begin{cases}
    1, & \mbox{ if $\tilde{a}_i = a_i$},\\
    0, & \mbox{ otherwise}.
    \end{cases}
\end{align}
%\black
%\cob
%\cob
From equations \eqref{eq: eta_def}, \eqref{eq: rev_psi} and \eqref{eq: rev_sigma} 
%\cor
we have
%\begin{align*}
%	\eta(\psi,\sigma)(a) &= 
    %\sum_{b \in B} \psi(b) \prod_{i \in N} \sigma_i(b_i)(a_i) \\
%    \sum_{(\tilde{a}_i, m_i) \in B_i, i \in N} \psi\l((\tilde{a}_1,m_1),\dots,(\tilde{a}_n,m_n)\r) \prod_{i \in N} \sigma_i\l((\tilde{a}_i,m_i)\r)(a_i) \\
%	&= \mu(a) \times \sum_{m_i, i \in N}  \frac{\prod\limits_{i=1}^n \l\{\theta_{i,a_i}^{m_i} \zeta_{-i,a_i}^{m_i}(a_{-i})\r\}}{\sum_{\tilde m_i, i \in N } \prod\limits_{i=1}^n \l\{\theta_{i,a_i}^{\tilde m_i} \zeta_{-i,a_i}^{\tilde m_i}(a_{-i})\r\} } \\
%	&= \mu(a).
%\end{align*}
%\bb
%\red
\begin{align*}
	\eta(\psi,\sigma)(a) &= 
    %\sum_{b \in B} \psi(b) \prod_{i \in N} \sigma_i(b_i)(a_i) \\
    \sum_{(\tilde{a}_i, m_i) \in B_i, i \in [n]} \psi\l((\tilde{a}_1,m_1),\dots,(\tilde{a}_n,m_n)\r) \prod_{i \in [n]} \sigma_i\l((\tilde{a}_i,m_i)\r)(a_i) \\
    &= \sum_{m_i, i \in [n]} \psi\l((a_1,m_1),\dots,(a_n,m_n)\r) \\
	&= \mu(a) \sum_{m_i, i \in [n]}  \frac{\prod\limits_{i=1}^n \l\{\theta_{i,a_i}^{m_i} \zeta_{-i,a_i}^{m_i}(a_{-i})\r\}}{\sum_{\tilde m_i, i \in [n] } \prod\limits_{i=1}^n \l\{\theta_{i,a_i}^{\tilde m_i} \zeta_{-i,a_i}^{\tilde m_i}(a_{-i})\r\} } \\
	&= \mu(a).
\end{align*}
%\black
%\cob
%\cob
%Let $\pi_i^{A_i}(b_i) = a_i, \pi_i^{M_i}(b_i) = m_i$. 
From equations \eqref{eq: tilde_mu}, \eqref{eq: rev_psi},
\eqref{eq:denom}, \eqref{eq:marginal} and \eqref{eq: rev_sigma} 
%\cor
we have
%\begin{align*}
%	\tilde \mu_{-i}(a_{-i} | (a_i,m_i)) &= 
    %\sum_{b_{-i} \in B_{-i}} \psi_{-i}(b_{-i}|b_i) \prod_{j \in N \back i} \sigma_j(b_j)(a_j) \\
 %   \sum_{(\tilde{a}_j, m_j) \in B_j, j \in N \back i} \psi_{-i}(\l((\tilde{a}_j, m_j) \in B_j, j \in N \back i\r)|(a_i,m_i)) \prod_{j \in N \back i} \sigma_j((\tilde{a}_j,m_j))(a_j) \\
	%& = \sum_{\substack{b_j, j \in N\back i \\ \pi_j^{A_j}(b_j) = a_j } }  \psi_{-i}(b_{-i}|b_i)\\
	%&\propto \mu(a) \times \frac{ \l\{\theta_{i,a_i}^{m_i} \zeta_{-i,a_i}^{m_i}(a_{-i})\r\}}{\sum\limits_{m_i} \l\{\theta_{i,a_i}^{m_i} \zeta_{-i,a_i}^{m_i}(a_{-i})\r\} }\\
%	&\propto  \mu(a) \times \frac{ \l\{\theta_{i,a_i}^{m_i} \zeta_{-i,a_i}^{m_i}(a_{-i})\r\}}{ \mu_{-i}(a_{-i}|a_i) }\\
%	& = \l\{\mu_i(a_i) \theta_{i,a_i}^{m_i}\r\} \zeta_{-i,a_i}^{m_i}(a_{-i}).
%\end{align*}
%\bb
\begin{align*}
	\tilde \mu_{-i}(a_{-i} | (a_i,m_i)) &= 
    %\sum_{b_{-i} \in B_{-i}} \psi_{-i}(b_{-i}|b_i) \prod_{j \in N \back i} \sigma_j(b_j)(a_j) \\
    \sum_{(\tilde{a}_j, m_j) \in B_j, j \in [n] \back i} \psi_{-i}(\l((\tilde{a}_j, m_j), j \in [n] \back i\r)|(a_i,m_i)) \prod_{j \in [n] \back i} \sigma_j((\tilde{a}_j,m_j))(a_j) \\
    &= \sum_{m_j, j \in [n] \back i} \psi_{-i}(\l((a_j, m_j), j \in [n] \back i\r)|(a_i,m_i)) \\
    &= \frac{\sum_{m_j, j \in [n] \back i} \psi\l((a_1,m_1),\dots,({a}_n,m_n)\r)}{\psi_i((a_i,m_i))}\\
	%& = \sum_{\substack{b_j, j \in N\back i \\ \pi_j^{A_j}(b_j) = a_j } }  \psi_{-i}(b_{-i}|b_i)\\
	%&\propto \mu(a) \times \frac{ \l\{\theta_{i,a_i}^{m_i} \zeta_{-i,a_i}^{m_i}(a_{-i})\r\}}{\sum\limits_{m_i} \l\{\theta_{i,a_i}^{m_i} \zeta_{-i,a_i}^{m_i}(a_{-i})\r\} }\\
	%&\propto  \mu(a) \times \frac{ \l\{\theta_{i,a_i}^{m_i} \zeta_{-i,a_i}^{m_i}(a_{-i})\r\}}{ \mu_{-i}(a_{-i}|a_i) }\\
	%& = \l\{\mu_i(a_i) \theta_{i,a_i}^{m_i}\r\} \zeta_{-i,a_i}^{m_i}(a_{-i}).
    & = \zeta_{-i,a_i}^{m_i}(a_{-i}).
\end{align*}
%\cob
%\cob
%Thus 
%\cor
%$\tilde \mu_{-i}(\cdot | (a_i,m_i)) = \zeta_{-i,a_i}^{m_i}$, 
%\cob
%and we have 
%\cor
%\bb
Thus we have
%\cob
$\tilde \mu_{-i}(\cdot|(a_i,m_i)) \in C(\Gamma,i,a_i)$.
%\cob
Hence $\mu \in D(\Gamma)$.
%establishing 
%\bb
We have established that 
%\cob
$\cap_{i \in N} \conv(C(\Gamma, i)) \subset D(\Gamma)$.

%\cor
%To prove the other direction of statement (iii), 
%\bb
For the other direction of statement (iii), 
%\cob
%\cob
let $\mu \in D(\Gamma)$. Then there exists a signal system $(B_i)_{i \in [n]}$, a mediator distribution $\psi \in \Delta(B)$, and a mediated CPT Nash equilibrium $\sigma \in \Sigma(\Gamma, (B_i)_{i \in [n]}, \psi)$ such that $\mu = \eta(\psi,\sigma)$. 
%\bb
Fix $i \in [n]$.
%\cob
For 
%\bb
$b_i \in \supp(\psi_i)$ and
$a_i \in \supp \l(\sigma_i(b_i)\r)$, we have 
$\tilde \mu_{-i}(\cdot|b_i) \in C(\Gamma,i,a_i)$, 
from equations \eqref{eq: best_response} and \eqref{eq: best_reaction_ineq}. 
%\cor
%This means, from equations \eqref{eq: tilde_mu} and \eqref{eq: eta_def}, that we have 
%\bb
But
%\com 
\[
\mu_{-i}(\cdot|a_i) = \sum_{b_i \in \supp(\psi_i)} \frac{\psi_i(b_i)\sigma_i(b_i)(a_i)}{\mu_i(a_i)} \tilde \mu_{-i}(\cdot|b_i).
\]
%\cob
Hence 
$\mu_{-i}(\cdot|a_i) \in \conv \l(C(\Gamma,i,a_i)\r)$.
%\cob
%and hence 
%\bb
Since this holds for all $i \in [n]$, we have
%\cob
$\mu = \eta(\psi,\sigma) \in \cap_{i \in [n]} \conv(C(\Gamma, i))$. This completes the proof.
\end{proof}

%%%%%%%%%%%%%%%%%%%%%%%%%%%%%%%%%%%%%%%%%%%%%%%%%%%%%%%%

For the $2$-person game
%\com
$\Gamma^*$
%\cob
 in example \ref{ex: nonconvergence_calib}, we observed that the set 
 %\com
 $C(\Gamma^*)$
 %\cob
  is non-convex and hence 
%\com
  $C(\Gamma^*) \neq D(\Gamma^*)$. 
%\cob
%\cor
%Need to clarify the preceding sentence. $C(\Gamma)$ is not really discussed in that example.
%\cob
If $\Gamma$ is a $2\times 2$ game, i.e., a game with $2$ players, each having two actions, 
%\cor
and both behaving according to CPT, then \citet{phade2019geometry} prove that the sets $C(\Gamma,i)$, corresponding to both these players are convex, 
%\cob
and hence also the set $C(\Gamma)$. From Lemma~\ref{lem: mediated_corr_eq}, we have the following result, having the flavor of the revelation principle:

\begin{proposition}
	If $\Gamma$ is a $2\times 2$ game, then the set of all CPT correlated equilibria is equal to the set of all mediated CPT correlated equilibria.
\end{proposition}

%\cor
In the context of mediated games, a strategy $\sigma_i$ for player $i$ 
%\cob
is said to be \emph{pure} if $\supp \l(\sigma_i\r)$ is singleton and a strategy profile $\sigma = (\sigma_i)_{i \in [n]}$ is said to be a \emph{pure strategy profile} if each $\sigma_i$ is a pure strategy.

\begin{remark}
\label{rem: rev_pure_strat}
From the proof of Lemma~\ref{lem: mediated_corr_eq}, we observe that for any $\mu \in D(\Gamma)$, there exists a signal system $(B_i)_{i \in [n]}$ 
%\bb
(of size $|B_i| = |A_i| \times |M_i| = |A|$), 
%\cob
a mediator distribution $\psi \in \Delta(B)$, and a mediated CPT Nash equilibrium $\sigma \in \Sigma(\Gamma, (B_i)_{i \in [n]}, \psi)$ such that $\mu = \eta(\psi,\sigma)$ where $\sigma$ is a pure strategy profile.
%\footnote{\red The remark that originally followed this one 
%has been moved to the end of Appendix~\ref{app: eq_notions}
%as suggested. \black}
\end{remark}

%%%%%%%%%%%%%%%%%%%%%%%%%%%%%%%%%%%%%%%%%%%%%%%

 %One of the attractive properties about Aumann's correlated equilibrium is that it has a nice geometry, namely, it is a convex polytope. It has been observed that the geometry of CPT CE as defined by Keskin is much more complicated. In particular, it is not convex. In fact it could also be disconnected (Allerton paper). Another important feature of Aumann's correlated equilibrium is that the joint empirical distribution of action profiles in a repeated game converges to the set of correlated equilibrium of the stage game under certain learning strategies, calibrated learning being one of them. 
 %The set $D(\Gamma)$ preserves these two properties. Indeed, lemma \ref{lem: mediated_corr_eq} shows that $D(\Gamma)$ is convex. In the next section we show that under calibrated learning, the joint empirical distribution of action profiles in a repeated game converges to $D(\Gamma)$.

%%%%%%%%%%%%%%%%%%%%%%%%%%%%%%%%%%%%%%%%%%%%%%%
%%%%%%%%%%%%%%%%%%%%%%%%%%%%%%%%%%%%%%%%%%%%%%%

%!TEX root = ../main.tex

\subsection*{Calibrated learning to mediated CPT correlated equilibrium}
\label{sec: calib_mediated}

Let $\eda^t$ denote the empirical joint distribution of the action play up to step $t$. Formally,
 \[
 	\eda^t = \frac{1}{t} \sum_{\tau = 1}^t e_{a^\tau},
 \]
 where $e_{a^t}$ is an $|A|$-dimensional vector with its $a^t$-th component equal to $1$ and the rest $0$.
 %\cor
%\bb
We write the coordinates of $\eda^t$ as $(\eda^t(a), a \in A)$. 
For each $i \in [n]$, we write $\eda^t_i := (\eda^t_i(a_i), a_i \in A_i)$
for the empirical distribution of the actions of player $i$. Thus $\eda^t_i$ is
the $i$-th marginal distribution corresponding to $\eda^t$. Similarly,
for $i \in [n]$, $\eda^t_{-i} := (\eda^t_{-i}(a_{-i}|a_i), a \in A)$ are conditional
distributions corresponding to $\eda^t$, 
where $\eda_{-i}^t(a_{-i}|a_i)$ is defined to be 0 when $\xi^t(a) = 0$.
%\cob

Let the distance between a vector $x$ and a set $X$ 
%\color{red}
in the same Euclidean space
%\color{black}
be given by
\[
	d(x,X) = \inf_{x' \in X} \Vert x - x' \Vert,
\]
where $\Vert x \Vert$ denotes the standard Euclidean norm of $x$.
%\cob
We say that a sequence $(x^t,t \geq 1)$ converges to a set $X$ if the following holds:
\[
	\lim_{t \to \infty} d(x^t,X) = 0.
\]
%and the distance between two sets $X$ and $X'$ be given by
%\[
%	d(X,X') = \sup_{x' \in X'} d(x',X).
%\]

%%%%%%%%%%%%%%%%%%%%%%%%%%%%%%%%%%%%%%%%%%%%%%%

\begin{theorem}
	\label{thm: calibrated_CPTeq}
	Assume that 
	%all the players use calibrated forecasters to predict assessments and choose the best reaction to these assessments at every step. 
	%\red 
	the assessment schemes and best reaction maps of the players are such that if each player at each step plays the best reaction to her assessment then each player is calibrated with respect to the sequence of action profiles of the other players. 
	%\black
	Then the 
	%joint empirical distribution 
	%\color{red} 
	empirical joint distribution 
	%\color{black}
	of action play $\eda^t$ converges to the set of mediated CPT correlated equilibria.
	%, i.e.,
	%\[
	%	d(\eda^t,D(\Gamma)) \to 0.
	%\]
\end{theorem}

%%%%%%%%%%%%%%%%%%%%%%%%%%%%%%%%%%%%%%%%%%%%%%%%%%%%%%%%%%%%%%%%%%%%%%%%%%%%%%%%%%%%%

\begin{proof}
%[Proof of theorem \ref{thm: calibrated_CPTeq}]
	Consider the sequence of empirical distributions $\eda^t$. Since the simplex $\Delta(A)$ of all joint distributions over action profiles is a compact set, every 
	%\red 
	such %\black 
	sequence has a convergent subsequence. Thus, it is enough to show that the limit of any convergent subsequence of $\eda^t$ is in $D(\Gamma)$. Let $\eda^{t_k}$ be such a convergent subsequence and denote its limit by $\hat \eda$.

%\cor
	Let $a_i$ be an action of  player $i$
 %   \cob
    such that $\hat \eda_i(a_i) > 0$. 
	Let $R_i(a_i) \subset \Delta(A_{-i})$ be the set of all joint distributions $\mu_{-i}$ for which action $a_i$ is player $i$'s %best reaction under CPT preferences.
	%\red 
	best reaction. %\black
	%Let $\hat R_i(a_i)$ be the set of assessments for which $a_i$ is the best response , i.e.,
	%\begin{align*}
	%	\hat R_i(a_i) :=
	%	  	\l\{\mu_{-i} \in \Delta (A_{-i}) \right.| &\forall \tilde a_i \in A_i,\\
	%	  	&\left. V_i^{r_i}\l(L(\mu_{-i},a_i)\r) \geq V_i^{r_i}\l(L(\mu_{-i},\tilde a_i)\r) \r\}.
	%\end{align*}
	%Note that $\hat R_i(a_i)$ is closed and hence also compact and $R_i(a_i) \subset \hat R_i(a_i)$. 
	Note that $R_i(a_i)$ forms a partition of the simplex $\Delta(A_{-i})$. 
	%\com
	%Recall that 
	%\red 
	Let 
	%\black
	$\mu_{-i}^t \in \Delta(A_{-i})$ 
	%denotes 
	%\red 
	denote %\black
	player $i$'s assessment at step $t$, and 
	%\cob
	let $Q_i^{t}$ denote the set of assessments made by her up to step $t$.
	Since $\hat \eda_i(a_i) > 0$, there exists an integer $k_0 \geq 1$ and an $\epsilon > 0$ such that, for all $k \geq k_0$,
	%\red 
	we have 
	%\black
	$\eda_i^{t_k}(a_i) > \epsilon$.
	For all $k \geq k_0$, we have
	\begin{align*}
		\eda^{t_k}_{-i}(a_{-i} | a_i) \eda_i^{t_k}(a_i) t_k  &= \sum_{\substack{\tau \leq t_k \\ \text{s.t. } \mu^\tau_{-i} \in R_i(a_i)}} \1 \{a_{-i}^\tau = a_{-i} \}\\
		&=   \sum_{q \in R_i(a_i) \cap Q_i^{t_k}} \sum_{\substack{\tau \leq t_k \\ \text{s.t. } \mu_{-i}^\tau = q } } \1 \{a_{-i}^\tau = a_{-i} \}\\
		&=  \sum_{q \in R_i(a_i)\cap Q_i^{t_k}} \rho(q,a_{-i},t_k) N(q,t_k)\\
		&=  \sum_{q \in R_i(a_i)\cap Q_i^{t_k}} q(a_{-i}) N(q,t_k) \\
		&+  \sum_{q \in R_i(a_i)\cap Q_i^{t_k}} \l(\rho(q,a_{-i},t_k) - q(a_{-i}) \r) N(q,t_k).
	\end{align*}
	Using
	%\red
	\[
		\eda_i^{t_k}(a_i) t_k = \sum_{q \in R_i(a_i) \cap Q_i^{t_k}} N(q, t_k),
	\]
	%\black
	we get, for all $k \geq k_0$,
	\begin{align*}
		\eda^{t_k}_{-i}(a_{-i} | a_i) &= \frac{\sum_{q \in R_i(a_i)\cap Q_i^{t_k}} q(a_{-i}) N(q,t_k)}{\sum_{q \in R_i(a_i) \cap Q_i^{t_k}} N(q, t_k)} \\
		&+  \frac{1}{\eda_i^{t_k}(a_i)} \sum_{q \in R_i(a_i)\cap Q_i^{t_k}} \l(\rho(q,a_{-i},t_k) - q(a_{-i}) \r) \frac{N(q,t_k)}{t_k}.
	\end{align*}
	Since 
	%the forecast being used by 
	player $i$ is calibrated
	%\red 
	with respect to the sequence of action profiles of her opponents, 
	%\black
	the second term in the last expression goes to zero as $k \to \infty$ (here, we use the fact that $\eda_i^{t_k}(a_i) > \epsilon > 0$ for all $k \geq k_0$). Further, we have, for all $k \geq 1$,
%    \cor
	\[
		\frac{ \sum_{q \in R_i(a_i)\cap Q_i^{t_k}} q N(q,t_k)}{\sum_{q \in R_i(a_i)\cap Q_i^{t_k}} N(q,t_k) } \in \conv \l( R_i(a_i)\r).
	\]
	Taking the limit as $k \to \infty$ we have, $\hat \eda_{-i}(\cdot |a_i) \in \bar{\conv} \l(R_i(a_i)\r)$, where $\bar{\conv}(\cdot)$ denotes the closed convex hull.
    Note that $R_i(a_i) \subset C(\Gamma,i,a_i)$
    and $C(\Gamma,i,a_i)$ is closed. 
 %   \cob
    Thus $\hat \eda_{-i}(\cdot |a_i) \in \conv \l(C(\Gamma,i,a_i)\r)$ for all $a_i \in A_i$ such that $\hat \eda_i(a_i) > 0$. 
    By Lemma~\ref{lem: mediated_corr_eq}, we have $\hat \eda \in \conv \l(C(\Gamma,i)\r)$, and since this is true for all players $i$, we have $\hat \eda \in D(\Gamma)$.
\end{proof}

%%%%%%%%%%%%%%%%%%%%%%%%%%%%%%%%%%%%%%%%%%%%%%%%%%%%%%%%%%%%%%%%%%%%%%%%%%%%%%%%%%%%%
\begin{remark}
\label{rem: calib_conv_individual}
In the proof above we, in fact, prove the following stronger statement: 
If player $i$'s assessments are calibrated 
%\red 
with respect to the sequence of action profiles of her opponents 
%\black
and she 
%chooses best reaction 
%\bb
chooses the best reaction
%\cob 
to 
%these 
%\red 
her 
%\black
assessments at every step, then the joint empirical distribution of action play converges to the set $\conv \l(C(\Gamma,i)\r)$.
\end{remark}

%\red NOTE: There is a problem with the text starting at this point
%and till the end of the "proof" of Proposition 3.10. \black

%\cor
%This establishes the convergence of the empirical distribution of action play 
%\cob
%to the set of mediated CPT correlated equilibria 
%\blue
%It is crucial that each player $i$ are calibrated with respect to the sequence of action profiles of her opponents.
Now the question remains whether 
%the players 
%\red 
each player $i$ 
%\black
can make assessments that are guaranteed to be calibrated no matter what strategies her opponents use.
But this has nothing to do whether the players have 
%CPT preferences or EUT preferences 
%\red 
EUT or CPT preferences, 
%\black
and has been answered 
%\red 
in the affirmative 
%\black
by \citep[Theorem 3]{foster1997calibrated}.
%In particular, 
%\red 
To be precise, 
%\black
%it has been proved that there exists 
%a randomized forecasting scheme 
%for 
%player $i$ 
%such that, no matter what strategies 
%$\strat_j, j \neq i,$ the 
%opponents employ, the assessments of 
%player $i$ are calibrated almost surely.
%Here, by a randomized forecasting scheme we mean the following: At
at each step $t$, the player $i$ predicts a distribution $\mu_{-i}^t \in \Delta(A_{-i})$ by drawing one from a distribution over the space of distributions $\Delta(A_{-i})$, determined by the history $H^{t-1}$
(which we recall is given by the sequence of action profiles
of all the players over the steps up to $t-1$) and 
a random seed $U_i^t$, where the seeds
$(U_i^t, t \ge 1)$ are i.i.d. and independent of the
randomizations, if any, used by the other players.
The rule by which this probability distribution 
is created as a function of $H^{t-1}$ and $U_i^t$ is assumed
to be common knowledge to all the players.
The assessment of player $i$ at step $t$ is then the realization of
this random choice. 
%\black
%The existence of a randomized forecasting scheme follows from the following result on calibrated learning in the Nature-forecaster framework, where we imagine the opponents of player $i$ to play the role of Nature and player $i$ as the forecaster.
%\black
%under calibrated learning. 
%\citet{oakes1985self} proves that there does not exist a deterministic forecasting scheme that is calibrated 
%\cor
%for all sequences played by Nature.
%\cob
%As noted in \citep[Theorem 3]{foster1997calibrated}, however, there exists a randomized forecasting scheme such that, 
%no matter what outcome sequence 
%\red no matter what 
%\blue
%strategy
% \black
%Nature plays, the forecaster is almost surely calibrated. 
%\blue
%Let the forecaster have a randomized forecasting scheme, i.e. at each step $t$, the forecaster predicts a distribution $q^t \in \Delta(S)$ by drawing one randomly from a distribution $p^t$ over the space of distributions $\Delta(S)$,
 %(note that $\Delta(S)$ is a Polish space and hence the space of distributions $\Delta(S)$ is well-defined), 
 %based on the history of Nature's actions $(y^1,\dots,y^{t-1})$.
%Now consider the scenario in which Nature is modeled 
%as an adaptive adversary \citep{foster1998asymptotic}.
%as follows: 
%\red 
Lumping together the opponents of player $i$ as Nature 
from the point of view of this player, 
%\black
at each step $t$, Nature 
%\red 
can be assumed to have access not only to the history
$H^{t-1}$ but also to the realizations of the past seed values
$(U_i^1, \ldots, U_i^{t-1})$, so Nature knows the
assessments of the player $i$ from steps $1$ to $t-1$.
%has access to the history $(q^1, q^2, \dots, q^{t-1}, y^1, \dots, y^{t-1})$, 
%it knows the forecasting scheme used by the forecaster, and hence has access to $(p^1, \dots, p^t)$.
%Note that at step $t$ Nature knows the distribution $p^t$ on forecaster's predictions but does not know the prediction $q^t$.
%\red 
Crucially, while Nature now knows the distribution of the
assessment of player $i$ at time $t$, Nature does not know
the realization of this assessment till the next time step. 
%\black
%Let $\strat_f$ and $\strat_N$ be the strategies used by the forecaster and the Nature respectively and let $P_{\strat_f,\strat_N}$ denote the probability distribution generated by them. 
In this scenario 
%refered 
%\red 
(referred 
%\black
to as 
%an adaptive adversary 
%\red 
the {\em adaptive adversary} scenario 
%\black
in \citet{foster1998asymptotic}), a strategy for Nature 
%comprises of the 
%\red 
is comprised of 
%\black
Nature playing an action 
%$y^t$ 
at step $t$ by drawing one randomly from a distribution on 
%$S$, 
%\red 
her set of actions (i.e. the set $A_{-i}$ of action profiles of the opponents of player
$i$) 
%\black
based on the information available to 
%Nature 
%\red 
her 
%\black
at this 
%stage.
%\red 
step, namely $H^{t-1}$ and $(U_i^1, \ldots, U_i^{t-1})$. 
%\black
The calibrated learning result 
%as proved 
%\red 
proved 
%\black
in \citet{foster1998asymptotic} says that there exists 
%\red 
such 
%\black
a randomized forecasting scheme 
%\red 
on the part of player $i$ 
%\black
such that, no matter what 
%\red 
randomized 
%\black
strategy Nature employs
%\red 
as above, %\black
%the forecaster's calibration score
%\red 
we have %\black
\begin{equation}        \label{eq:calibscore}
	 \sum_{q \in Q^t} | \rho(q,y,t) - q(y) | \frac{N(q,t)}{t} \to 0, \text{ as } t \to \infty,
\end{equation}
for all 
%$y \in S$, 
$y \in A_{-i}$, 
almost surely (over the 
%randomization in the randomized forecasting scheme and 
%\red 
random seeds of player $i$ and the randomization in %\black
Nature's strategy).
%irrespective of the outcome sequence played by the Nature
\footnote{\citet{foster1998asymptotic} prove the existence of a randomized forecasting scheme that 
%\cor
makes the forecaster's calibration score,
%\red 
i.e. the expression in equation \eqref{eq:calibscore}, 
%\black 
tend to zero in probability. 
%\cob
However, as noted in \cite{cesa2006prediction}, the same argument proves that the convergence  of the calibration score holds, in fact, almost surely.}
%\red 
Here, as in equation \eqref{eq:calibrated}, $Q^t$ denotes the
set of probability distributions in $\Delta(A_{-i})$ actually
predicted by player $i$ up to step $t$. 
%\black

Combining this result with theorem \ref{thm: calibrated_CPTeq} we have,
%\black

\begin{corollary}
\label{cor: CPT_corr_convergence}
%\cor
	There exist 
	a
	randomized 
	%calibrated forecasting schemes 
	%\red 
	assessment scheme and a best reaction map %\black
	for each player such that, if each player predicts her assessments according to her scheme and plays the best reaction to her assessments, 
	%\red 
	then it is almost surely true (over the randomization in the %\black
	%\blue
	%choice of the assessment schemes) 
	randomized assessment schemes for the players)
	%\red
	that each player is calibrated with respect to the sequence of action profiles of her opponents, and hence %\black
	the empirical distribution of action play converges 
	%almost surely 
	to the set of mediated CPT correlated equilibria.
 %   \cob
\end{corollary}
\begin{proof}
Let player $i$ be the forecaster and let all the opponents together form Nature from the point of view of the player. 
%\cor
Thus Nature's action set is 
%$S = A_{-i}$. 
%\red 
$A_{-i}$. %\black
By the \citet{foster1998asymptotic} result, 
%\cob
there exists a randomized 
%\red 
assessment %\black
scheme for player $i$ 
%\red 
such that, whatever the randomized strategy that Nature uses,
the sequence of assessments of player $i$
%\black
is calibrated
%\red 
almost surely with respect to the sequence of actions of 
Nature. %\black 
Let player $i$ use 
%this 
%\red 
such a %\black
randomized scheme to 
%predict 
%\red 
determine %\black
her assessments. 
From remark \ref{rem: calib_conv_individual}, it follows that the empirical distribution of play converges to the set $\conv \l(C(\Gamma,i) \r)$ almost surely. 
If 
%each player $i$ 
%\red 
each player %\black
follows such a strategy, then we get almost sure convergence to $D(\Gamma)$.
\end{proof}
%%%%%%%%%%%%%%%%%%%%%%%%%%%%%%%%%%%%%%%%%%%%%%%%%

%\cor
We now show that, in a certain sense, the set $D(\Gamma)$ is the smallest possible extension 
%\cob
of the set $C(\Gamma)$ that guarantees convergence of the empirical distribution 
%of play
%\bb
of action play
%\cob
to this set, 
when all the players have 
%calibrated assessments and play the best reaction to these assessments. 
%\red 
assessment schemes and best reaction maps such that when each player plays the best reaction to her assessment at each step the player is calibrated with respect to the sequence of action profiles of her opponents. %\black
In particular, we claim the following.

\begin{proposition}
\label{prop: med_converse}
%\color{red}
%vacomment: deleted a phrase in the following sentence that seemed irrelevant. Please check.
%\color{black}
	For all games $\Gamma$ 
	%and any $\mu \in D(\Gamma)$,
such that the sets $C(\Gamma, i, a_i), i \in [n], a_i \in A_i$ do not have any isolated points, if $\mu \in D(\Gamma)$, then 
%\cor
%there exists a sequence of actions $(a^t_i)_{t \geq 1}$ and a %sequence of assessments $(\mu_{-i}^t)_{t \geq 1}$, for each player $i \in [n]$,
% such that the 
% \red sequence of assessments 
% of each player $i$ 
% is calibrated
% with respect to the sequence of actions of Nature, represented by the opponents
% of player $i$ playing according to their respective sequences of 
% actions, \black
% and such that the action $a_i^t$ is the best reaction
 %\footnote{
 %Here we ask $a_i^t$ to be one of the best responses to the assessment $\mu_{-i}^t$, instead of a fixed best response for each assessment.
 %} 
% to the assessment $\mu_{-i}^t$, for all $t \geq 1, i \in [n]$, and
%\cob
%\red 
there exists an assessment scheme and a best reaction map for each player such that if each player plays her best reaction to her assessment at each step then each player's assessments are calibrated with respect to the sequence of action profiles of her opponents
and %\black
 the empirical distribution of action play converges to $\mu$. 
\end{proposition}

%\cor
%With respect to the preceding sentence, how do we know that such games exist?
%More generally, some comment should be made about how many such games there are,
%e.g. is the situation generic?
%\cob
The following proposition (proved in Appendix~\ref{app: generic}) shows 
under some technical conditions on the value function of each player
%that. for 
%\blue
that, for 
%\black
%``generic'' 
%\color{red}
generic
%\color{black}
games $\Gamma$, the sets $C(\Gamma, i, a_i)$, $i \in [n], a_i \in A_i$, do not have any isolated points.
For any player $i$, we know that 
%\blue
the 
%\black
value function $v_i^{r_i}(x)$ is a strictly increasing continuous function.
Let the 
%\red 
open %\black
interval $Y_i \subset \bbR$ denote the range of $v_i^{r_i}$, and let $\lambda_i^*$ denote the 
%pushforward 
%\red 
push forward %\black
measure of the Lebesgue measure on $\bbR$ with respect to the function $v_i^{r_i}$.
Let $\hat \lambda_i$ denote the Lebesgue measure restricted to the interval $Y_i$.
We will require that the function $v_i^{r_i}$ is such that $\lambda_i^* \ll \hat \lambda_i$ (i.e., the measure $\lambda_i^*$ is absolutely continuous with respect to the measure $\hat \lambda_i$).
Since the function $v_i^{r_i}$ is strictly increasing, its inverse function $(v_i^{r_i})^{-1} : Y_i \to \bbR$ is well defined. 
We have $\lambda_i^* \ll \hat \lambda_i$ if and only if the function $(v_i^{r_i})^{-1}$ is absolutely continuous.

\begin{proposition}
\label{prop: generic_isolated}
For any fixed 
%preference features 
%\red 
CPT features %\black
$r_i, v_i^{r_i}, w_i^\pm$
such that $(v_i^{r_i})^{-1}$ is absolutely continuous,
and a fixed action set $A_i$ for each of the players $i \in [n]$ 
%\blue
(here, we assume $n > 1$ and $|A_i| > 1, \forall i \in [n]$),  
%\black
%for all of them, 
the set of all games $\Gamma$ for which there exists a player $i \in [n]$ and an action $a_i \in A_i$ such that the set $C(\Gamma, i, a_i)$ has an isolated point is a null set with respect to the Lebesgue measure $\lambda$ on the space of payoffs $(x_i(a), a \in A, i \in [n])$, viewed as an $n \times |A|$-dimensional Euclidean space.
\end{proposition}

\begin{proof}[Proof of Proposition~\ref{prop: med_converse}]
Since $\mu \in D(\Gamma)$, as noted in Remark~\ref{rem: rev_pure_strat}, there exists a signal system $(B_i)_{i \in [n]}$
%\red 
where $B_i$ can be identified with $A_i \times A_{-i}$, %\black
a mediator distribution $\psi \in \Delta(B)$, and a mediated CPT Nash equilibrium $\sigma \in \Sigma(\Gamma, (B_i)_{i \in [n]}, \psi)$ such that $\mu = \eta(\psi,\sigma)$, where $\sigma$ is a pure strategy profile. 
With 
an
abuse of notation, let $\sigma_i(b_i)$ denote the unique element in the support of $\sigma_i(b_i)$.
Let $(b^1,b^2,\dots)$ be a sequence of signal profiles such that the empirical distribution of these signal profiles converges to $\psi$
%\red 
and such that $\psi(b^t) > 0$ for all $t \ge 1$. %\black
At step $t$, let player $i$ predict her assessment $\tilde \mu_{-i}(\cdot|b_i)$ (as defined in equation~\eqref{eq: tilde_mu}) and play $\sigma_i(b_i)$. 
%\cor
%\st{It is easy to see that} 
The 
%assessments of all players are calibrated. 
%\red 
sequence of assessments of each player is calibrated with
respect to the sequence of 
%actions of Nature, represented by the opponents
% of player $i$ playing according to their respective sequences of 
% actions. 
action profiles of her opponents.
 %\black
%\cob
%\cor
%Why is this easy to see? This needs an argument.
%\cob
%\com
To see this, fix a player $i$, 
%and 
let $q \in \Delta(A_{-i})$ be one of the assessments made by her, and let $G = \{b_i \in B_i | \tilde \mu_{-i}(\cdot|b_i) = q\}$. 
Let $t^k(b_i)$ denote the step when player $i$ receives signal $b_i$ for the $k$th time. 
By construction, the empirical average of the action profiles of 
%\red 
the 
opponents of player $i$ %\black
over 
%\red 
the %\black 
steps $(t^k(b_i))_{k \geq 1}$ converges to $\tilde \mu_{-i}(\cdot|b_i)$. 
As a result, the empirical average of the action profiles of 
%\red 
the opponents of player $i$ %\black 
over 
%\red 
the %\black
steps when player $i$ receives a signal $b_i \in G$ converges to $q$. 
Since this holds for any assessment $q$ made by player $i$, her assessments are calibrated.
%\cob
Further, by construction, the empirical distribution of 
%play 
%\red 
action play %\black
converges to $\mu$. 

If $\tilde \mu_{-i}(\cdot|b_i) = \tilde \mu_{-i}(\cdot|\tilde b_i)$ implies $\sigma_i(b_i) = \sigma_i(\tilde b_i)$, for all $b_i,\tilde b_i \in B_i, i \in [n]$, then we can define $\sigma_i(b_i)$ as the best reaction to the assessment $\tilde \mu_{-i}(\cdot|b_i)$ and the claim is proved. 
%Suppose 
%\red 
If %\black
there exist
%\cob
$b_i,\tilde b_i$ such that $\tilde \mu_{-i}(\cdot|b_i) = \tilde \mu_{-i}(\cdot|\tilde b_i)$ but $\sigma_i(b_i) \neq \sigma_i(\tilde b_i)$, 
then there is a problem in defining the best reaction to the assessment $\tilde \mu_{-i}(\cdot|b_i)$.
%It is for this reason that we claim the result for almost all games and not all games. 
%{\color{blue} give a rigorous proof in appendix}
%In Appendix~\ref{app: nonunique_best_react}, 
We now describe a way to get around such a situation,
%\red 
analogous to the scheme used in \cite{foster1997calibrated}
to resolve the same kind of issue. %\black
% and this completes the proof.
Let $\zeta^*_{-i} := \tilde \mu_{-i}(\cdot|b_i) = \tilde \mu_{-i}(\cdot|\tilde b_i)$ and let $a_i^* := \sigma_i(b_i) \neq \sigma_i(\tilde b_i)$. 
By the assumption that the set $C(\Gamma, i, a_i^*)$ does not have any isolated points, 
there exists a sequence $(\hat \zeta_{-i}^l)_{l \geq 1}$ of distinct
probability
distributions in $C(\Gamma, i, a_i^*)$ such that $\hat \zeta_{-i}^l \to \zeta_{-i}^*$ 
and $(\hat \zeta_{-i}^l)_{l \geq 1}$ are all distinct from the distributions $(\tilde \mu_{-i}(\cdot|b_i), \forall b_i \in B_i)$. 
Further, let the sequence 
$(\hat \zeta_{-i}^l)_{l \geq 1}$ 
be such that 
$|\hat \zeta_{-i}^l(a_{-i}) - \zeta^*_{-i}(a_{-i})| < 1/l$, for all $a_{-i} \in A_{-i}$,
%\red 
i.e. $\hat \zeta_{-i}^l$ is within $1/l$ of 
$\zeta^*_{-i}$ in the sup norm, %\black
for all $l \geq 1$. 
%\bb
We will now replace the assessments $\zeta_{-i}^*$ at 
the %\black
steps $(t^k(b_i))_{k \geq 1}$ by 
the %\black
assessments $(\hat \zeta_{-i}^l)_{l \geq 1}$, with each $\hat \zeta_{-i}^l$ repeated sufficiently many times that the empirical distribution of the action profiles of the opponents over the steps that player $i$'s assessment is $\hat \zeta_{-i}^l$ is within $1/l$ of $\zeta_{-i}^*$
%\red 
in the sup norm. %\black
To achieve this, start by replacing the assessment at step $t^1(b_i)$ by $\hat \zeta_{-i}^1$.
Next replace the assessments at steps $t^k(b_i), k = 2, 3, \dots$ with $\hat \zeta_{-i}^2$ until the empirical distribution of the action profiles of the opponents over these steps is within $1/2$ of $\zeta_{-i}^*$
%\red 
in the sup norm. %\black
In general, keep replacing the assessments at steps $t^k(b_i)$ with $\hat \zeta_{-i}^l$ until the empirical distribution of the action profiles of the opponents over these steps is within $1/l$ of $\zeta_{-i}^*$
%\red 
in the sup norm, %\black
and then switch to replacing by $\hat \zeta_{-i}^{l+1}$.
Note that each assessment $\hat \zeta_{-i}^l$ will be used only for a finite number of steps since the empirical distribution of the action profiles of the opponents over the steps $(t^k(b_i))_{k \geq 1}$converges to $\zeta^*_{-i}$.
Thus, the empirical distribution of the action profiles of the opponents over the steps when player $i$ makes assessment $\hat \zeta_{-i}^l$ is within $2/l$ of $\hat \zeta_{-i}^l$
%\red 
in the sup norm. %\black
We know that if a sequence 
%\red 
of probability distributions $(s_t)_{t \geq 1}$ 
on $A_{-i}$ %\black
converges to 
%\red 
a probability distribution $s$ on $A_{-i}$, %\black
then the sequence of the running averages $S_t = (1/t)\sum_{\tau = 1}^t s_\tau, t \geq 1,$ also converges to $s$.
Using this fact, we observe that 
%\red 
the sequence of %\black
player $i$'s assessments 
%continue 
%\red 
continues %\black
to be calibrated 
%\red 
with
respect to the sequence of 
%actions of Nature, represented by the opponents
% of player $i$ playing according to their respective sequences of 
% actions, 
 action profiles of her opponents
 %\black
 even after the above replacement.
Since the assessments $\{\hat \zeta^l_{-i}\}$ are distinct from the assessments $(\tilde \mu_{-i}(\cdot|b_i), \forall b_i \in B_i)$, we can define action $a_i^*$ as the best reaction to $\hat \zeta^l_{-i}$ for all $l \geq 1$.
The above trick can be used to resolve all instances where $\tilde \mu_{-i}(\cdot|b_i) = \tilde \mu_{-i}(\cdot|\tilde b_i)$ but $\sigma_i(b_i) \neq \sigma_i(\tilde b_i)$. 
Each time taking the corresponding sequence $\{\hat \zeta^l_{-i}\}$ distinct from all previously used assessments. 
This solves the problem of defining 
%\red 
the best reaction map of each player %\black
%for all the assessments 
and completes the proof.
\end{proof}

\cob

\section{No-regret learning and CPT correlated equilibrium}
\label{sec: noregrete}

%\red NOTE: There is a problem with the following paragraph.
%The point is that Remark 3.9 assumes that the other players
%are playing in a way that the given player is calibrated with
%respect to their action profiles. This is not what is being
%claimed in the following paragraph, and it is not really clear 
%if what is claimed in the following paragraph is true. \black

The randomized forecasting scheme proposed in \citet{foster1998asymptotic} generates a probability distribution on the space of assessments of player $i$. Player $i$ draws her assessment from this distribution and then plays her best reaction. This two step process gives rise to a randomized strategy for player $i$ at each step. 
%The randomized forecasting scheme assumes that the Nature, at each step $t$, might know only the forecasting rule of the forecaster and not the actual forecast at step $t$ beforehand. Thus, in the context of a repeated game, the opponents could know the randomized strategy of player $i$ at each step but not the actual action played beforehand. The fact that the opponents do not know the action played by player $i$ beforehand is consistent with the underlying assumption that all players play simultaneously. 
%\blue
Together with Remark~\ref{rem: calib_conv_individual} we get that, no matter what strategies the opponents play, 
%\black
player $i$ can guarantee that the empirical distribution of action play 
%\cor
converges almost surely to 
%\cob 
the set $\conv \l(C(\Gamma,i) \r)$.

Under EUT, player $i$ has a strategy that guarantees 
%\cor
the almost sure convergence 
%\cob
of the empirical distribution 
%of play
%\bb
of action play
%\cob
to the set $C(\Gamma,i)$. This convergence is related to the notion of no-regret learning. We now describe this approach. 
Suppose that, at step $t$, player $i$ imagines replacing action $a_i$ by action $\tilde a_i$, every time she played action $a_i$ in the past. 
%\cor
Assuming the actions of the other players
%\cob
did not change, her payoff would become $x_i(\tilde a_i,a_{-i}^\tau)$ for all $\tau \leq t$ such that $a_i^t = a_i$, instead of $x_i(a_i,a_{-i}^\tau)$, while for all $\tau \leq t$ such that $a_i^t \neq a_i$ it will continue to be $x_i(a^t)$. We define the resulting \emph{CPT regret} of player $i$ for having played action $a_i$ instead of action $\tilde a_i$ as
%\cor
 \begin{equation}   \label{eq:CPTregret}
 	K_i^t(a_i,\tilde a_i) := \eda_i^t(a_i) \reg_i\l[\l\{\l(\eda_{-i}^t(a_{-i}|a_i),x_i(\tilde a_i,a_{-i}),x_i(a_i,a_{-i}) \r)\r\}_{a_{-i} \in A_{-i}}\r],
 \end{equation}
 %\cob
 where
 \begin{equation}
 \label{eq: reg_def}
 	\reg_i \l[\l\{ (\nu_l,\hat z_l, z_l)  \r\}_{l=1}^m\r] := V_i\l(\l\{(\nu_l,\hat z_l)\r\}_{l=1}^m \r) - V_i\l(\l\{(\nu_l, z_l)\r\}_{l=1}^m \r)
 \end{equation}
 is the difference in the CPT values of the lotteries $\l\{(\nu_l,\hat z_l)\r\}_{l=1}^m$ and $\l\{(\nu_l,z_l)\r\}_{l=1}^m$.
 We associate player $i$ 
% \cor
 with CPT regrets 
% \cob
 $\l\{K_i^t(a_i,\tilde a_i), a_i,\tilde a_i \in A_i, a_i \neq \tilde a_i\r)\}$ at each step $t$.
 Under EUT, 
% \cor
 this simplifies to
% \cob
 \begin{equation}
 \label{eq: EUT_regret_avg}
 	K_i^t(a_i,\tilde a_i) = \frac{1}{t} \sum_{\tau \leq t : a_i^\tau = a_i} [x_i(\tilde a_i,a_{-i}^\tau) - x_i(a^\tau)],
 \end{equation}
 in agreement with the definition given in \citet{hart2000simple}.

The following proposition shows the connection between regrets and correlated equilibrium.

 \begin{proposition}\label{prop: adpt_reg_eq}
 	Let $(a^t)_{t \geq 1}$ be a sequence of action profiles played by the players. Then $\limsup_{t \to \infty} K_i^t(a_i,\tilde a_i) \leq 0$, for every $i \in [n]$ and every $a_i,\tilde a_i \in A_i, a_i \neq \tilde a_i$, if and only if 
    %\cor
    the sequence of empirical distributions
%\cob    
    $\eda^t$ converges to the set $C(\Gamma)$ of CPT correlated equilibrium.
 	% in the sense that
	%\[
	%	d(\eda^t,C_\epsilon(\Gamma)) \to 0.
	%\] 
 \end{proposition}

%%%%%%%%%%%%%%%%%%%%%%%%%%%%%%%%%%%%%%%%%%%%%%%%%%%%%%%%%%%%%%%%%%%%%%%%%%%%%%%%%%%%%

\begin{proof}
%[Proof of proposition \ref{prop: adpt_reg_eq}]
Since $\Delta(A)$ is a compact set, $\eda^t$ converges to the set $C(\Gamma)$ iff for every convergent subsequence $\eda^{t_k}$, say, converging to $\hat \eda$, we have $\hat \eda \in C(\Gamma)$. Let $\eda^{t_k} \to \hat \eda$ be a convergent subsequence. For each player $i$, and for every $a_i,\tilde a_i \in A_i, a_i \neq \tilde a_i$ such that $\hat \eda_i(a_i) > 0$, we have
%\cor
\begin{equation}    \label{eq:regreteq}
	K_i^{t_k}(a_i,\tilde a_i) \to \hat \eda_i(a_i) \reg_i\l[\l\{\l(\hat \eda_{-i}(a_{-i}|a_i),x_i(\tilde a_i,a_{-i}),x_i(a_i,a_{-i}) \r)\r\}_{a_{-i} \in A_{-i}}\r],
\end{equation}
%\cob
by continuity of $V_i(p,x)$ as a function of the probability vector $p$ for a fixed outcome profile $x$. The result is immediate from the definition of CPT correlated equilibrium.
\end{proof}

%%%%%%%%%%%%%%%%%%%%%%%%%%%%%%%%%%%%%%%%%%%%%%%%%%%%%%%%%%%%%%%%%%%%%%
Player $i$ is said to have a no-regret learning strategy
%her regrets tend to be arbitrarily small almost surely, irrespective of other players' strategies, i.e.
%\bb
if, irrespective of the strategies of the other players,
her regrets satisfy
%\cob
\[
	P \l(\limsup_{t \to \infty} K_i^t(a_i,\tilde a_i) \leq 0\r) = 1, \text{ for every } a_i,\tilde a_i \in A_i, a_i \neq \tilde a_i.
\]
This is equivalent to asking if the vector of regrets $\l(K_i^t(a_i,\tilde a_i), a_i,\tilde a_i \in A_i, a_i \neq \tilde a_i\r))$, converges to the 
%negative orthant 
nonpositive orthant %\black
almost surely. 
This is related to the concept of approachability, the setup for which is as follows.
Consider a repeated two player game, where now at step $t$, if the row player and the column player play actions $\hat a_{row}^t$ and $\hat a_{col}^t$ respectively, then the row player receives a vector payoff $\vec x(\hat a_{row}^t,\hat a_{col}^t)$ instead of a scalar payoff. A subset $S$ is said to be approachable by the row player if she has a (randomized) strategy such that, no matter how the column player plays, we have 
%\black
\[
	\lim_{t \to \infty} d \l(\frac{1}{t} \sum_{\tau = 1}^t \vec x(\hat a_{row}^t,\hat a_{col}^t) , S \r) = 0, \text{ almost surely}.
\]
%\cor
Blackwell's approachability theorem, \cite{blackwell1956}, 
%states 
establishes
that 
%\cob
a convex closed set $S$ is approachable if and only if every halfspace $\hyp$ containing $S$ is approachable. 
%\magenta
%Comment: The preceding is not an accurate description of Blackwell's
%approachability theorem. It is a consequence of the theorem, but the
%statement is really about being ``able to get to the other side of the 
%hyperplane in one step". Please rewrite.
%\black

\citet{hart2000simple} cast the repeated game with stage game $\Gamma$ in the above setup as a two player repeated game where player $i$ is the row player and the opponents together form the column player. Let $\vec x(\hat a_i,\hat a_{-i})$ be the vector payoff when player $i$ plays action $\hat a_i$ 
%\cor
and the others play 
%\cob
$\hat a_{-i}$, with components given by
\[
	\vec x_{a_i,\tilde a_i}(\hat a_i,\hat a_{-i}) = \begin{cases}
		x_i(\tilde a_i, \hat a_{-i}) - x_i(a_i,\hat a_{-i}) &\text{ if } a_i = \hat a_i,\\
		0 &\text{ otherwise},
	\end{cases}
\]
for all $a_i,\tilde a_i \in A_i, a_i \neq \tilde a_i$. 
%\cor
%\cor
%The preceding equation makes no sense. Something called $a_{-i}$
%is appearing on the RHS. Is this supposed to be $\hat{a}_{-i}$.
%Please check this and correct the problem.
%\cob
Under EUT, the average vector payoff 
%\cob
of the row player corresponds to the regret of player $i$ (see equation~\ref{eq: EUT_regret_avg}). \citet{hart2000simple} show that the 
%negative orthant 
%\red 
nonpositive orthant %\black
is approachable for the row player and hence player $i$ has a no-regret learning strategy. Under CPT, if the average vector payoffs were to match 
%with player $i$'s regrets, 
%\red 
the regrets of player $i$, 
%\black
then the vector payoffs for the row player at step $t$ %turn out to 
%\red 
would need to %\black
depend on the empirical distribution of action play up to step $t$. 
Indeed, the component corresponding to the pair $(a_i, \tilde a_i)$ of the vector payoff for the row player at step $t$ when player $i$ plays action $\hat a_i$ and the others play $\hat a_{-i}$ 
%should 
%\red 
would need to %\black
match the difference
\[
	(t+1)K_i^{t+1}(a_i,\tilde a_i) - t K_i^t(a_i,\tilde a_i).
\]
%And this 
This
difference depends on the empirical distribution of action play up to step $t$, and hence in general changes with $t$.
This suggests that there might be difficulties in adapting the approach of \citet{hart2000simple} to study no-regret
learning strategies under CPT.

The following example shows that 
%under CPT, 
under CPT 
approachability 
of the 
%nonnegative orthant 
%\red 
nonpositive orthant %\black
need not hold.
%Blackwell's sufficiency condition \cite[Section 2]{blackwell1956} therefore does not hold with such state dependent payoffs.  
%\footnote{ \red Deleted a sentence here. See the source and 
%confirm 
%if it is okay. \blue This is okay. \black}
In other words, it can happen under CPT that at least one of the players does not have a no-regret learning strategy.
%\magenta
%Comment:
%In the preceding paragraph, the first sentence itself seems to mean
%that one of the players does not have a no-regret learning strategy.
%What is the point of mentioning Blackwell's sufficiency condition?
%\black
%\cob

%\input{tex/ex_nonapproach}
%!TEX root = ../main.tex

%%%%%%%%%%%%%%%%%%%%%%%%%%%%%%%%%%%%%%%%%%%%%%%%%%%%%%%%%%%%%%%%%%%%%%%%%%%%%%%%%%%%%%%%%%%%%%%%%%%%%%%%%%%%%%%%%%%%%%

\begin{example}
\label{ex: nonapproach_CPT}
	Consider the $2$-player repeated game from Example~\ref{ex: nonconvergence_calib}. 
	%Suppose player $1$ plays a no-regret learning strategy. 
	Recall the following distributions on player $2$'s actions: $\sigma_{odd} = (0.5,0,0.5,0), \sigma_{even} = (0,0.5,0,0.5)$ and $\sigma_{unif} = (0.25,0.25,0.25,0.25)$. We observed that player $1$'s action 1 is not a best response to $\sigma_{odd}$ and $\sigma_{even}$ and player $1$'s action 0 is not a best response to $\sigma_{unif}$.
	For an integer $T > 2$, consider the following strategy for player $2$:
	\begin{itemize}
		\item play mixed strategy $\sigma_{odd}$ at step $1$,
		\item play mixed strategy $\sigma_{even}$ at step $2$,
		\item play mixed strategy $\sigma_{odd}$ at steps $2T^k < t \leq T^{k+1}$, for $k \geq 0$,
		\item play mixed strategy $\sigma_{even}$ at steps $T^{k+1} < t \leq 2T^{k+1}$, for $k \geq 0$.
	\end{itemize}
\end{example}
	The rest of this section will be devoted to proving that player $1$ cannot have a no-regret learning strategy.
	In particular, we will prove the following:

\begin{proposition}
     \label{prop: nonconv_reg}
     	In the above example, for a suitable choice of $T, \tilde \delta > 0$ and $\tilde \epsilon > 0$, 
     %\cob
     there exists an integer $k_0$ such that no matter what learning strategy player $1$ uses, for all $k \geq k_0$ we have
	\[
		P\l( \bar K^k > \tilde \epsilon \r) > \tilde \delta,
	\]
	where
	\begin{equation}        \label{eq:3regrets}
		\bar K^k := 
		%[K_1^{\T^{k+1}}(\text{0},\text{1})]^+ 
		 [K_1^{T^{k+1}}(\text{1},\text{0})]^+ 
		+ [K_1^{2T^{k+1}}(\text{0},\text{1})]^+ 
		+ [K_1^{2T^{k+1}}(\text{1},\text{0})]^+,
	\end{equation}
	%and 
	%\red 
	using the notation %\black
	$[\cdot]^+ := \max \{\cdot,0\}$.
	%\red 
	Here, for actions $a_i$ and $\tilde a_i$ of player $1$,
	$K_1^t(a_i,\tilde a_i)$ are the CPT regrets of player $1$ at step $t$, as 
	defined in equation \eqref{eq:CPTregret}. %\black
\end{proposition}
	%By the strong law of large numbers, for any $\delta > 0$, there exists a $t_0 > 1$ such that for all $t \geq t_0$,
	%\[
	%	P(|\eda_2^t(\text{I}) - \eda_2^t(\text{III})| < \delta) > 1 - \delta,
	%\]
	%and
	%\[
	%	P(|\eda_2^t(\text{II}) - \eda_2^t(\text{IV})| < \delta) > 1 - \delta.
	%\]

	 %We now show that for a suitable choice of $\T, \epsilon > 0$ and $\delta > 0$,
     %\bb    
    %\begin{proof}

%%%%%%%%%%%%%%%%%%%%%%%%%%%%%%%%%%%%%%%%%%%%%%%%%%%%%%%%%%%

    Consider the subsequence of steps $(t^l_{odd})_{l \geq 1}$ when player $2$ played $\sigma_{odd}$. Let $\nu^l_{odd}(a_1,a_2)$ denote 
    the empirical distribution 
    %\bb
    over those times
    %\cob
    of
    %\cor
    the action profile $(a_1,a_2)$, where $a_1 \in \{\text{0,1}\}, a_2 \in \{\text{I,III}\}$, 
    %\cob
    i.e.
	\begin{equation}
		\label{eq: nu_odd_def}
		\nu^l_{odd}(a_1,a_2) := \frac{1}{l}\sum_{u = 1}^l \1 \{a^{t^u_{odd}} = (a_1,a_2)\}.
	\end{equation}
	Similarly, consider the sequence of steps $(t^l_{even})_{l \geq 1}$ when player $2$ played
    %\cor
    $\sigma_{even}$. 
    %\cob
    Let $\nu^l_{even}(a_1,a_2)$ denote the empirical distribution
    %\bb
    over those times
    %\cob
    %of the action profile 
    of the action profile 
    $(a_1,a_2)$, where 
    $a_1 \in \{\text{0,1}\}, a_2 \in \{\text{II,IV}\}$, i.e.
	\begin{equation}
		\label{eq: nu_even_def}
		\nu^l_{even}(a_1,a_2) := \frac{1}{l}\sum_{u = 1}^l \1 \{a^{t^u_{even}} = (a_1,a_2)\}.
	\end{equation}

\begin{lemma}
		\label{lem: nu_close}
		For any $\delta > 0$, there exists an integer $l_\delta > 1$, such that for all $l \geq l_\delta$, we have
		\begin{align}
		\label{eq: mart_err_1}
		P\l(|\nu^l_{odd}(\text{0,I})) - \nu^l_{odd}(\text{0,III})| < \delta \r) &> 1 - \delta,\\
		\label{eq: mart_err_2}
		P\l(|\nu^l_{odd}(\text{1,I})) - \nu^l_{odd}(\text{1,III})| < \delta \r) &> 1 - \delta,\\
		\label{eq: mart_err_3}
		P\l(|\nu^l_{even}(\text{0,II})) - \nu^l_{even}(\text{0,IV})| < \delta \r) &> 1 - \delta,\\
		\label{eq: mart_err_4}
		P\l(|\nu^l_{even}(\text{1,II})) - \nu^l_{even}(\text{1,IV})| < \delta \r) &> 1 - \delta.
	\end{align}
\end{lemma}
The proof of Lemma~\ref{lem: nu_close} can be found in Appendix~\ref{app: lemma-nu-close-proof}.

%%%%%%%%%%%%%%%%%%%%%%%%%%%%%%%%%%%%%%%%%%%%%%%%%%%%%%%%%%%%%%%%%%%%%%%%%%%%%%%%%%%%%%%%%%%%%%%%%%%%%%%%%%%%%%%%%%%%%%

For a vector $q \in \bbR^S$ and $\epsilon > 0$, let $[q]_{\epsilon} := \l\{\tilde q \in \bbR^S : |\tilde q(s) - q(s)| < \epsilon, \forall s \in S\r\}$ denote the set of all vectors 
%\red 
strictly %\black
within $\epsilon$ of $q$ in the sup norm.
    %\cor
    %The preceding notation seems to be used later for distributions that are not probability distributions. It should be defined accordingly.
    %\cob
	Select positive constants 
    %$\epsilon_3,c_3,\epsilon_2,c_2,\epsilon_1,c_2$ 
    %\bb
    $\epsilon_3,c_3,\epsilon_2,c_2,\epsilon_1,c_1$ 
    %\cob
    as follows: 
	\begin{itemize}
		\item Let $\epsilon_3 < 1$ and $c_3$ be such that 
		%action $0$ is not a best response to $\sigma_{unif}$ within $\epsilon_3$ error and 
		%\red 
		for the indicated regret we have %\black
		\begin{equation}
		\label{eq: eps_regret_3}
			\reg_1 \l[\l\{ \l(\mu(\cdot),x_1(\text{1},\cdot),x_1(\text{0},\cdot)\r)\r\}\r] > c_3,
		\end{equation}
		for all 
		%\red 
		probability %\black 
		distributions $\mu \in [\sigma_{unif}]_{\epsilon_3}$ (such constants exist because action 0 is not a best response to $\sigma_{unif}$). 
		Let 
		\begin{equation}
		\label{eq: del_3_def}
			\delta_3 := \epsilon_3/2.
		\end{equation}
		%\red 
		Note that $\delta_3 < 1/2$. %\black
		\item Let $\epsilon_2 < 1$ and $c_2$ be such that
		 %action $1$ is not a best response to $\sigma_{even}$ within $\epsilon_2$ error and 
		 %\red 
		 for the indicated regret we have %\black
		 \begin{equation}
		 \label{eq: eps_regret_2}
		 	\reg_1 \l[\l\{\l(\mu(\cdot),x_1(\text{0},\cdot),x_1(\text{1},\cdot)\r)\r\}\r] > c_2,
		 \end{equation}
		  for all 
		  %\red 
		  probability %\black 
		  distributions $\mu \in [\sigma_{even}]_{\epsilon_2}$ (such constants exist because action 1 is not a best response to $\sigma_{even}$). 
		 Let 
		 \begin{equation}
		 \label{eq: del_2_def}
		 	\delta_2 := {\epsilon_2 \delta_3}/{4}.
		 \end{equation}
		 %\red 
		 Note that $\delta_2 < 1/8$. %\black
		\item Let $\epsilon_1 < 0.5$ and $c_1$ be such that
		 %action $1$ is not a best response to $\sigma_{odd}$ within $\epsilon_1$ error and 
		 %\red 
		 for the indicated regret we have %\black 
		 \begin{equation}
		 	\label{eq: eps_regret_1}
		 	\reg_1 \l[\l\{\l(\mu(\cdot),x_1(\text{0},\cdot),x_1(\text{1},\cdot)\r)\r\}\r] > c_1,	
		 \end{equation}
		  for all 
		  %\red 
		  probability %\black
		  distributions $\mu \in [\sigma_{odd}]_{\epsilon_1}$ (such constants exist because action 1 is not a best response to $\sigma_{odd}$). 
		 Let 
\begin{equation}
\label{eq: del_1_def}
	\delta_1 := {\epsilon_1 \delta_2}.
\end{equation}
	%\red 
	Note that $\delta_1 < 1/16$. %\black	 
	\end{itemize}
	Let $T > {2}/{\delta_1}$ and $k_0$ be such that 
	\begin{equation}
	\label{eq: k_min_choose}
		T^{k_0+1} > \max \l\{t^{l_{\delta_1} }_{odd}, t^{l_{\delta_1} }_{odd}, t^{l_{\delta_1} }_{even}, t^{l_{\delta_1} }_{even} \r\},
	\end{equation}
	where $l_{\delta_1}$ is such that the inequalities in Lemma~\ref{lem: nu_close} hold for $\delta = \delta_1$.

%%%%%%%%%%%%%%%%%%%%%%%%%%%%%%%%%%%%%%%%%%%%%%%%%%%%%%%%%%%%%%%%%%%%%%%%%%%%%%%%%%%%%%%%%%%%%%%%%%%%%%%%%%%%%%%%%%%%%%

%Let 
%\red 
For $k \ge 0$, let %\black
$f_1^{k+1}$ denote the fraction of times player $2$ plays $\sigma_{even}$ up to step $t = T^{k+1}$. 
	From the definition of the strategy of player $2$, we have 
	\begin{equation}
	\label{eq: f_1 bound}
		f_1^{k+1} < \frac{2 T^k}{T^{k+1}} = \frac{2}{T}.
	\end{equation} 
	%Let 
	%\red 
	Similarly, for $k \ge 0$, let %\black
	$f_2^{k+1}$ denote the fraction of times player $2$ plays $\sigma_{even}$ up to step $t = 2 T^{k+1}$. We have
	\begin{equation}
	\label{eq: f_2 bound}
		f_2^{k+1} = \frac{T^{k+1} + \frac{T^{k+1} - 1}{T - 1}}{2T^{k+1}} \in \l[\frac{1}{2},\frac{1}{2} + \frac{1}{T}\r],
	\end{equation}
	where the last inclusion follows from the 
	%fact 
	%\red 
	assumption %\black
	that $T > 2$. 
	%(since $T > 2/\delta_1$ 
	%\bb
	%and $\delta_1 = (\epsilon_1 \epsilon_2 \epsilon_3)/8 < 1/8$).
	%\cob
	%\footnote{ \red Deleted the explanation here, since it is unnecessary. It was assumed at the outset that $T > 2$. Please check. \black}
    %\cor
    %The interval on the RHS might be $\l[\frac{1}{2},\frac{1}{2} + \frac{1}{2(\T - 1)}\r]$.
    %\cob
	%Let $T > 1/\epsilon_1$ and let $k_0$ be such that $T^{k+1} > \max\limits_{s = 1,2,3,4} \{t^{l^s_{\epsilon_1} }_{odd} \}$. Note that $k_0$ depends on the constants $T$ and $\epsilon_1$. 
	Note that
	\[
	 	f_2^{k+1}  = 1/2 + f_1^{k+1}/2.
	\] 
	%Let 
	%\red 
	Next, for $k \ge 0$, let %\black
	\begin{equation}        \label{eq:f3}
		f_3^{k+1} := \eda_1^{T^{k+1}}(\text{0}),
	\end{equation}
	%\red 
	i.e. the fraction of times player $1$ plays action $0$ up to step $t = T^{k+1}$,
	and let 
	\begin{equation} \label{eq:f4}
	f_4^{k+1} := 2 \xi_1^{2 T^{k+1}}(0) - \xi_1^{T^{k+1}}(0),
	\end{equation}
	i.e. the fraction of times player $1$ plays action $0$ among the
	steps from $T^{k+1} + 1$ to $2 T^{k+1}$. %\black
	Note that
	%\red 
	$f_3^{k+1}$ and $f_4^{k+1}$ are random variables, %\black
	 in contrast with $f_1^{k+1}$ and $f_2^{k+1}$.

%%%%%%%%%%%%%%%%%%%%%%%%%%%%%%%%%%%%%%%%%%%%%%%%%%%%%%%%%%%%%%%%%%%%%%%%%%%%%%%%%%%%%%%%%%%%%%%%%%%%%%%%%%%%%%%%%%%%%%

	We will establish the proof of Proposition~\ref{prop: nonconv_reg} is stages through several lemmas. 
    In the next couple of paragraphs we first outline our proof strategy.

	Depending on the strategy of player 1, we have two possibilities, either $P(f_3^{k+1} < 1 - \delta_2) > 1/4$ or $P(f_3^{k+1} < 1 - \delta_2) \leq 1/4$.
%	\footnote{ \red Note that the first sentence of this 
%	paragraph has been moved down relative to where it was in the earlier version of the document. Also a couple of imprecise sentences have been deleted in this paragraph and the next, since they are not needed. Please check. \black}
	%Suppose the learning strategy of player $1$ is such that 
	%there is a significant probability 
	%that she plays action 1 until step $T^{k+1}$ 
	%for a fraction of times bounded below by $\delta_2 > 0$. 
	%Then 
	%\red 
	In the former case, %\black
	in Lemma~\ref{lem: emp_rough_cond_prob_1}, we show that the empirical distribution 
	$\eda^{T^{k+1}}(\text{1}, \cdot)$ 
	is restricted to be of a certain type with significant probability, conditioned on 
	$\{f_3^{k+1} < 1 - \delta_2\}$. 
	The purpose of this lemma is to show that the conditional distribution 
	$\eda_{-1}^{T^{k+1}}(\cdot|\text{1})$ 
	is close to $\sigma_{odd}$.
	We explain this in Lemma~\ref{lem: prob_f3_big}, and use it to establish that player $1$ has a significant regret at step $T^{k+1}$ for not having played action 0 whenever she played action 1 up to that step, i.e. $K_1^{T^{k+1}}(\text{1,0})$ is considerable.

	%Now suppose instead, the learning strategy of player $1$ is such that with high probability she plays action 0 for almost all the steps until step $T^{k+1}$.
	%Then 
	%\red 
	In the latter case, %\black
	in Lemma~\ref{lem: emp_prob_2}, we show that the distribution $\eda^{2T^{k+1}}$ is restricted to be of a certain type with significant probability, conditioned on 
	$\{f_3^{k+1} \geq 1 - \delta_2\}$. 
	We then consider two cases depending on 
	%the fraction of times player $1$ plays action 0 from step $T^{k+1} + 1$ to step $2T^{k+1}$.
	%Let $f_4^{k+1}$ denote this fraction.
	%\red 
	$f_4^{k+1}$, which was defined in equation \eqref{eq:f4}.
	%\black
	If $f_4^{k+1}$ is less than $1 - \delta_3$, 
	then, in Lemma~\ref{lem: prob_f3_small_f4_small}, 
	we show that the conditional distribution 
	$\eda_{-1}^{2T^{k+1}}(\cdot|\text{1})$ 
	is similar to $\sigma_{even}$ 
	and hence
	player $1$ suffers from a significant regret at step $2T^{k+1}$ for not having played action 0 whenever she played action 1 up to that step, 
	i.e. 
	%$K_1^{T^{k+1}}(\text{1,0})$ 
	%\red 
	$K_1^{2 T^{k+1}}(\text{1,0})$ %\black
	is considerable.
	If $f_4^{k+1}$ is greater than or equal to $1 - \delta_3$, 
	then, in Lemma~\ref{lem: prob_f3_small_f4_big}, 
	we show that 
	the conditional distribution 
	$\eda_{-1}^{2T^{k+1}}(\cdot|\text{0})$ 
	is similar to $\sigma_{unif}$ 
	and hence
	player $1$ suffers from a significant regret at step $2T^{k+1}$ for not having played action 1 whenever she played action 0 up to that step, 
	i.e. 
	%$K_1^{T^{k+1}}(\text{0,1})$ 
	%\red 
	$K_1^{2 T^{k+1}}(\text{0,1})$ %\black
	is considerable.
	%Then finally we combine 
	%\red 
	Finally, we can combine %\black
	these results to show that player $1$ faces 
	%some or the other regret 
	%\red 
	some regret %\black
	either at step $T^{k+1}$ or $2T^{k+1}$ for all $k \geq k_0$, and hence the regret vector of player $1$ never 
	%converges to $0$.
	%\red 
	converges to the nonpositive orthant. %\black
%%%%%%%%%%%%%%%%%%%%%%%%%%%%%%%%%%%%%%%%%%%%%%%%%%%%%%%%%%%%%%%%%%%%%%%%%%%%%%%%%%%%%%%%%%%%%%%%%%%%%%%%%%%%%%%%%%%%%%
	
	Here are two simple technical lemmas that we will use repeatedly in the 
	%\red 
	rest of the discussion in this section. 
	The proof of each of these lemmas is elementary, and is therefore
	omitted. %\black
	%following lemmas.
	%The proof of these lemmas can be found in Appendix~\ref{sec: app_proofs}.
	\begin{lemma}
		\label{lem: cond_prob_bound_tech}
		If $P(F_1) > \alpha$ and $P(F_2) \geq \beta$ such that $\alpha + \beta > 1$, then
		\[
			P(F_1|F_2) \ge P(F_1 \cap F_2) > \alpha - (1-\beta).
		\]
		\hfill $\square$
	\end{lemma}

	\begin{lemma}
		\label{lem: sum_diff_tech}
		If 
		%\blue
		$\delta > 0$, 
		%\black
		and $x,y, a, b$ are real numbers such that $x+y \in [a,b]$ and 
		%\blue
		$|x - y| < \delta$,
		%\black
		 then
		 %\blue
		\[
			x,y \in ((a - \delta)/2, (b + \delta)/2).
		\]
		%\black
		\hfill $\square$
	\end{lemma}

	Let $E_1^k$ denote the event that the following inclusion holds:
	\begin{equation}		
		\label{eq:exbound}
		\eda^{T^{k+1}}(\text{1},\cdot) \in \l[\l(\frac{1 - f_3^{k+1}}{2},0,\frac{1 - f_3^{k+1}}{2},0\r)\r]_{\delta_1}.
	\end{equation}

	\begin{lemma}
		\label{lem: emp_rough_cond_prob_1}
		%\red 
		Recall that $k_0$ is defined in equation 
		\eqref{eq: k_min_choose}. %\black
		For any $k \geq k_0$, if $P(f_3^{k+1} < 1 - \delta_2) > 1/4$, then
		\begin{equation}
			\label{eq: emp_rough_cond_prob_1}
			P(E_1^k|f_3^{k+1} < 1 - \delta_2) > 1/4 - \delta_1.
		\end{equation}
	\end{lemma}
	\begin{proof}
	Fix $k \geq k_0$. 
	From inequality~\eqref{eq: f_1 bound} and the assumption $T > {2}/{\delta_1}$, we have 
    %$\eda^{T^{k+1}}(\text{1},II) + \eda^{T^{k+1}}(\text{1},III) \leq \frac{2}{T} < \delta_1$,
    %\bb
    %\blue
    \begin{equation}
    	\label{eq: emp_1_even_sum_1_bound}
    	\eda^{T^{k+1}}(\text{1,II}) + \eda^{T^{k+1}}(\text{1,IV}) < \frac{2}{T} < \delta_1,
    \end{equation}
    %\black
    and hence each term is strictly less than $\delta_1$, i.e.
    %\blue
    \begin{equation}
    \label{eq: emp_1_even_small}
     	%\eda^{T^{k+1}}(\text{1,II}),  \eda^{T^{k+1}}(\text{1,IV}) \in [0,\delta_1]
     	\eda^{T^{k+1}}(\text{1,II}),  \eda^{T^{k+1}}(\text{1,IV}) \in [0,\delta_1).
     \end{equation}
     %\black
    %\cob
    Since $k \geq k_0$, from \eqref{eq: k_min_choose} and \eqref{eq: mart_err_2}, for $l := \max\{l: t^l_{odd} \leq T^{k+1}\}$, we have 
    \begin{align*}
    	P\l(|\eda^{T^{k+1}}(\text{1,I}) - \eda^{T^{k+1}}(\text{1,III})| < \delta_1 \r) &= P \l (|\nu^l_{odd}(\text{1,I})) - \nu^l_{odd}(\text{1,III})| (1 - f_1^{k+1}) < \delta_1 \r)\\
    	&\geq P\l(|\nu^l_{odd}(\text{1,I})) - \nu^l_{odd}(\text{1,III})| < \delta_1\r) > 1 - \delta_1,
    \end{align*} 
    % Since $P(f_3^{k+1} < 1 - \delta_2) > 1/4$, we have $P(f_3^{k+1} \geq 1 - \delta_2) \leq 3/4$. 
    % Besides, $P(f_3^{k+1} < 1 - \delta_2) \leq 1$ and $P(E_1 | f_3^{k+1} \geq 1 - \delta_2) \leq 1$ since they are probabilities.
    % If we denote the event 
    % \[
    % 	\l\{|\eda^{T^{k+1}}(\text{1,I}) - \eda^{T^{k+1}}(\text{1,III})| < \delta_1\r\}
    % \] 
    % by $F_1^k$, then
    % \begin{align*}
    % 	P(F_1^k) > 1 -\delta_1 \iff &P(F_1^k | f_3^{k+1} < 1 - \delta_2)P(f_3^{k+1} < 1 - \delta_2)\\
    % 	 &+ P(F_1^k | f_3^{k+1} \geq 1 - \delta_2)P(f_3^{k+1} \geq 1 - \delta_2) > 1 -\delta_1 \\
    % 	 \implies &P(F_1^k | f_3^{k+1} < 1 - \delta_2) + 1 \times (3/4) > 1 -\delta_1.
    % \end{align*}
    % Thus we have,
    In Lemma~\ref{lem: cond_prob_bound_tech}, taking $F_1$ to be the event
    \[
    	\l\{|\eda^{T^{k+1}}(\text{1,I}) - \eda^{T^{k+1}}(\text{1,III})| < \delta_1\r\},
    \] 
    and $F_2$ to be the event $\{f_3^{k+1} < 1 - \delta_2\}$, we have $P(F_1) > 1 - \delta_1$, $P(F_2) > 1/4$. 
    %And since 
    %\red 
    Since %\black
    $\delta_1 < 1/4$, we have 
    \begin{equation}
    \label{eq: prob_eda_1_odd_diff_bound}
    	P\l(|\eda^{T^{k+1}}(\text{1,I}) - \eda^{T^{k+1}}(\text{1,III})| < \delta_1 \Big| f_3^{k+1} < 1 -\delta_2\r) > 1/4 - \delta_1.
    \end{equation}
    Since
    \[
    	\eda^{T^{k+1}}(\text{1,I}) + \eda^{T^{k+1}}(\text{1,II}) + \eda^{T^{k+1}}(\text{1,III}) + \eda^{T^{k+1}}(\text{1,IV}) = 1 - f_3^{k+1},
    \]
    combined with \eqref{eq: emp_1_even_sum_1_bound}, we have
    \begin{equation}
    \label{eq: emp_1_odd_sum_1_bound}
    	\eda^{T^{k+1}}(\text{1,I}) + \eda^{T^{k+1}}(\text{1,III}) \in \l[1 - f_3^{k+1} - \delta_1, 1 - f_3^{k+1}\r].
    \end{equation}
    From \eqref{eq: prob_eda_1_odd_diff_bound}, \eqref{eq: emp_1_odd_sum_1_bound} and Lemma~\ref{lem: sum_diff_tech}, we have
    %\blue
    %\[
    %	P\l( \eda^{T^{k+1}}(\text{1,I}), \eda^{T^{k+1}}(\text{1,III}) \in \l[\frac{1 - f_3^{k+1} - \delta_1}{2} - \frac{\delta_1}{2},  \frac{1 - f_3^{k+1}}{2} + \frac{\delta_1}{2} \r] \Big | f_3^{k+1} < 1 - \delta_2\r) > 1/4 - \delta_1.
    %\]
    \[
    	P\l( \eda^{T^{k+1}}(\text{1,I}), \eda^{T^{k+1}}(\text{1,III}) \in \l(\frac{1 - f_3^{k+1} - \delta_1}{2} - \frac{\delta_1}{2},  \frac{1 - f_3^{k+1}}{2} + \frac{\delta_1}{2} \r) \Bigg | f_3^{k+1} < 1 - \delta_2\r) > 1/4 - \delta_1.
    \]
    %\black
    Combined with \eqref{eq: emp_1_even_small}, we get \eqref{eq: emp_rough_cond_prob_1} and this completes the proof of the lemma.
	\end{proof}

	\begin{lemma}
		\label{lem: prob_f3_big}
		For any $k \geq k_0$, if $P(f_3^{k+1} < 1 - \delta_2) > 1/4$, then
		\begin{equation}
			\label{eq: prob_K_case_1}
			%P(\bar K^k > \delta_2c_1) 
			P\l([K_1^{T^{k+1}}(\text{1,0})]^+ > \delta_2 c_1 \r) 
			> \frac{1}{4}\l(\frac{1}{4} - \delta_1\r).
		\end{equation}
	\end{lemma}
	\begin{proof}
		From \eqref{eq: f_1 bound}, we know that player $2$ plays $\sigma_{odd}$ for 
		%most 
	%\red 
	at least a fraction $1 - \frac{2}{T}$ %\black
		of the steps up to step $t = T^{k+1}$. 
 	Since action 1 is not a best response of player $1$ for $\sigma_{odd}$, 
    we will now show that, if player $1$ does not play action 0
    for a sufficiently high fraction of steps up to step $t = T^{k+1}$, then she will have a 
    %considerably large 
    %\red 
    significant %\black
    regret $K_1^{T^{k+1}}(\text{1,0})$. 
	More precisely, for any $k \geq k_0$, if $f_3^{k+1} < 1 - \delta_2$ and the inclusion 
	\eqref{eq:exbound} holds, then
	%\red 
	we can write %\black
	\begin{align*}
		&\eda_{-1}^{T^{k+1}}(\cdot | \text{1}) \eda_{1}^{T^{k+1}}(\text{1}) 
		\in \l[\l(\frac{1 - f_3^{k+1}}{2},0,\frac{1 - f_3^{k+1}}{2},0\r)\r]_{\delta_1}\\
		&\iff \eda_{-1}^{T^{k+1}}(\cdot | \text{1}) (1 - f_3^{k+1}) \in 
		\l[\l(\frac{1 - f_3^{k+1}}{2},0,\frac{1 - f_3^{k+1}}{2},0\r)\r]_{\delta_1}\\
		&\iff \eda_{-1}^{T^{k+1}}(\cdot | \text{1}) 
		\in 
		\l[\l(\frac{1}{2},0,\frac{1}{2},0\r)\r]_{\delta_1/(1 - f_3^{k+1})} \\
		&\implies \eda_{-1}^{T^{k+1}}(\cdot | \text{1}) \in [\sigma_{odd}]_{\frac{\delta_1}{\delta_2}}.
	\end{align*}
	Hence, from \eqref{eq: del_1_def} and \eqref{eq: eps_regret_1}, 
	we have 
	%\red
	%$\frac{\delta_1}{\delta_2 - 2\delta_1} < \epsilon_1$ and hence
	 %and since $\epsilon_1 << \epsilon_2$, 
	 %the regret 
	 \[
	 	K_1^{T^{k+1}}(\text{1,0}) = \eda_1^{T^{k+1}}(\text{1}) \reg_1\l[\l\{\l(\eda_{-1}^{T^{k+1}}(\cdot | \text{1}), x_1(\text{0},\cdot), x_1(\text{1}, \cdot)\r)\r\}\r] >  \delta_2 c_1,
	 \]
	 on the event where $f_3^{k+1} < 1 - \delta_2$ and the inclusion 
	\eqref{eq:exbound} holds.
	%\black
	%\footnote{ \red In the equation above, the order of the terms in the regret expression on the right hand side needed to be interchanged. Please check. \black}
	 %for some fixed constant $c_1 > 0$ dependent on $\epsilon_1$ (because action $1$ is not a best response to $\sigma_{odd}$).
	%Let $\1\{\cdot\}$ be the indicator random vector that takes value $1$ if the expression inside the brackets is true and $0$ otherwise. 
	Thus for any $k \geq k_0$, if $P(f_3^{k+1} < 1 - \delta_2) > 1/4$, then we have
	\begin{align*}
		%P(\bar K^k > \delta_2c_1) &\geq 
		&P\l([K_1^{T^{k+1}}(\text{1,0})]^+ > \delta_2 c_1 \r) \\
		&=  P\l([K_1^{T^{k+1}}(\text{1,0})]^+ > \delta_2 c_1 \Big | f_3^{k+1} < 1 - \delta_2\r)P(f_3^{k+1} < 1 - \delta_2)\\
		&\geq P\l(E_1^k | f_3^{k+1} < 1 -\delta_2 \r)P(f_3^{k+1} < 1 - \delta_2)\\
		&> \frac{1}{4}\l(\frac{1}{4} - \delta_1\r),
	\end{align*}
	where the 
	%second last inequality 
	%\red 
	last but one inequality %\black
	follows from the fact that 
	$E_1^k$ and $\{f_3^{k+1} < 1 - \delta_2\}$ imply $[K_1^{T^{k+1}}(\text{1,0})]^+ > \delta_2 c_1$, 
	and the
	last inequality follows from the condition $P(f_3^{k+1} < 1 - \delta_2) > 1/4$ and Lemma~\ref{lem: emp_rough_cond_prob_1}. 
	\end{proof}

Consider now the 
%\red 
probability %\black 
distribution $\hat \mu$ 
%as shown 
%\red 
shown %\black
in Table~\ref{tab: 2x4 game non approach 2}.
Recall that $f_4^{k+1}$ is the fraction of times player $1$ plays action 0 
%from step 
%\red 
among the steps from step %\black
$T^{k+1} + 1$ 
to step $2T^{k+1}$. 
%\red 
Note that, since $f_4^{k+1}$ is a random variable, 
so is $\hat \mu$. %\black

%We now show that for any $k \geq k_0$, the empirical distribution at step $2T^{k+1}$ is approximately as shown in Table~\ref{tab: 2x4 game non approach 2} within $\delta_2$ error, 
%i.e. $\eda^{2T^{k+1}} \in [\hat \mu]_{\delta_2}$, with probability greater than $(1/4 - 3\delta_1)$, conditioned on $\{f_3^{k+1} \geq 1 - \delta_2\}$. 

%%%%%%%%%%%%%%%%%%%%%%%%%%%%%%%%%%%%%%%%%%%%%%%%%%%%%%%%%%%

 \begin{table}
 \centering
 \begin{tabular}{c | c | c | c | c |}
 	 \multicolumn{1}{c}{}	& \multicolumn{1}{c}{I}   &  \multicolumn{1}{c}{II}  &  \multicolumn{1}{c}{III}  &  \multicolumn{1}{c}{IV} \\
 	 \cline{2-5}
 	0 & $0.25$  &  $0.25 f_4^{k+1} $  &  $0.25$  & $0.25 f_4^{k+1}$\\
 	 \cline{2-5}
 	 1	& $0$ & $0.25(1 - f_4^{k+1})$ & $0$ & $0.25(1 - f_4^{k+1})$\\
 	 \cline{2-5}
 	 \end{tabular} 
 	 \caption{Empirical distribution $\hat \mu$ in example~\ref{ex: nonapproach_CPT}.}\label{tab: 2x4 game non approach 2}
 \end{table}

%%%%%%%%%%%%%%%%%%%%%%%%%%%%%%%%%%%%%%%%%%%%%%%%%%%%%%%%%%%

\begin{lemma}
\label{lem: emp_prob_2}
	For all $k \geq k_0$, if $P(f_3^{k+1} < 1 - \delta_2) \leq 1/4$, then 
	%\red
		 \begin{equation}
	 \label{eq: cond_prob_second}
	 	P(\eda^{2T^{k+1}} \in [\hat \mu]_{\delta_2} | f_3^{k+1} \geq 1 - \delta_2) > 1/4 - 3\delta_1.
	 \end{equation}
	 %\black
	 %\red 
	 We also recall that $\delta_1 < 1/16$, so the lower bound in \eqref{eq: cond_prob_second} is strictly positive. %\black
\end{lemma}
%\footnote{ \red The inequality has been changed into a strict inequality. As a general remark, you are incredibly sloppy about the distinction between strict and weak inequality even in places where it actually matters (unlike here), as you can see by looking at some of the changes that are coming up later. If you continue to do mathematical work in future this lack of concern for corner cases may lead to nontrivial mistakes. \black}
\begin{proof}
	Since player $2$ plays $\sigma_{even}$ from step $T^{k+1}+1$ to step $2 T^{k+1}$, 
	 if $f_3^{k+1} \geq 1 - \delta_2$, 
	%$f_3^{k+1} \geq 1 - \delta_2$, 
	 then 
     %the sum $\eda^{2T^{k+1}}(\text{1},I) + \eda^{2T^{k+1}}(\text{1},III) < \delta_2/2$
     %\com
     \begin{equation}
     \label{eq: emp_1_odd_bound_up}
     	\eda^{2T^{k+1}}(\text{1,I}) + \eda^{2T^{k+1}}(\text{1,III}) \leq \eda^{T^{k+1}}_1(1)/2 = (1-f_3^{k+1})/2
     %- \frac{1}{2} (1 - f_4^{k+1}) 
     \leq \delta_2/2.
     \end{equation}
     %\cob
     %And hence 
     %\red 
     This means that %\black
     each term is strictly less than $\delta_2$, 
     %i.e.
     %\red 
     so we have %\black
     %\blue
     \begin{equation}
     \label{eq: eda_1_odd_bound}
     	\eda^{2T^{k+1}}(\text{1,I}), \eda^{2T^{k+1}}(\text{1,III}) \in [0,\delta_2).
     \end{equation}
     %\black
     %\cor
 	 %I have not checked the rest of the proof, because of the issue in red above.
     %I also note that you don't seem to be estimating $[K_1^{\T^{k+1}}(\text{0},\text{1})]^+$ anywhere in this proof and that the issue of what you actually mean by the indicator functions is ambiguous in the estimates for
     %$[K_1^{\T^{k+1}}(\text{1},\text{0})]^+$ (above) and for 
     %$[K_1^{2 \T^{k+1}}(\text{0},\text{1})]^+$ and $[K_1^{2 \T^{k+1}}(\text{1},\text{0})]^+$ (below).
     %\cob
	Further, from equation \eqref{eq: f_2 bound} and the assumption $T > 2/\delta_1$, 
	we have 
	\begin{align*}
		\eda^{2T^{k+1}}(\text{0,I}) &+ \eda^{2T^{k+1}}(\text{0,III}) 
	+ \eda^{2T^{k+1}}(\text{1,I}) + \eda^{2T^{k+1}}(\text{1,III}) \\
	&= 1 - f_2^{k+1} 
	\in [0.5 - 1/T, 0.5] \subset [0.5 - \delta_1, 0.5].
	\end{align*}
	%Combined 
	%\red 
	Combining this %\black
	with \eqref{eq: emp_1_odd_bound_up}, we have
	\begin{equation}
	\label{eq: emp_0_odd_sum_2}
		\eda^{2T^{k+1}}(\text{0,I}) + \eda^{2T^{k+1}}(\text{0,III})  
	\in [0.5 - \delta_1 - \delta_2/2, 0.5],
	\end{equation}
	%\red 
	on the event where $f_3^{k+1} \geq 1 - \delta_2$. %\black
	Since $k \geq k_0$, from \eqref{eq: k_min_choose} and \eqref{eq: mart_err_1}, for $l := \max\{l: t^l_{odd} \leq 2T^{k+1}\}$, we have 
	%since $\delta_1(1 - f_2^{k+1}) < \delta_1$,
	%\blue
	\begin{align*}
		P\l(|\eda^{2T^{k+1}}(\text{0,I}) - \eda^{2T^{k+1}}(\text{0,III})| < \delta_1\r) 
		&= P\l(|\nu^l_{odd}(\text{0,I})) - \nu^l_{odd}(\text{0,III})|(1 - f_2^{k+1})  < \delta_1 \r) \\
		&\geq P\l(|\nu^l_{odd}(\text{0,I})) - \nu^l_{odd}(\text{0,III})|  < \delta_1 \r) > 1 - \delta_1.
	\end{align*} 
	%\black
	% Denoting the event 
	% \[
	% 	\l\{|\eda^{2T^{k+1}}(\text{0,I}) - \eda^{2T^{k+1}}(\text{0,III})| < \delta_1\r\}
	% \]
	% by $F_2^{k+1}$,
	% we have
	% \begin{align*}
	% 	P(F_2^{k+1}) > 1 - \delta_1 \iff 
	% 	&P(F_2^{k+1}|f_3^{k+1} \geq 1 -\delta_2)P(f_3^{k+1} \geq 1 -\delta_2) \\
	% 	&+ P(F_2^{k+1}|f_3^{k+1} < 1 -\delta_2)P(f_3^{k+1} < 1 -\delta_2) 
	% 	> 1 - \delta_1 \\
	% 	\implies &P(F_2^{k+1}|f_3^{k+1} \geq 1 -\delta_2) + 1 \times (1/4) > 1 - \delta_1.
	% \end{align*}
	In Lemma~\ref{lem: cond_prob_bound_tech}, taking $F_1$ to be the event
	%\blue
    \[
	 	\l\{|\eda^{2T^{k+1}}(\text{0,I}) - \eda^{2T^{k+1}}(\text{0,III})| < \delta_1\r\}
	\]
	%\black
    and $F_2$ to be the event $\{f_3^{k+1} \geq 1 - \delta_2\}$, we have $P(F_1) > 1 - \delta_1$, $P(F_2) \geq 3/4$. 
    %And since 
    %\red 
    Since %\black
    $\delta_1 < 1/4$, we have 
    %\blue
	\begin{equation}
	\label{eq: emp_0_odd_diff_2}
		P\l(|\eda^{2T^{k+1}}(\text{0,I}) - \eda^{2T^{k+1}}(\text{0,III})| < \delta_1 \Big| f_3^{k+1} \geq 1 - \delta_2\r) > 3/4 - \delta_1.
	\end{equation}
	%\black
	Form \eqref{eq: emp_0_odd_sum_2}, \eqref{eq: emp_0_odd_diff_2} and Lemma~\ref{lem: sum_diff_tech}, we have
	%\blue
	\begin{align*}
		P\l( \eda^{2T^{k+1}}(\text{0,I}), \eda^{2T^{k+1}}(\text{0,III}) \in (0.25 - \delta_1 - \delta_2/4, 0.25 + \delta_1/2) \Big | f_3^{k+1} \geq 1 - \delta_2 \r) > 3/4 - \delta_1.
	\end{align*}
	%\black
	%\red 
	Here we note that $0.25 - \delta_1 - \delta_2/4 > 0$.
	%\black
	% As a result, each of the two terms, $\eda^{2T^{k+1}}(\text{0,I})$ and $\eda^{2T^{k+1}}(\text{0,III})$, 
	% lie in the interval $[0.25 - \delta_1 - \delta_2/4, 0.25 + \delta_1]$, 
	% with probability greater than $(3/4 - \delta_1)$, 
	% conditioned on $\{f_3^{k+1} \geq 1 - \delta_2\}$. 
	 Since $\epsilon_1 < 0.5$ and $\delta_1 = \epsilon_1 \delta_2$, we have 
	 %\blue
	 \begin{equation}
	 \label{eq: eda_0_odd_bound}
	 	P\l( \eda^{2T^{k+1}}(\text{0,I}), \eda^{2T^{k+1}}(\text{0,III}) \in (0.25 - \delta_2,0.25 + \delta_2) \Big | f_3^{k+1} \geq 1 - \delta_2 \r) > 3/4 - \delta_1.
	 \end{equation}
	 %\black
	 From \eqref{eq: f_1 bound} and the assumption $T > 2/\delta_1$, we have
	 \begin{align}
	 \label{eq: emp_0_even_sum_bound}
	 	\eda^{2T^{k+1}}(\text{0,II}) + \eda^{2T^{k+1}}(\text{0,IV}) &\in [0.5f_4^{k+1}, 0.5f_4^{k+1} + 0.5f_1^{k+1}] \nonumber\\
	 	&\in [0.5f_4^{k+1}, 0.5f_4^{k+1} + \delta_1].
	 \end{align}
	 Since $k \geq k_0$, from \eqref{eq: k_min_choose} and \eqref{eq: mart_err_3}, for $l := \max\{l: t^l_{even} \leq 2T^{k+1}\}$, we have 
	 %\blue
	  \begin{align*}
		P\l(|\eda^{2T^{k+1}}(\text{0,II}) - \eda^{2T^{k+1}}(\text{0,IV})| < \delta_1\r) 
		&= P\l(|\nu^l_{even}(\text{0,II})) - \nu^l_{even}(\text{0,IV})|(f_2^{k+1})  < \delta_1 \r) \\
		&\geq P\l(|\nu^l_{even}(\text{0,II})) - \nu^l_{even}(\text{0,IV})|  < \delta_1 \r) > 1 - \delta_1.
	\end{align*} 
	%\black
% 	Using the bound $P(f_3^{k+1} < 1 - \delta_2) \leq 1/4$, and denoting the event 
% \[
% 		\l\{|\eda^{2T^{k+1}}(\text{0,I}) - \eda^{2T^{k+1}}(\text{0,III})| < \delta_1\r\}
% 	\]
% 	by $F_3^{k+1}$,
% 	we have
% 	\begin{align*}
% 		P(F_3^{k+1}) > 1 - \delta_1 \iff 
% 		&P(F_3^{k+1}|f_3^{k+1} \geq 1 -\delta_2)P(f_3^{k+1} \geq 1 -\delta_2) \\
% 		&+ P(F_3^{k+1}|f_3^{k+1} < 1 -\delta_2)P(f_3^{k+1} < 1 -\delta_2) 
% 		> 1 - \delta_1 \\
% 		\implies &P(F_3^{k+1}|f_3^{k+1} \geq 1 -\delta_2) + 1 \times (1/4) > 1 - \delta_1.
% 	\end{align*}
	  %and since $|\eda^{2T^{k+1}}(\text{0},II) - \eda^{2T^{k+1}}(\text{0},IV)| < \delta_1$ with probability at least $1 - \delta$, we have 
	  In Lemma~\ref{lem: cond_prob_bound_tech}, taking $F_1$ to be the event
	 % \blue
    \[
		\l\{|\eda^{2T^{k+1}}(\text{0,II}) - \eda^{2T^{k+1}}(\text{0,IV})| < \delta_1\r\}
 	\]
 	%\black
    and $F_2$ to be the event $\{f_3^{k+1} \geq 1 - \delta_2\}$, we have $P(F_1) > 1 - \delta_1$, $P(F_2) \geq 3/4$. 
    %And since
    %\red 
    Since %\black
    $\delta_1 < 1/4$, we have 
    %\blue
	\begin{equation}
	\label{eq: emp_0_even_diff_3}
		P\l(|\eda^{2T^{k+1}}(\text{0,II}) - \eda^{2T^{k+1}}(\text{0,IV})| < \delta_1 \Big | f_3^{k+1} \geq 1 - \delta_2 \r) > 3/4 - \delta_1.
	\end{equation}
	%\black
	From \eqref{eq: emp_0_even_sum_bound}, 
	%\eqref{eq: emp_0_even_sum_bound} 
	%\red 
	\eqref{eq: emp_0_even_diff_3} %\black
	and Lemma~\ref{lem: sum_diff_tech}, we have
	%\blue
	  \begin{align}
	  	\label{eq: eda_0_even_bound}
	  	P\l(\eda^{2T^{k+1}}(\text{0,II}), \eda^{2T^{k+1}}(\text{0,IV}) \in (0.25 f_4^{k+1} - \delta_1, 0.25 f_4^{k+1} + \delta_1) \Big | f_3^{k+1} \geq 1 - \delta_2 \r) > 3/4 - \delta_1,
	  \end{align}
	  %\black
	  %\red 
	  Note that here $0.25 f_4^{k+1} - \delta_1$ could be negative. %\black
	%Similarly, we get 
	From \eqref{eq: f_1 bound} and the assumption $T > 2/\delta_1$, we have
	 \begin{align}
	 \label{eq: emp_1_even_sum_bound_2}
	 	\eda^{2T^{k+1}}(\text{1,II}) + \eda^{2T^{k+1}}(\text{1,IV}) &\in [0.5(1-f_4^{k+1}), 0.5(1-f_4^{k+1}) + 0.5f_1^{k+1}] \nonumber\\
	 	&\in [0.5(1-f_4^{k+1}), 0.5(1-f_4^{k+1}) + \delta_1].
	 \end{align}
	 Since $k \geq k_0$, from \eqref{eq: k_min_choose} and \eqref{eq: mart_err_4}, for $l := \max\{l: t^l_{even} \leq 2T^{k+1}\}$, we have
	 %\blue
	  \begin{align*}
		P\l(|\eda^{2T^{k+1}}(\text{1,II}) - \eda^{2T^{k+1}}(\text{1,IV})| < \delta_1\r) 
		&= P\l(|\nu^l_{even}(\text{1,II})) - \nu^l_{even}(\text{1,IV})|(f_2^{k+1})  < \delta_1 \r) \\
		&\geq P\l(|\nu^l_{even}(\text{1,II})) - \nu^l_{even}(\text{1,IV})|  < \delta_1 \r) > 1 - \delta_1.
	\end{align*} 
	%\black
In Lemma~\ref{lem: cond_prob_bound_tech}, taking $F_1$ to be the event
%\blue
    \[
	 	\l\{|\eda^{2T^{k+1}}(\text{1,II}) - \eda^{2T^{k+1}}(\text{1,IV})| < \delta_1\r\}
	\]
	%\black
    and $F_2$ to be the event $\{f_3^{k+1} \geq 1 - \delta_2\}$, we have $P(F_1) > 1 - \delta_1$, $P(F_2) \geq 3/4$. 
    %And since 
    %\red 
    Since %\black 
    $\delta_1 < 1/4$, we have 
    %\blue
	\begin{equation}
	\label{eq: emp_1_even_diff_2}
		P\l(|\eda^{2T^{k+1}}(\text{1,II}) - \eda^{2T^{k+1}}(\text{1,IV})| < \delta_1 \Big| f_3^{k+1} \geq 1 - \delta_2\r) > 3/4 - \delta_1.
	\end{equation}
	%\black
	Form \eqref{eq: emp_1_even_sum_bound_2}, \eqref{eq: emp_1_even_diff_2} and Lemma~\ref{lem: sum_diff_tech}, we have
	%\blue
	\begin{align}
	\label{eq: eda_1_even_bound}
		&P\l( \eda^{2T^{k+1}}(\text{1,II}), \eda^{2T^{k+1}}(\text{1,IV}) \in (0.25 (1-f_4^{k+1}) - \delta_1, 0.25 (1-f_4^{k+1}) + \delta_1) \Big | f_3^{k+1} \geq 1 - \delta_2 \r) \nonumber \\
		&> 3/4 - \delta_1.
	\end{align}	
	%\black
	%\red 
	Note that $0.25 (1-f_4^{k+1}) - \delta_1$ could be negative. %\black
	From \eqref{eq: eda_0_odd_bound}, \eqref{eq: eda_0_even_bound}, \eqref{eq: eda_1_even_bound} and \eqref{eq: eda_1_odd_bound} we get \eqref{eq: cond_prob_second}, 
	and this completes the proof.
\end{proof}

We now consider 
%the two scenarios 
%\red 
two scenarios %\black
based on 
%\red 
whether %\black
$f_4^{k+1} < 1 - \delta_3$ or $f_4^{k+1} \geq 1 - \delta_3$.

\begin{lemma}
	\label{lem: prob_f3_small_f4_small}
	For any $k \geq k_0$, if $f_4^{k+1} < 1 - \delta_3$ and $\eda^{2T^{k+1}} \in [\hat \mu]_{\delta_2}$, then 
	%\red
	%the regret 
	$K_1^{2T^{k+1}}(\text{1,0}) > 0.5\delta_3 c_2$. 
	%\black
\end{lemma}
\begin{proof}
	If $f_4^{k+1} < 1 - \delta_3$, then $\eda^{2T^{k+1}} \in [\hat \mu]_{\delta_2}$ implies that $\eda_{-1}^{2T^{k+1}}(\cdot | \text{1}) \in [\sigma_{even}]_{{(2\delta_2)}/{(0.5\delta_3)}}$. 
	 Indeed, since $\eda_1^{2T^{k+1}}(\text{1}) \geq (1 - f_4^{k+1})/2 > 0.5\delta_3$, 
	 normalizing $\eda^{2T^{k+1}}(\text{1},\cdot)$ by $\eda_1^{2T^{k+1}}(\text{1})$, we get 
	 %\blue
	 \begin{equation}
	 \label{eq: emp_cond_1_odd_bound}
	 	\eda_{-1}^{2T^{k+1}}(\text{I}|\text{1}), \eda_{-1}^{2T^{k+1}}(\text{III}|\text{1}) \in [0,{2 \delta_2}/{\delta_3}),
	 \end{equation}
	 %\black
	 and 
	 %\blue
	 \begin{equation}
	 \label{eq: emp_cond_1_even_diff_bound}
	 	|\eda_{-1}^{2T^{k+1}}(\text{II}|\text{1}) - \eda_{-1}^{2T^{k+1}}(\text{IV}|\text{1})| < \frac{4\delta_2}{\delta_3}. 
	 \end{equation}
	 %\black
	 Since 
	 \[
	 	\eda_{-1}^{2T^{k+1}}(\text{I}|\text{1}) + \eda_{-1}^{2T^{k+1}}(\text{II}|\text{1}) + \eda_{-1}^{2T^{k+1}}(\text{III}|\text{1}) + \eda_{-1}^{2T^{k+1}}(\text{IV}|\text{1}) = 1,
	 \]
	 we have,
	 %\red
	 \begin{equation}
	 \label{eq: emp_cond_1_even_sum_bound}
	 	\eda_{-1}^{2T^{k+1}}(\text{II}|\text{1}) + \eda_{-1}^{2T^{k+1}}(\text{IV}|\text{1}) \in \l[1 - \frac{4\delta_2}{\delta_3}, 1\r].
	 \end{equation}
	 %\black
	 From 
	 %\eqref{eq: emp_cond_1_even_sum_bound}, 
	 %\eqref{eq: emp_cond_1_even_diff_bound} 
	 %\red
	 \eqref{eq: emp_cond_1_even_diff_bound},
	 \eqref{eq: emp_cond_1_even_sum_bound}
	 %\black
	 and Lemma~\ref{lem: sum_diff_tech}, we have
	 %\blue
	 \begin{equation}
	 	\label{eq: emp_cond_1_even_bound}
	 	\eda_{-1}^{2T^{k+1}}(\text{II}|\text{1}), \eda_{-1}^{2T^{k+1}}(\text{IV}|\text{1}) \in \l(\frac{1}{2}  - \frac{4\delta_2}{\delta_3},\frac{1}{2} + \frac{2\delta_2}{\delta_3}\r),
	 \end{equation}
	 %\black
	 and hence
	 %\red
	 $\eda_{-1}^{2T^{k+1}}(\cdot | \text{1}) \in [\sigma_{even}]_{{(4\delta_2)}/{\delta_3}}$. 
	 %\black
	 Then, from the assumption
	 %the assumption $\delta_2 = {\epsilon_2 \delta_3}/4$, 
	 \eqref{eq: del_2_def} we have  
	 $\eda_{-1}^{2T^{k+1}}(\cdot | \text{1}) \in [\sigma_{even}]_{\epsilon_2}$, and hence from \eqref{eq: eps_regret_2}
	 we have
	 %\red
	 \begin{equation}
	 	K_1^{2T^{k+1}}(\text{1,0}) = \eda_1^{2T^{k+1}}(\text{1}) \reg_1\l[\l\{\l(\eda_{-1}^{2T^{k+1}}(\cdot | \text{1}), x_1(\text{0},\cdot), x_1(\text{1}, \cdot)\r)\r\}\r] >  0.5\delta_3 c_2. 
	 \end{equation}
	 %\black
	 %\footnote{ \red The order of the terms in the regret expression in the middle of the preceding equation needed to be flipped. Please check. \black}
	 %$\frac{\delta_2}{0.25\delta_3 - 2\delta_2} < \epsilon_2$ and hence
	 %and since $\epsilon_2 << \epsilon_3$, 
	 %$\eda_{-1}^{2T^{k+1}}(\cdot | \text{1}) \in [\sigma_{even}]_{{\delta_2}/{(0.5\delta_3)}}$ implies that 
	 %Thus,
	 %\[
	%	P\l([K_1^{2T^{k+1}}(\text{1,0})]^+ \geq 0.5\delta_3 c_2 \bigg|\l\{f_3^{k+1} \geq 1 - \delta_2, f_4^{k+1} < 1 - \delta_3 \r\}\r) > 1 - 3\delta_1.
	%\]
\end{proof}

\begin{lemma}
\label{lem: prob_f3_small_f4_big}
	For any $k \geq k_0$, if $f_4^{k+1} \geq 1 -\delta_3$, $f_3^{k+1} \geq 1 - \delta_2$ and $\eda^{2T^{k+1}} \in [\hat \mu]_{\delta_2}$, then $K_1^{2T^{k+1}}(\text{0,1}) > (1 - \delta_3) c_3$. 
\end{lemma}
\begin{proof}
	If $f_4^{k+1} \geq 1 -\delta_3$ and $f_3^{k+1} \geq 1 - \delta_2$, then 
	$\eda^{2T^{k+1}} \in [\hat \mu]_{\delta_2}$ implies that
	%\bb
	\begin{equation}
	\label{eq: eda_include_unif}
		\eda_{-1}^{2T^{k+1}}(\cdot | \text{0}) \in [\sigma_{unif}]_{\frac{\delta_3/4 + \delta_2 + \delta_3/8 + \delta_2/8}{1 - \delta_3/2 - \delta_2/2}}.
	\end{equation}
	%This is 
	%because $\hat \mu(\text{0},\cdot) \in [\sigma_{unif}]_{\delta_3/4}$ and hence $\eda^{2T^{k+1}}(\text{0},\cdot) \in [\sigma_{unif}]_{\delta_3/4 + \delta_2}$. 
	%\blue
	To see this, note that $f_4^{k+1} \geq 1 -\delta_3$ and $\eda^{2T^{k+1}} \in [\hat \mu]_{\delta_2}$ imply that 
	$\eda^{2T^{k+1}}(\text{0},\cdot) \in [\sigma_{unif}]_{\delta_3/4 + \delta_2}$.
	%\black
	We have 
	%$\eda_1^{2T^{k+1}}(0) = 1 - f_3^{k+1}/2 - f_4^{k+1}/2 \geq 1 - \delta_3/2 - \delta_2/2$.
%	\red
	$\eda_1^{2T^{k+1}}(0) = f_3^{k+1}/2 + f_4^{k+1}/2 \geq 1 - \delta_3/2 - \delta_2/2$.
%	\black
	Let 
	%$\kappa := (f_3^{k+1}/2 + f_4^{k+1}/2)/4$. 
	%\red 
	$\kappa := (1 - \eda_1^{2 T^{k+1}}(0))/4$. %\black
	Thus 
	%$0 \leq \kappa \leq \delta_3/8 - \delta_2/8$.
	%\red
	$0 \leq \kappa \leq \delta_3/8 + \delta_2/8$.
	%\black
	Let $\sigma_{unif} - \kappa := (0.25 - \kappa, 0.25 - \kappa, 0.25 - \kappa, 0.25 - \kappa)$.
	%This implies that
	%\red 
	Then we have %\black
	\[
		\eda^{2T^{k+1}}(\text{0},\cdot) \in [0.25 - \kappa, 0.25 - \kappa, 0.25 - \kappa, 0.25 - \kappa]_{\delta_3/4 + \delta_2 + \kappa}.
	\]
	Normalizing $\sigma_{unif} - \kappa$ with $1 - 4\kappa = \eda_1^{2T^{k+1}}(0)$ gives us $\sigma_{unif}$. 
	As a result, normalizing $\eda^{2T^{k+1}}(\text{0},\cdot)$ with $\eda_1^{2T^{k+1}}(0)$ gives \eqref{eq: eda_include_unif}.
	 Then, from the assumptions $\epsilon_3 < 1, \epsilon_2 < 1, \delta_2 = \epsilon_2\delta_3/4$ and $\delta_3 = \epsilon_3/2$, we have
	 \[
	 	\frac{\delta_3/4 + \delta_2 + \delta_3/8 + \delta_2/8}{1 - \delta_3/2 - \delta_2/2} \leq \frac{\delta_3}{1 - \delta_3} \leq \epsilon_3.
	 \]
	 %\cob
	Thus, 
	 $\eda_{-1}^{2T^{k+1}}(\cdot | \text{0}) 
	 %\in [\sigma_{even}]_{\frac{\delta_2 + \delta_3 / 4}{1 - \delta_3/2 - \delta_2/2}}   
	 \in [\sigma_{unif}]_{\epsilon_3}$,
	 and hence 
	 %from \eqref{eq: eps_regret_1} 
	 %\red 
	 from \eqref{eq: eps_regret_3} %\black 
	 we have
	 %$K_1^{2T^{k+1}}(\text{0,1}) > (1 - \delta_3) c_3$.
	 %\red
	 \begin{align}
	 	K_1^{2T^{k+1}}(\text{0,1}) &= \eda_1^{2T^{k+1}}(\text{0}) \reg_1\l[\l\{\l(\eda_{-1}^{2T^{k+1}}(\cdot | \text{0}), x_1(\text{1},\cdot), x_1(\text{0}, \cdot)\r)\r\}\r] \\
	 	&>  (1 - \delta_3/2 - \delta_2/2) c_3 > (1 - \delta_3)c_3. 
	 \end{align}
	 %\black
	 %\footnote{ \red The order of the two terms in the regret in the middle of the preceding equation needed to be interchanged. Please check. \black}
	 %Thus,
	 %\[
	%	P\l([K_1^{2T^{k+1}}(\text{0,1})]^+ \geq (1 - \delta_3) c_3 \1\l\{f_3^{k+1} \geq 1 - \delta_2, f_4^{k+1} \geq 1 -\delta_3 \r\}\r) \geq 1 - 3\delta_1.
	%\]
\end{proof}

\begin{lemma}
\label{lem: prob_f3_small}
	For any $k \geq k_0$, if $P(f_3^{k+1} < 1 - \delta_2) \leq 1/4$, then
	\begin{equation}
		P\l(\bar K^k > \min\{0.5\delta_3c_2,(1-\delta_3)c_3\}\r) > \frac{3}{4}\l(\frac{1}{4} - 3 \delta_1\r),
	\end{equation}
	%\red 
	where $\bar K^k$ is defined in equation \eqref{eq:3regrets}. %\black 
\end{lemma}
\begin{proof}
	
	From lemmas \ref{lem: prob_f3_small_f4_small} and \ref{lem: prob_f3_small_f4_big} we obtain the following:
	if $f_3^{k+1} \geq 1 - \delta_2$ and $\eda^{2T^{k+1}} \in [\hat \mu]_{\delta_2}$, then 
	$\bar K^k > \min\{0.5\delta_3c_2,(1-\delta_3)c_3\}$.
	As a result, from Lemma~\ref{lem: emp_prob_2}, if $P(f_3^{k+1} < 1 - \delta_2) \leq 1/4$, then
	%\red
	\begin{align*}
		&P\l(\bar K^k > \min\{0.5\delta_3c_2,(1-\delta_3)c_3\}\r)\\
		 &\ge P\l(\bar K^k > \min\{0.5\delta_3c_2,(1-\delta_3)c_3\} | f_3^{k+1} \geq 1 - \delta_2 \r) P(f_3^{k+1} \geq 1 - \delta_2)\\
		 &\geq P\l( \eda^{2T^{k+1}} \in [\hat \mu]_{\delta_2} | f_3^{k+1} \geq 1 - \delta_2 \r) P(f_3^{k+1} \geq 1 - \delta_2)\\
		&> \frac{3}{4}\l(\frac{1}{4} - 3 \delta_1\r).
	\end{align*}
	%\black
\end{proof}

	\begin{proof}[Proof of Proposition~\ref{prop: nonconv_reg}]
		Take 
		\[
			\tilde \epsilon =  \min\l\{\delta_2c_1, 0.5\delta_3c_2, (1-\delta_3)c_3\r\}
		\]
		and
		\[
			\tilde \delta = \min\l\{\frac{1}{4}\l(\frac{1}{4} - \delta_1\r), \frac{3}{4}\l(\frac{1}{4} - 3 \delta_1\r)\r\}.
		\]
		From lemma \ref{lem: prob_f3_big} and \ref{lem: prob_f3_small} it follows that for all $k \geq k_0$,
	\[
		P\l( \bar K^k > \tilde \epsilon \r) >\tilde \delta_1,
	\]
		and this concludes the proof.
	\end{proof}

%!TEX root = ../main.tex

\section{Conclusion}
\label{sec: conclusion}

%One of the goals of game theory is to model real life situations via simple mathematical models that can be analyzed. It is unlikely that any mathematical model can be complex enough to account for all the characteristics of the real world and at the same time give useful insights into the problem. Such a trade-off has been observed time and again in scientific studies. However, even simple mathematical models can lead to insights that can be used in real world applications, for example in making policies. There has been a lot of study where agents are assumed to have EUT behavioral preferences. Many real world applications demand a more general modeling of people's behavior. Several researchers have expressed the importance of boundedly rational behavior in social and competitive interactions. This paper is an attempt to contribute in this vein where cumulative prospect theory is used to model boundedly rational behavior.

We studied how some of the results from the theory of learning in games are affected when the players in the game have 
%cumulative prospect theory preferences. 
%\red 
cumulative prospect theoretic preferences. %\black
For example, we saw that the notion of 
%\red
mediated CPT correlated equilibrium arising from
mediated games is more appropriate %\black
than the notion of CPT correlated equilibrium while studying the convergence of the empirical distribution 
%of play, 
%\color{red} 
of action play, 
%\color{black}
in particular for calibrated learning schemes. One can ask similar questions with respect to other learning schemes such as 
%\red 
{\em follow the perturbed leader} 
%\black 
\citep{fudenberg1995consistency}, 
%\red 
{\em fictitious play} %\black
\citep{brown1951iterative}, etc. We leave this for future work. In general, it seems that the results from 
%\red 
the theory of %\black
learning in games continue to hold under CPT with slight modifications. 
%\com
We also observed that the revelation principle does not hold under CPT.
%\cob
\appendix
\appendixpage
%\section{Proofs}
%\label{sec: app_proofs}
%\label{app: proofs}
%\input{tex/nonunique_bestreact}
%!TEX root = ../main.tex

%\red
\section{Notions of equilibrium}
\label{app: eq_notions}

In this appendix, we explore the relationship between the different notions
of equilibrium for a finite $n$-person normal form game 
$\Gamma$ with CPT players, organizing our observations 
into a sequence of remarks.
For convenience, we first briefly recall the four notions 
of equilibrium that played a role in the discussion in the paper.
A CPT correlated equilibrium of the game $\Gamma$, 
%as defined by Keskin, 
see Definition \ref{def: CPT_Nash_eq}, is 
%a probability 
%distribution on the set of action profiles, i.e. 
an element
of $\Delta(A)$. 
A CPT Nash equilibrium of the game $\Gamma$, 
%as defined by Keskin, 
see
Definition \ref{def: CPT_Nash_eq_real}, is 
%a product probability 
%distribution on the set of action profiles, i.e. 
an element
of $\Delta^*(A)$. 
Given a signal system $(B_i)_{i \in [n]}$ and a mediator
distribution $\psi \in \Delta(B)$, where
$B := \prod_{i=1}^n B_i$, a mediated CPT Nash equilibrium
of the mediated game $\tilde{\Gamma} := (\Gamma, (B_i)_{i \in [n]})$
with respect to the mediator distribution $\psi$,
see Definition \ref{def: mediated_cpt_nash_eq},
is a randomized strategy profile $\sigma = (\sigma_1, \ldots, \sigma_n)$, where $\sigma_i: B_i \to \Delta(A_i)$.
A mediated CPT correlated equilibrium of the game $\Gamma$,
see Definition \ref{def: mediated_cpt_corr_eq},
is an element of $\Delta(A)$.

\begin{remark}      \label{app: rem: allNash}
Let $\mu := \prod_{i=1}^n \mu_i \in \Delta^*(A)$
be a CPT Nash equilibrium of the game $\Gamma$. 
Then, for every signal system 
$(B_i)_{i \in [n]}$ and
mediator distribution $\psi \in \Delta(B)$, the 
randomized strategy profile 
$\sigma = (\sigma_1, \ldots, \sigma_n)$, where $\sigma_i: B_i \to \Delta(A_i)$ is the constant function given by
$\sigma_i (b_i) = \mu_i$ for all $b_i \in B_i$, is a 
mediated CPT Nash equilibrium
of the mediated game $\tilde{\Gamma} := (\Gamma, (B_i)_{i \in [n]})$
with respect to the mediator distribution $\psi$.
Conversely, if $\sigma$ is defined in terms of 
$\mu \in \Delta^*(A)$ as above
and $\sigma$ is a mediated CPT Nash equilibrium
of the mediated game $\tilde{\Gamma} := (\Gamma, (B_i)_{i \in [n]})$
with respect to the mediator distribution $\psi$, then 
$\mu$ is a CPT Nash equilibrium of the game $\Gamma$. 

To see this, note that for the strategy profile $\sigma$, for 
all $b_i \in B_i$, we have 
$\tilde{\mu}_{-i}(a_{-i}|b_i) = \prod_{j \neq i} \mu_j(a_j)$ for all $a_{-i} \in A_{-i}$, where $\tilde{\mu}_{-i}(a_{-i}|b_i)$ is as
defined in equation \eqref{eq: tilde_mu}. Hence $\sigma_i \in BR_i(\psi, \sigma)$, where $BR_i(\psi, \sigma)$ is as defined in equation \eqref{eq: best_response}, iff $\mu_i \in BR_i(\mu)$, where 
$BR_i(\mu)$ is as defined in equation \eqref{eq:bestresponseset}. 
This establishes the claim.
\end{remark}

\begin{remark}      \label{app: rem: CinD}
Every CPT correlated equilibrium of the game $\Gamma$ is a mediated CPT correlated equilibrium of the game $\Gamma$. Namely $C(\Gamma) \subset D(\Gamma)$.

To see this, let $\mu \in C(\Gamma)$. Consider the signal system
$(A_i)_{i \in [n]}$ (i.e. take $B_i = A_i$ for all $i \in [n]$)
with the mediator distribution $\mu$,
and consider the deterministic strategy profile 
$\sigma = (\sigma_1, \ldots, \sigma_n)$ given, with an 
abuse of notation, by $\sigma_i(b_i) = \1\{b_i = a_i\}$.
Note that $\eta(\psi,\sigma)$, as defined in 
equation \eqref{eq: eta_def}, equals $\mu$.
Since $\mu \in C(\Gamma)$, it verifies the condition 
in equation \eqref{eq: CPT_corr_ineq}, which 
then implies that $\sigma_i \in BR_i(\psi,\sigma)$, where 
$\psi = \mu$ and $BR_i(\psi, \sigma)$ is as defined in equation \eqref{eq: best_response}. This implies that $\mu \in D(\Gamma)$.
\end{remark}

\begin{remark}      \label{app: rem: prod_mediator}
Suppose the mediator distribution $\psi$ is of product form,
which we write as $\psi \in \Delta^*(B)$.
Let $\sigma = (\sigma_1, \ldots, \sigma_n)$ be 
a 
mediated CPT Nash equilibrium
of the mediated game $\tilde{\Gamma} := (\Gamma, (B_i)_{i \in [n]})$
with respect to the mediator distribution $\psi$.
Let $\mu := \eta(\psi,\sigma)$, 
as defined in 
equation \eqref{eq: eta_def}.
Note that we will have $\mu \in \Delta^*(A)$.
A simple calculation shows that 
$\tilde{\mu}_i(a_{-i}|b_i) = \prod_{j \neq i} \mu_j(a_j)$
for all $i \in [n]$, $b_i \in B_i$, and $a_{-i} \in A_{-i}$, where 
$\tilde{\mu}_{-i}(a_{-i}|b_i)$ is as
defined in equation \eqref{eq: tilde_mu}.
Thus $\sigma_i \in BR_i(\psi,\sigma)$ iff for all $b_i \in
\mbox{supp}(\psi_i)$ we have $\sigma_i(b_i) \in BR_i(\mu)$. 
This, in turn, is equivalent to $\mu_i \in BR_i(\mu)$. 
This characterizes the mediated CPT Nash equilibria of a mediated
game $\tilde{\Gamma} := (\Gamma, (B_i)_{i \in [n]})$
with respect to product form mediator distributions $\psi \in 
\Delta^*(B)$ in terms the CPT Nash equilibria of the game $\Gamma$.
\end{remark}

%\footnote{ \red It would be interesting to have a characterization of the intersection of the set of all mediated CPT correlated equilibria with $\Delta^*(A)$. It is unclear whether this is feasible or what form it will take because it is possible for a product form distribution to arise as a mixture of non-product form distributions. Nevertheless, it may be worth thinking about this a bit to see if something interesting might be lurking here. \black}

\begin{remark}
\label{rem: med_nadh_boundary}
\citet{nau2004geometry} showed that for any 
%\red 
finite 
%\black
n-person game the Nash equilibria all lie on the boundary of the set of correlated equilibria. 
\citet{phade2019geometry} extend this result to the CPT setting and show that all the CPT Nash equilibria lie on the boundary of the set of CPT correlated equilibria. 
It is natural to ask whether the CPT Nash equilibria in fact lie on the boundary of the set of all mediated CPT correlated equilibria. 
We know this is true for any $2 \times 2$ game $\Gamma$, since $C(\Gamma) = D(\Gamma)$ for such games. 
However, it is not known if this property holds in general for 
%any game,
%\red 
all finite $n$-person CPT games, 
%\black
and we leave this for future work.
%\footnote{\blue Should we move this remark to Appendix~\ref{app: eq_notions}?}
\end{remark}

\black

\section{Generalized signal spaces}
\label{app: polish_signal_space}

We now allow the
%signal sets $B_i$ to be infinite sets.
%In particular, suppose each set $B_i$ is a Polish space 
%\red 
signal set $B_i$ to be an arbitrary Polish space %\black
(a complete separable metric space) 
for all $i \in [n]$.
The product spaces $B := \prod_{i \in [n]} B_i$ and $B_{-i} := \prod_{j \neq i} B_j$, for all $i \in [n]$, are 
%\red 
then also %\black
Polish 
%\red 
spaces %\black
because a countable product of Polish spaces is a Polish space.
Let $\mcal{B}_i, \mcal{B}$ and $\mcal{B}_{-i}$ denote the $\sigma$-algebra of Borel sets on the spaces $B_i, B$ and $B_{-i}$ respectively.
Let the mediator be characterized by a probability 
%measure 
%\red 
distribution %\black
$\psi$ on $(B, \mcal{B})$.
Let $\psi_i$ denote the marginal probability distribution on $B_i$ induced by $\psi$.
Let $\psi_{-i} : B_i \times \mcal{B}_{-i} \to [0,1]$ be a function which satisfies:
\begin{enumerate}
 	\item $\psi_{-i}(b_i, \cdot)$ is a probability 
 	%\red 
 	distribution %\black
 	on $(B_{-i}, \mcal{B}_{-i})$, for all $b_i \in B_i$,
 	\item $\psi_{-i}(\cdot, X)$ is a measurable function on $(B_i, \mcal{B}_i)$, for all $X \in \mcal{B}_{-i}$,
 	\item for all $X \in \mcal{B}_{-i}$ and $Y \in \mcal{B}_i$,
 	\begin{equation}
 		\psi(Y \times X) = \int_{Y} \psi_{-i}(y, X) \psi_i(dy).	
 	\end{equation}
 \end{enumerate} 
%The function $\psi_{-i}$ is called a 
%\emph{product regular conditional probability} of the product measure space $(B_{-i} \times B_i, \mcal{B}_{-i} \times \mcal{B}_i, \psi)$ \citep{faden1985existence}, and its existence is guaranteed by the fact that $B_{-i}$ is a Radon space (every Polish space is a Radon space) \citep{leao2004regular}.
%\red 
The function $\psi_{-i}$ is called a 
\emph{regular conditional probability}. For a 
proof of its existence, see \cite[Theorem 1]{chang-pollard97}
(this theorem needs to be used in the framework of
\cite[Example 2]{chang-pollard97}). %\black

Let a randomized strategy for any player $i$ be given by a measurable function $\sigma_i: B_i \to \Delta(A_i)$ with respect to the Borel $\sigma$-algebra on $\Delta(A_i)$, and let $\sigma = (\sigma_1, \dots, \sigma_n)$ denote the 
%randomized profile strategy profile 
%\red 
randomized strategy profile %\black
as before.
Let $\sigma_{-i} := \prod_{j \neq i} \sigma_{j} : B_{-i} \to \Delta(A_{-i})$.
Let $\nu_{-i}(b_i)$ be the push forward probability distribution of $\psi_{-i}(b_i, \cdot)$ with respect to the function $\sigma_{-i}$, and let 
\begin{equation}
	\tilde \mu_{-i}(a_{-i} | b_i ) := \int_{\Delta(A_{-i})} p(a_{-i}) \nu_{-i}(b_i )(dp).
\end{equation}
Note that $\tilde \mu_{-i}(\cdot | b_i ) \in \Delta(A_{-i})$.
Let $\nu(\psi, \sigma)$ be the push forward probability distribution of $\psi$ with respect to the function $\sigma := \prod_{i \in [n]} \sigma_i : B \to \Delta(A)$,
and let
\begin{equation}
	\eta(\psi, \sigma)(a) := \int_{\Delta(A)} p(a) \nu(\psi, \sigma)(dp).
\end{equation}
Note that $\eta(\psi, \sigma) \in \Delta(A)$.

Let the best response set of player $i$ to a randomized strategy profile $\sigma$ and a mediator distribution $\psi$ be given by
\begin{align}
\label{eq: best_response_app}
	BR_i(\psi, \sigma) := \bigg\{ \sigma^*_i &: B_i \to \Delta(A_i) \text{ a measurable function} \bigg| \text{ for all } b_i \in 
    %B_i
    \supp(\psi_i)
    , \nonumber\\
	& \supp(\sigma^*_i(b_i)) \subset \arg \max_{a_i \in A_i} V_i\l(\l\{ \tilde \mu_{-i}(a_{-i}|b_i), x_i(a_i,a_{-i}) \r\}_{a_{-i} \in A_{-i}} \r) \bigg\},
\end{align}
where
$\supp(\psi_i)$ is the smallest closed set $Y \subset B_i$ with $\psi_i(B_i\back Y) = 0$.

%We can now define a mediated CPT Nash equilibrium and a mediated CPT correlated equilibrium verbatim as in Definitions \ref{def: mediated_cpt_nash_eq} and \ref{def: mediated_cpt_corr_eq} and denote these by $\Sigma^*(\Gamma, (B_i)_{i \in [n]}, \psi)$ and $D^*(\Gamma)$ respectively, as before. 
%\red
We can now define,
exactly as in Definition \ref{def: mediated_cpt_nash_eq},
the notion of a mediated CPT Nash equilibrium for the
mediated game $\tilde{\Gamma} := (\Gamma, (B_i)_{i \in [n]})$
with respect to a probability distribution $\psi$ on 
$(B, \mcal{B})$, where now $(B_i, \mcal{B}_i)_{i \in [n]}$ 
are arbitrary Polish spaces. Let $\Sigma^*(\Gamma, (B_i)_{i \in [n]}, \psi)$ denote the set of such mediated CPT Nash equilibria. 
We can also define, 
exactly as in Definition \ref{def: mediated_cpt_corr_eq},
the notion of a mediated CPT correlated equilibrium
(which is a probability distribution in $\Delta(A)$, as before) in this
extended setting where the signal spaces are allowed to be 
arbitrary Polish spaces. Let $D^*(\Gamma)$ denote the 
set of mediated CPT correlated equilibria in this extended sense.
%\black
Let $C(\Gamma, i, a_i)$ and $C(\Gamma, i)$ be defined as before.

\begin{lemma}
	For any game $\Gamma$, we have
	%\red
	\[
		D^*(\Gamma) \subset \cap_{i \in [n]} \conv (C(\Gamma, i)).
	\]
	%\black
\end{lemma}
\begin{proof}
Let 
%$\mu \in D(\Gamma)$. 
%\red 
$\mu \in D^*(\Gamma)$.%\black
Then there exists a signal system 
%$(B_i)_{i \in [n]}$, 
%\red 
comprised of Polish spaces $(B_i, \mcal{B}_i)_{i \in [n]}$, %\black
%a mediator distribution $\psi$, 
%\red 
a mediator distribution $\psi$ which is a probability 
distribution on $(B, \mcal{B})$, %\black
and a mediated CPT Nash equilibrium 
%$\sigma \in \Sigma(\Gamma, (B_i)_{i \in [n]}, \psi)$ 
%\red 
$\sigma \in \Sigma^*(\Gamma, (B_i)_{i \in [n]}, \psi)$ %\black
such that $\mu = \eta(\psi,\sigma)$. 
Fix $i \in [n]$.
For 
$b_i \in \supp(\psi_i)$ and
$a_i \in \supp \l(\sigma_i(b_i)\r)$, we have 
$\tilde \mu_{-i}(\cdot|b_i) \in C(\Gamma,i,a_i)$, 
from equations \eqref{eq: best_response_app} and \eqref{eq: best_reaction_ineq}. 
Let $a_i$ be such that $\mu_i(a_i) > 0$.
We have
\[
\mu_{-i}(\cdot|a_i) = \int_{B_i} \frac{\sigma_i(b_i)(a_i)}{\mu_i(a_i)} \tilde \mu_{-i}(\cdot | b_i) \psi_i(db_i).
\]
Also, since $\sigma$ is the product function $\prod_{i \in [n]} \sigma_i$ and $\mu$ is the push forward probability distribution of $\psi$ with respect to $\sigma$, we have that $\mu_i$ is the push forward probability distribution of $\psi_i$ with respect to the function $\sigma_i$, i.e.
\[
	\mu_{i}(a_i) = \int_{B_i} {\sigma_i(b_i)(a_i)} \psi_i(db_i).
\]
Since the set $\conv(C(\Gamma, i, a_i))$ is closed, we have 
$\mu_{-i}(\cdot|a_i) \in \conv \l(C(\Gamma,i,a_i)\r)$.
Since this holds for all $i \in [n]$, we have
$\mu = \eta(\psi,\sigma) \in \cap_{i \in [n]} \conv(C(\Gamma, i))$. This completes the proof.
\end{proof}

Since a finite set $B_i$ is a Polish space with respect to the discrete topology, we have $D(\Gamma) \subset D^*(\Gamma)$. 
From the above lemma and lemma~\ref{lem: mediated_corr_eq} we have $D^*(\Gamma) = D(\Gamma)$.
Hence, it is enough to restrict our attention to signals $B_i$ that are finite sets.
In fact, 
%\red 
it suffices to restrict attention to %\black
%signals 
%\red 
signal sets %\black
$B_i$ of size at most $|A|$ (see remark \ref{rem: rev_pure_strat}).

%!TEX root = ../main.tex

\section{Proof of Proposition~\ref{prop: generic_isolated}}
\label{app: generic}

\begin{proof}[Proof of Proposition~\ref{prop: generic_isolated}]
	%\blue
	For each of the players $i \in [n]$, let us fix the CPT features $r_i, v_i^{r_i}, w_i^\pm$
	such that $(v_i^{r_i})^{-1}$ is absolutely continuous.
	We also fix the action set $A_i$ for each of the players $i \in [n]$.
	%\black
	Since $n$ and $|A_i|, \forall i$ are finite, it is enough to show that for any fixed $i \in [n]$ and $a_{i} \in A_{i}$ the set of all games $\Gamma$ for which the set $C(\Gamma, i, a_i)$ has an isolated point is a null set.
	Since the set of all games for which any two payoffs of player $i$ are equal, i.e. $x_i(a) = x_i(\tilde a)$, $a \neq \tilde a$, is a null set, we can restrict our attention to games where all the payoffs for player $i$ 
	%\blue
	corresponding to her playing $a_i$ are distinct. 
	%\black
	Let $(\pi_i(1), \pi_i(2), \dots, \pi_i(|A_{-i}|))$ be a permutation of $A_{-i}$ such that
		\[
			x_i(a_i,\pi_i(1)) > x_i(a_i, \pi_i(2)) > \dots > x_i(a_i, \pi_i(|A_{-i}|)).
		\]

	Suppose we fix $x_j(a) \in \bbR$ for all $j \neq i$, and $x_i(\tilde a_i, a_{-i}) \in \bbR$ for all $\tilde a_i \neq a_i, a_{-i} \in A_{-i}$.
	Then the game $\Gamma$ is completely determined by the vector of payoffs $(x_i(a_i, a_{-i}))_{a_{-i} \in A_{-i}}$.
	% \in \bbR^{|A_{-i}|}$.
	%\red
	Let $S$ denote the set of all $(x_i(a_i, a_{-i}))_{a_{-i} \in A_{-i}}$
	for which the set $C(\Gamma, i, a_i)$ has isolated points.
	We will show that $S$
	is a null set with respect to the 
	%$|A_{-i}|$-dimensional 
	Lebesgue measure on $\bbR^{|A_{-i}|}$. %\black
	Then, by Tonelli's theorem,
	%{\color{blue} Tonelli's theorem actually. And we show that the set of interest is a subset of a zero measure set and since Lebesgue measure is complete our set is of zero measure.}
	 we have the required result.

%\red 
Recall that $Y_i \subset \bbR$ denotes the range of 
$v_i^{r_i}$ and that $Y_i$ is an open interval because
$v_i^{r_i}$ is assumed to be continuous and strictly increasing
on $\bbR$. 
Also recall that $\lambda_i^*$ is the measure on $Y_i$
that is the push forward of the
Lebesgue measure on $\bbR$ under $v_i^{r_i}$,
$\hat{\lambda}_i$ denotes Lebesgue measure restricted 
to $Y_i$, and 
%we have assumed 
that the assumption that $(v_i^{r_i})^{-1}$ 
is absolutely continuous implies
that $\lambda_i^*$ is 
absolutely continuous with respect to $\hat{\lambda}_i$. %\black
	 Consider the function 
	 %\bb
	 $f: \bbR^{|A_{-i}|} \to Y_i^{|A_{-i}|}$ 
	 %\cob
	 given by
	 \[
	 	f\l((x_i(a_i, a_{-i}))_{a_{-i} \in A_{-i}}\r) := (v_i^{r_i}(x_i(a_i, a_{-i}))_{a_{-i} \in A_{-i}}
	 \]
	 Let $y_i(a_{-i}) := v_i^{r_i}(x_i(a_i, a_{-i})) \in Y_i$ for all $a_{-i} \in A_{-i}$.
	 Since $v_i^{r_i}$ is strictly increasing, 
	 the mapping
	 $f$ is a bijection between 
	 %\blue
	 $(x_i(a_i, a_{-i}))_{a_{-i} \in A_{-i}} \in \bbR^{|A_{-i}|}$ and
	 $(y_i(a_{-i}))_{a_{-i} \in A_{-i}} \in Y_i^{|A_{-i}|}$.
	 %\black
	 Also, we have
	\[
		y_i(\pi_i(1)) > y_i(\pi_i(2)) > \dots > y_i(\pi_i(|A_{-i}|)).
	\]
	Suppose we could show that the set $f(S)$ is a null set with respect to the Lebesgue measure on $Y_i^{|A_{-i}|}$.
	Since the Lebesgue measure on $Y_i^{|A_{-i}|}$ is the completion of $(\hat \lambda_i)^{|A_{-i}|}$, 
	%\red 
	this would imply that 
	%\black 
	there exists a subset $S^*$ such that $f(S) \subset S^* \subset Y_i^{|A_{-i}|}$ and $(\hat \lambda_i)^{|A_{-i}|}(S^*) = 0$.
	Since $\lambda_i^* \ll \hat \lambda_i$, we have $(\lambda_i^*)^{|A_{-i}|} \ll (\hat \lambda_i)^{|A_{-i}|}$ and hence 
	%\red 
	we would have %\black
	$(\lambda_i^*)^{|A_{-i}|}(S^*) = 0$.
	Since $\lambda_i^*$ is the push forward of the Lebesgue measure $\lambda_i$
	%\red 
	under $v_i^{r_i}$, we 
	%get 
	%\red 
	would have %\black 
	$(\lambda_i)^{|A_{-i}|}(f^{-i}(S^*)) = 0$, and hence $S$ is a null set with respect to the Lebesgue measure on $\bbR^{|A_{-i}|}$.

	We will now show that the set $f(S)$ is a null set with respect to the Lebesgue measure on $Y_i^{|A_{-i}|}$.
	The vector $(y_i(a_{-i}))_{a_{-i} \in A_{-i}}$ is completely determined by choosing each of the following:
	\begin{enumerate}[(i)]
		\item a permutation $(\pi_i(1), \pi_i(2), \dots, \pi_i(|A_{-i}|))$ of $A_{-i}$,
	 	\item the differences $y_i(\pi_i(t)) - y_i(\pi_i(t+1)) > 0$ for all $1 \leq t < |A_{-i}|$,
		\item $y_i(\pi_i(|A_{-i}|)) \in Y_i$ such that
		%\blue
		\[
			y_i(\pi_i(1)) = y_i(\pi_i(|A_{-i}|)) + \sum_{t = 1}^{|A_{-i}| -1} y_i(\pi_i(t)) - y_i(\pi_i(t+1)) \in Y_i.
		\]
		%\black
	\end{enumerate}
     Further, we observe that the Lebesgue measure on $Y_i^{|A_{-i}|}$ is the 
     completion of the
     product measure of the following:
	 \begin{enumerate}[(1)]
		\item the uniform distribution on the set of permutations of $A_{-i}$,
	 	\item Lebesgue measure on $y_i(\pi_i(t)) - y_i(\pi_i(t+1)) > 0$ for all $1 \leq t < |A_{-i}|$,
	 	\item Lebesgue measure on $y_i(\pi_i(|A_{-i}|)) \in \bbR$,
	 restricted to 
	 %\blue
	 $y_i(\pi_i(|A_{-i}|))$ 
	 %\black
	 belonging to the interval such that $y_i(\pi_i(|A_{-i}|)) \in Y_i$ and 
	 %\blue
		\[
			y_i(\pi_i(1)) = y_i(\pi_i(|A_{-i}|)) + \sum_{t = 1}^{|A_{-i}|-1} \l[y_i(\pi_i(t)) - y_i(\pi_i(t+1))\r] \in Y_i.
		\]
		%\black
		\end{enumerate}
	
	We will now show that for any fixed 
	%permuation 
	%\red 
	permutation %\black
	$(\pi_i(1), \pi_i(2), \dots, \pi_i(|A_{-i}|))$ and any fixed positive differences $y_i(\pi_i(t)) - y_i(\pi_i(t+1)) > 0$ for all $1 \leq t < |A_{-i}|$, the set of all $y_i(\pi_i(|A_{-i}|))$ such that $(y_i(a_{-i}))_{a_{-i} \in A_{-i}} \in f(S)$ is a null set with respect to the one-dimensional Lebesgue measure.

	%\blue
	Let $(\underline \delta, \overline \delta)$ be the largest open interval such that if $y_i(\pi_i(|A_{-i}|)) = \delta$ for any $\delta \in (\underline \delta, \overline \delta)$, then $y_i(\pi_i(|A_{-i}|)), y_i(\pi_i(1)) \in Y_i$.
	Note that the interval $(\underline \delta, \overline \delta)$ could be empty depending on the fixed positive differences $y_i(\pi_i(t)) - y_i(\pi_i(t+1)) > 0$ for all $1 \leq t < |A_{-i}|$.
	%Let $\Gamma^\delta$ denote the game defined by letting $y_i(\pi_i(|A_{-i}|)) := \delta$ for $\delta \in (\underline \delta, \overline \delta)$.
	%\red 
	For $\delta \in (\underline \delta, \overline \delta)$,
	let $\Gamma^\delta$ denote the game defined by letting $y_i(\pi_i(|A_{-i}|)) := \delta$. 
	%\black
	In particular, for the game $\Gamma^\delta$, the payoffs corresponding to player $i$ and action $a_i$ are given by 
	%\blue
	\[
		x_i^\delta(a_i, a_{-i}) := (v_i^{r_i})^{-1}(y_i(a_{-i})),
	\]
	 for all $a_{-i} \in A_{-i}$,
	 where
	 \[
	 	y_i(a_i) = \delta + \sum_{t = \pi_i^{-1}(a_{-i})}^{|A_{-i}| - 1} \l[y_i(\pi_i(t)) - y_i(\pi_i(t+1))\r].
	 \]
	%For any game $\Gamma$ with payoffs $(x_i(a), i \in [n], a \in A)$, consider the function $F_i^{a_i} : \Delta(A_{-i}) \to \bbR$, given by
	%\[
	%	F_i^{a_i}(\mu_{-i}) := \max_{\tilde a_i \neq a_i} \reg_i[\{(\mu_{-i}(a_{-i}), x_i(\tilde a_i, a_{-i}), x_i(a_i, a_{-i}))\}_{a_{-i} \in A_{-i}}],
	%\]
	%where $\reg_i[\cdot]$ is as defined in equation~\eqref{eq: reg_def}.
	Consider the function $G_i^{a_i} : \Delta(A_{-i}) \times (\underline \delta, \overline \delta) \to \bbR$, given by
	\[
		G_i^{a_i}(\mu_{-i}, \delta) := \max_{\tilde a_i \neq a_i} \reg_i[\{(\mu_{-i}(a_{-i}), x_i(\tilde a_i, a_{-i}), x_i^\delta(a_i, a_{-i}))\}_{a_{-i} \in A_{-i}}],
	\]
	where 
	%\red
	the regret function
	%\black
	$\reg_i[\cdot]$ is as defined in equation~\eqref{eq: reg_def}.
	Since the probability weighting functions and the value function for player $i$ are assumed to be continuous, 
	the CPT value function $V_i(L)$ is continuous with respect to the probabilities and the outcomes in 
	%\red
	the lottery
	%\black
	$L$.
	Thus, the regret function $\reg_i[\cdot]$ is continuous
	%\red
	in its arguments, 
	%\black
	and hence we get that the function $G_i^{a_i}$ is continuous
	%\red
	in its arguments. 
	%\black
	%Thus we get that the function $F_i^{a_i}$ is continuous.

	%For any $\delta \in \bbR$ such that $y_i(\pi_i(1)) + \delta \in Y_i$ and $y_i(\pi_i(|A_{-i}|)) + \delta \in Y_i$, let $\Gamma^\delta$ be the game with all its payoffs except those corresponding to player $i$ and action $a_i$ be same as that of the game $\Gamma$; 
	%and the payoffs corresponding to player $i$ and action $a_i$ be given by 
	%\blue
	%\[
	%	x_i^\delta(a_i, a_{-i}) := (v_i^{r_i})^{-1}(y_i(a_{-i}) + \delta),
	%\]
	%\black
	% for all $a_{-i} \in A_{-i}$.  

	Now observe that, for any fixed $\delta \in (\underline \delta, \overline \delta)$, the outcomes 
	%$x_i^\delta(a_i, a_{-i}), \forall a_{-i} \in A_{-i},$ 
	%\red
	$(x_i^\delta(a_i, a_{-i}))_{a_{-i} \in A_{-i}}$
	%\black
	are divided into gains and losses depending on the reference point $r_i$.
	Hence, for some $0 \leq t_r \leq |A_{-i}|$, we have the outcomes $x_i^\delta(a_i, \pi_i(t)), \forall t \leq t_r,$ as gains, and the outcomes $x_i^\delta(a_i, \pi_i(t)), \forall t > t_r,$ as losses,
	%\red
	where $t_r =0$ corresponds to the case where all the outcomes
	$(x_i^\delta(a_i, a_{-i}))_{a_{-i} \in A_{-i}}$
	are losses.
	%\black
	As a result, the interval 
	$(\underline \delta, \overline \delta)$ can be partitioned into sub-intervals 
	$(\underline \delta, \delta_1), [\delta_1, \delta_2), \dots, [\delta_s, \overline \delta),$ 
	where
	$\underline \delta < \delta_1 < \delta_2 \dots < \delta_s < \overline \delta,$
	such that over any subinterval $I$ the outcomes are divided into gains and losses at the same point $t_r$.
	%\red
	Here $0 \le s \le |A_{-i}|$, with the case $s=0$ corresponding
	to the scenario where the division of the outcomes 
	$(x_i^\delta(a_i, a_{-i}))_{a_{-i} \in A_{-i}}$
	into gains and losses is the same throughout 
	$(\underline \delta, \overline \delta)$.
	%\black
	%Note that the partitioning depends on the positive differences $y_i(\pi_i(t)) - y_i(\pi_i(t+1)) > 0$ for all $1 \leq t < |A_{-i}|$,
	%and the 
	%\red
	Note that such an
	%\black
	interval $I$ could be open or half-open and half-closed. %depending on the partitioning and the particular subinterval under consideration.
	In the following argument it will not matter whether the subinterval is open or half-open and half closed.

	Let us now consider the function $G_i^{a_i}$ restricted to %$\Delta(|A_{-i}|) \times I$ 
	%\red
	$\Delta(A_{-i}) \times I$ 
	%\black
	for a fixed subinterval $I$. 
	Let $0 \leq t_r \leq |A_{-i}|$ be the point that divides the outcomes 
	%$x_i^\delta(a_i, \pi_i(t)), \forall 1 \leq t \leq |A_{-i}|,$
	%\red
	$(x_i^\delta(a_i, a_{-i}))_{a_{-i} \in A_{-i}}$
	%\black
	into gains and losses.
	Suppose we 
	%\red
	could
	%\black
	show that the set of $\delta \in I$ such that $(y_i(a_{-i}))_{a_{-i} \in A_{-i}} \in f(S)$ is a null set with respect to the one-dimensional Lebesgue measure.
	Since this 
	%is true 
	%\red
	would be true
	%\black
	for each of the subintervals $I$, and there are only finitely many such subintervals in the partitioning of $(\underline \delta, \overline \delta)$ above, we 
	%will 
	%\red
	would
	%\black
	get the 
	%required 
	%\red
	desired
	%\black
	result.

	We first prove the following useful property:
	%\begin{lemma}
		For any $\delta, \tilde \delta \in I$, and $\mu_{-i} \in \Delta(A_{-i})$, we have
		\begin{equation}
		\label{eq: G_linear}
			G_i^{a_i}(\mu_{-i}, \delta) - G_i^{a_i}(\mu_{-i}, \tilde \delta) = W_i(\mu_{-i}) (\tilde \delta - \delta),
		\end{equation}
		where
		\[
			W_i(\mu_{-i}) := w_i^+\l(\sum_{t = 1}^{t_r} \mu_{-i}(\pi_i(t))\r) + w_i^-\l(\sum_{t = t_r + 1}^{|A_{-i}|} \mu_{-i}(\pi_i(t))\r).
		\]
	%\end{lemma}
	%\begin{proof}
	%\red
		To see this, write
		\begin{align*}
			G_i^{a_i}(\mu_{-i}, \delta)  &=  \left( \max_{\tilde a_i \neq a_i} V_i(\{(\mu_{-i}(a_{-i}), x_i(\tilde a_i, a_{-i}))\}_{a_{-i} \in A_{-i}}) \right) \\
			&\quad -  V_i(\{(\mu_{-i}(a_{-i}), x_i^\delta(a_i, a_{-i}))\}_{a_{-i} \in A_{-i}}),
		\end{align*}
		which gives
		\begin{align*}
			G_i^{a_i}(\mu_{-i}, \delta) - G_i^{a_i}(\mu_{-i}, \tilde \delta) &= V_i(\{(\mu_{-i}(a_{-i}), x_i^{\tilde \delta}(a_i, a_{-i}))\}_{a_{-i} \in A_{-i}}) \\
			&\quad - V_i(\{(\mu_{-i}(a_{-i}), x_i^\delta(a_i, a_{-i}))\}_{a_{-i} \in A_{-i}}).
		\end{align*}
		Equation~\eqref{eq: G_linear} then follows from equation~\eqref{eq: CPT_value_cumulative}.
	%\end{proof}
	%\[
	%	\max_{\tilde a_i \neq a_i} \reg_i[\{(\mu_{-i}(a_{-i}), x_i^\delta(\tilde a_i, a_{-i}), x_i^\delta(a_i, a_{-i}))\}_{a_{-i} \in A_{-i}}] = F_i^{a_i}(\mu_{-i}) - \delta.
	%\]
	
        Note that $W_i(\mu_{-i}) > 0$ always.
		Indeed, since 
		$$\sum_{t = 1}^{t_r} \mu_{-i}(\pi_i(t)) + \sum_{t = t_r + 1}^{|A_{-i}|} \mu_{-i}(\pi_i(t)) = 1,$$ at least one of these two summations is positive, and $w_i^\pm(p) > 0$ for $p > 0$ from the assumptions on the probability weighting functions.
%\black

	%Thus, provided $\delta$ is such that $\Gamma^\delta$ is defined,
	For any $\delta \in I$, we have 
	$\mu_{-i} \in C(\Gamma^\delta, i , a_i)$ if and only if 
	%$\delta \geq F_i^{a_i}(\mu_{-i})$.
	$G_i^{a_i}(\mu_{-i}, \delta) \leq 0$.
	%\red
	%Suppose $\mu_{-i}$ is an isolated point in the set $C(\Gamma^\delta, i , a_i)$.
	%If $\delta < F_i^{a_i}(\mu_{-i})$
	 %then action $a_i$ cannot be a best response of player $i$ to the distribution $\mu_{-i}$.
	%However, since $\mu_{-i} \in C(\Gamma^\delta, i , a_i)$, we have $\delta \geq F_i^{a_i}(\mu_{-i})$.
	%We have $\mu_{-i} \in C(\Gamma^\delta, i, a_i)$ for all $\delta \geq F_i^{a_i}(\mu_{-i})$.
	%Thus $\delta \geq F_i^{a_i}(\mu_{-i})$.
	If 
	%$\delta > F_i^{a_i}(\mu_{-i})$ 
	$G_i^{a_i}(\mu_{-i}, \delta) < 0$
	then, by the continuity of the function 
	%$F_i^{a_i}$, 
	$G_i^{a_i}$,
	we will have a neighborhood around the point $\mu_{-i}$ that belongs to $C(\Gamma^\delta, i, a_i)$. 
	Since the domain $\Delta(A_{-i})$ itself does not have any isolated points, it prevents $\mu_{-i}$ from being an isolated point of $C(\Gamma^\delta, i , a_i)$.
	Thus, the fact that $\mu_{-i}$ is an isolated point of $C(\Gamma^\delta, i , a_i)$ implies that 
	%$\delta = F_i^{a_i}(\mu_{-i})$.
	$G_i^{a_i}(\mu_{-i}, \delta) = 0$.
	If $\mu_{-i}$ is not a strict local minimum of 
	%$F_i^{a_i}$, 
	$G_i^{a_i}(\cdot, \delta),$
	then there exists a sequence of points $(\mu_{-i}^t)_{t \geq 1}$ converging to $\mu_{-i}$ such that 
	%$F_i^{a_i}(\mu_{-i}^t) \leq \delta$ 
	$G_i^{a_i}(\mu_{-i}, \delta) \leq 0$,
	for all $t \geq 1$. 
	Then the sequence ($\mu_{-i}^t)_{t \geq 1}$ belongs to the set $C(\Gamma^\delta, i , a_i)$, contradicting the fact that $\mu_{-i}$ is an isolated point in the set $C(\Gamma^\delta, i , a_i)$. We have shown that if $\mu_{-i}$ is an isolated point in 
	the set $C(\Gamma^\delta, i , a_i)$, this implies that 
	%$\delta = F_i^{a_i}(\mu_{-i})$ 
	$G_i^{a_i}(\mu_{-i}, \delta) = 0$
	and that $\mu_{-i}$ is a strict local minimum of 
	%the function $F_i^{a_i}$.
	$G_i^{a_i}(\tilde \mu_{-i}, \delta)$ as a function of $\tilde \mu_{-i} \in \Delta(A_{-i})$.
	%\black

	% To see the other direction of the iff statement above, note that $\delta = F_i^{a_i}(\mu_{-i})$ implies $\mu_{-i} \in C(\Gamma^\delta, i, a_i)$. 
	% The fact that $\mu_{-i}$ is a strict local minima of the function $F_i^{a_i}$ implies that there exists a neighborhood of $\mu_{-i}$ in $\Delta(A_{-i})$ such that $F_i^{a_i}(\tilde \mu_{-i}) > \delta$ for all $\tilde \mu_{-i}$ in this neighborhood. 
	% This implies that the point $\mu_{-i}$ is an isolated point in the set $C(\Gamma^\delta, i , a_i)$. 
	% This completes the proof of the above iff statement.

	To complete the proof of the proposition, 
	it is enough to show that the set of all 
	%$\delta \in \bbR$, 
	$\delta \in I$
	%for which there exists a strict local minima $\mu_{-i}$ of $F_i^{a_i}$ such that $F_i^{a_i}(\mu_{-i}) = \delta$, 
	for which 
	%there exists a
	%\red
	there exists
	%\black
	$\mu_{-i} \in \Delta(A_{-i})$
	such that $G_i^{a_i}(\mu_{-i}, \delta) = 0$ and $\mu_{-i}$ is a strict local 
	%minima 
	%\red
	minimum
	%\black
	of $G_i^{a_i}(\cdot, \delta)$
	is a null set with respect to 
	%the one dimensional 
	%\red
	one dimensional 
	%\black
	Lebesgue measure. 
	Let $T \subset \Delta(A_{-i}) \times I$ be the set of all pairs $(\mu_{-i}, \delta)$ such that $G_i^{a_i}(\mu_{-i}, \delta) = 0$ and $\mu_{-i}$ is a strict local 
	%minima 
	%\red
	minimum
	%\black
	of $G_i^{a_i}(\cdot, \delta)$.
	We 
	%first 
	%\red
	will
	%\black
	prove that the set $T$ is countable.
	%We first prove that the set of all pairs $(\mu_{-i}, \delta)$ 
	%that are a strict local minima of the function $F_i^{a_i}$ is countable.
	%such that $G_i^{a_i}(\mu_{-i}, \delta) = 0$ and $\mu_{-i}$ is a strict local minima of $G_i^{a_i}(\cdot, \delta)$ is countable.
	To see this, for each 
	%strict local minima $\mu_{-i}'$ there exists a pair of vectors with rational elements, $(p^{\mu_{-i}'}(a_{-i}))_{a_{-i} \in A_{-i}}$ and $(q^{\mu_{-i}'}(a_{-i}))_{a_{-i} \in A_{-i}}$, such that
	pair $(\mu_{-i}, \delta) \in T$,
	%such that $G_i^{a_i}(\mu_{-i}, \delta) = 0$ and $\mu_{-i}$ is a strict local minima of $G_i^{a_i}(\cdot, \delta)$, 
	there exists a pair of vectors with rational elements, %$(p^{\mu_{-i}, \delta}(a_{-i}))_{a_{-i} \in A_{-i}}$ and $(q^{\mu_{-i}, \delta}(a_{-i}))_{a_{-i} \in A_{-i}}$, 
	%\red
	$(p_{\mu_{-i}, \delta}(a_{-i}))_{a_{-i} \in A_{-i}}$ and $(q_{\mu_{-i}, \delta}(a_{-i}))_{a_{-i} \in A_{-i}}$, 
	%\black
	such that
	%\[
	%	p^{\mu_{-i}, \delta}(a_{-i}) < \mu_{-i}(a_{-i}) < q^{\mu_{-i}, \delta}(a_{-i}), \text{ for all } a_{-i} \in A_{-i},
	%\] 
	%\red
	\[
		p_{\mu_{-i}, \delta}(a_{-i}) < \mu_{-i}(a_{-i}) < q_{\mu_{-i}, \delta}(a_{-i}), \text{ for all } a_{-i} \in A_{-i},
	\] 
	%\black
	 and for any $\tilde \mu_{-i} \in \Delta(A_{-i})$ such that
	 %\[
	%	p^{\mu_{-i}, \delta}(a_{-i}) < \tilde \mu_{-i}(a_{-i}) < q^{\mu_{-i}, \delta}(a_{-i}), \text{ for all } a_{-i} \in A_{-i},
	%\] 
	%\red
	\[
		p_{\mu_{-i}, \delta}(a_{-i}) < \tilde \mu_{-i}(a_{-i}) < q_{\mu_{-i}, \delta}(a_{-i}), \text{ for all } a_{-i} \in A_{-i},
	\] 
	%\black
	we have 
	%$F_i^{a_i}(\mu_{-i}) > F_i^{a_i}(\mu_{-i}')$.
	$G_i^{a_i}(\tilde \mu_{-i}, \delta) > G_i^{a_i}(\mu_{-i}, \delta)$.
	%We note that there cannot exist two distinct strict local minima $\mu_{-i}' \neq \mu_{-i}''$ such that $p^{\mu_{-i}'}(a_{-i}) = p^{\mu_{-i}''}(a_{-i})$ and $q^{\mu_{-i}'}(a_{-i}) = q^{\mu_{-i}''}(a_{-i})$ for all $a_{-i} \in A_{-i}$.
	%Suppone 
	%\red
	Suppose
	%\black
	there are two distinct pairs $(\mu_{-i}', \delta'), (\mu_{-i}'', \delta'') \in T$ 
	such that 
	%$G_i^{a_i}(\mu_{-i}', \delta') = 0$ and $\mu_{-i}'$ is a strict local minima of $G_i^{a_i}(\cdot, \delta)$, and $G_i^{a_i}(\mu_{-i}'', \delta'') = 0$ and $\mu_{-i}''$ is a strict local minima of $G_i^{a_i}(\cdot, \delta'')$, 
	%that have 
	%$p^{\mu_{-i}', \delta'}(a_{-i}) = p^{\mu_{-i}'', \delta''}(a_{-i}) =: p(a_{-i})$ and $q^{\mu_{-i}'}(a_{-i}) = q^{\mu_{-i}''}(a_{-i}) =: q(a_{-i})$ for all $a_{-i} \in A_{-i}$.
	%\red
	$p_{\mu_{-i}', \delta'}(a_{-i}) = p_{\mu_{-i}'', \delta''}(a_{-i}) =: p(a_{-i})$ and $q_{\mu_{-i}'}(a_{-i}) = q_{\mu_{-i}''}(a_{-i}) =: q(a_{-i})$ for all $a_{-i} \in A_{-i}$.
	%\black
	We note that 
	%\red
	in this case we must have
	%\black
	$\delta' \neq \delta''$.
	%from the way the vectors $(p(a_{-i}))_{a_{-i} \in A_{-i}}$ and $(q(a_{-i}))_{a_{-i} \in A_{-i}}$ are defined.
	Let $\delta' < \delta''$ without loss of generality.
	%We have $G_i(\tilde \mu_{-i}, \delta'') \geq 0$ for any $\tilde \mu_{-i} \in \Delta(A_{-i})$ such that
	 %\[
	%	p(a_{-i}) < \tilde \mu_{-i}(a_{-i}) < q(a_{-i}), \text{ for all } a_{-i} \in A_{-i}.
	%\] 
	%\red
	We have $G_i^{a_i}(\mu_{-i}', \delta'') \geq 0$ because
	 \[
		p(a_{-i}) < \mu_{-i}'(a_{-i}) < q(a_{-i}), \text{ for all } a_{-i} \in A_{-i}.
	\] 
	%\black
	From 
	%\red
	equation
	%\black
	\eqref{eq: G_linear}, we have 
	%$$G_i^{a_i}(\tilde \mu_{-i}, \delta') - G_i^{a_i}(\tilde \mu_{-i}, \delta'') = W_i(\tilde \mu_{-i})(\delta'' - \delta') > 0.$$
	%\red
	$$G_i^{a_i}(\mu_{-i}', \delta') - G_i^{a_i}(\mu_{-i}', \delta'') = W_i(\mu_{-i}')(\delta'' - \delta') > 0.$$
	%\black
	This implies $G_i^{a_i}(\mu_{-i}', \delta') > 0$ contradicting $(\mu_{-i}',\delta') \in T$.
	Thus we have an injective map from the set $T$
	%of all strict local minima of the function $F_i^{a_i}$ 
	to the set $\bbQ^{2|A_{-i}|}$.
	Hence the set $T$
	%of all strict local minima 
	is countable.
	Thus the set of all 
	%$\delta \in \bbR$, 
	$\delta \in I$,
	for which there exists a 
	%strict local minima 
	$\mu_{-i}$ 
	such that $(\mu_{-i}, \delta) \in T$
	%of $F_i^{a_i}$ such that $F_i^{a_i}(\mu_{-i}) = \delta$ 
	is also countable and hence a null set. 
	This completes the proof.
\end{proof}

%!TEX root = ../main.tex

\section{Proof of Lemma~\ref{lem: nu_close}}
\label{app: lemma-nu-close-proof}

\begin{proof}[Proof of Lemma~\ref{lem: nu_close}]
	We will first use the fact that player $2$ is randomizing over her actions I and III, independently at all the steps $(t^l_{odd})_{l \geq 1}$, and show that for sufficiently large $l$, $v_{odd}^l(\text{0},\text{I})$ and $v_{odd}^l(\text{0},\text{III})$ are almost equal with high probability. 
	To see this, observe that the sequence $\l(M_l, l \geq 1 \r)$ is a martingale, 
	where 
	\[
		M_l := l \times (\nu^l_{odd}(\text{0,I}) - \nu^l_{odd}(\text{0,III})).
	\]
	%Indeed, let $M_1^l := (M_1,\dots,M_l)$, and we have
	%\red 
	Indeed, letting $M_1^l := (M_1,\dots,M_l)$, we have %\black
	\begin{align*}
		&\bbE [M_{l+1} - M_{l} | M_1^l] = \bbE[M_{l+1} - M_l |  M_1^l, a^{t^{l+1}_{odd}}_1 = 0]P(a^{t^{l+1}_{odd}}_1 = 0 | M_1^l )\\
		& \hspace{4cm}+ \bbE[M_{l+1} - M_l |  M_1^l, a^{t^{l+1}_{odd}}_1 = 1]P(a^{t^{l+1}_{odd}}_1 = 1 | M_1^l)\\
		&= \bbE[\1 \{a^{t^{l+1}_{odd}} = (\text{0},\text{I})\} - \1 \{a^{t^{l+1}_{odd}} = (\text{0},\text{III})\}|M_1^l, a^{t^{l+1}_{odd}}_1 = 0]P(a^{t^{l+1}_{odd}}_1 = 0| M_1^l) + 0\\
		&= \frac{1}{2} - \frac{1}{2} = 0,
	\end{align*}
	where the last line follows from the fact that player $2$ plays $\sigma_{odd}$ at each of the steps $t^l_{odd}$ independently.
	Thus, for example by the Azuma-Hoeffding inequality, for any $\delta > 0$, there exists an integer $l^{(1)}_{\delta} > 1$, such that for all $l \geq l^{(1)}_{\delta}$,
	equation~\eqref{eq: mart_err_1} holds.
	Similarly, there exist integers $l^{(2)}_{\delta}, l^{(3)}_{\delta}, l^{(4)}_{\delta} > 1$, such that for all $l \geq l^{(2)}_{\delta}$,
	equation~\eqref{eq: mart_err_2} holds, 
	for all $l \geq l^{(3)}_{\delta}$,
	equation~\eqref{eq: mart_err_3} holds,
	and for all $l \geq l^{(4)}_{\delta}$,
    equation~\eqref{eq: mart_err_4}
    holds.
    This taking 
    \[
    	l_\delta := \max\{l^{(1)}_{\delta}, l^{(2)}_{\delta}, l^{(3)}_{\delta}, l^{(4)}_{\delta} \},
    \]
    we get the required result.
	\end{proof}
%\input{tex/proof_example}

%\clearpage
% Bibliography:
\bibliographystyle{abbrvnat} % Citation style
\bibliography{Bib_Database}

\end{document}